\documentclass[12pt,english]{article}
\usepackage[T1]{fontenc}
\usepackage[latin9]{inputenc}
\usepackage{geometry}
\usepackage{xcolor}
\usepackage{bm}
\usepackage{amsmath}
\usepackage{amsthm}
\usepackage{amssymb}
\usepackage{graphicx}
\usepackage[authoryear]{natbib}
\PassOptionsToPackage{normalem}{ulem}
\usepackage{ulem}

\makeatletter

\providecommand{\tabularnewline}{\\}

\theoremstyle{definition}
\newtheorem*{example*}{\protect\examplename}
\theoremstyle{plain}
\newtheorem{assumption}{\protect\assumptionname}
\theoremstyle{plain}
\newtheorem{lem}{\protect\lemmaname}
\theoremstyle{plain}
\newtheorem{thm}{\protect\theoremname}
\theoremstyle{plain}
\newtheorem{cor}{\protect\corollaryname}

\usepackage{bm}
\usepackage{hyperref}

\usepackage{xcolor}

\definecolor{mygray}{rgb}{0.6,0.6,0.6}

\newcommand{\mred}[1]{#1}

\newcommand{\norm}[1]{\left\vert\left\vert #1 \right\vert\right\vert_{2}}

\makeatletter
\newcommand*\bigcdot{\mathpalette\bigcdot@{.8}}
\newcommand*\bigcdot@[2]{\mathbin{\vcenter{\hbox{\scalebox{#2}{$\m@th#1\bullet$}}}}}
\makeatother
\newcommand{\bcd}{\bigcdot}

\makeatother

\usepackage{babel}
\providecommand{\assumptionname}{Assumption}
\providecommand{\corollaryname}{Corollary}
\providecommand{\examplename}{Example}
\providecommand{\lemmaname}{Lemma}
\providecommand{\theoremname}{Theorem}

\begin{document}
\title{Three-way Cross-Fitting and Pseudo-Outcome Regression for Estimation
of Conditional Effects and other Linear Functionals}
\author{Aaron Fisher \& Virginia Fisher}
\date{6/12/23}
\maketitle
\begin{abstract}
We propose an approach to better inform treatment decisions at an
individual level by adapting recent advances in average treatment
effect estimation to \emph{conditional} average treatment effect estimation.
Our work is based on doubly robust estimation methods, which combine
flexible machine learning tools to produce efficient effect estimates
while relaxing parametric assumptions about the data generating process.
Refinements to doubly robust methods have achieved faster convergence
by incorporating \emph{3-way cross-fitting}, which entails dividing
the sample into three partitions, using the first to estimate the
conditional probability of treatment, the second to estimate the conditional
expectation of the outcome, and the third to perform a first order
bias correction step.

Here, we combine the approaches of 3-way cross-fitting and pseudo-outcome
regression to produce personalized effect estimates. We show that
this approach yields fast convergence rates under a smoothness condition
on the conditional expectation of the outcome.
\end{abstract}

\section{Introduction}

Estimation of\emph{ }conditional effects quantifies the problem of
intervention decisions at the individual level. One of the most widely
studied estimation targets in this domain is the conditional average
treatment effect (CATE; e.g., \citealp{Robins1995-aw,Hill2011-zd,Imai2013-bz,Athey2016-vj,Richard_Hahn2017-hi,Powers2018-ya,Kunzel2019-ni}),
defined as the conditional mean of a contrast between potential outcomes
given a (possibly multi-dimensional) personalization variable. 

Formally, let $A\in\{0,1\}$ be an indicator of recieving treatment,
let $Y^{(1)}$ and $Y^{(0)}$ be the potential outcomes under treatment
and control respectively, so that $Y=AY^{(1)}+(1-A)Y^{(0)}$ is the
observed outcome. Let $X$ be a vector of confounders and effect modifiers;
let $C$ be a subvector of $X$ representing variables used as personalization
factors; and let $\mu(a,x)=\mathbb{E}\left(Y|X=x,A=a\right)$. Under
conventional identifiability assumptions (see Section \ref{sec:Estimands}),
the CATE can be expressed as $\mathbb{E}\left[Y^{(1)}-Y^{(0)}|C\right]=\mathbb{E}\left[\mu(1,X)-\mu(0,X)|C\right]$.
Our results for the CATE primarily use the fact that the CATE is a
conditional expectation of a function that is linear in $\mu$. For
this reason, our results also hold for other parameters with comparable
linearity properties (i.e., \emph{linear functionals, }see Section
\ref{sec:Estimands}), such as the conditional covariance, $Cov(A,Y|X)$.

\textcolor{black}{Many CATE estimation methods employ} the use of
so-called \emph{pseudo-outcomes} (also known as ``unbiased transformations''
or ``modified outcomes''), which are functions of the observed
data that act as stand-ins for latent, unobserved outcomes (e.g.,
$Y^{(1)}-Y^{(0)}$). By fitting a regression model against a pseudo-outcome,
we can mimic the idealized scenario in which a regression could be
fit directly to the latent outcome itself (\citealp{Rubin2005-pe,Van_der_Laan2006-os,Tian2014-fz,Chen2017-uv,Powers2018-ya,Foster2019-nb,Semenova2020-bw,Kennedy2022-da};
see also \citealt{Buckley1979-gg,Fan1994-vz,Rubin2007-ka,Diaz2018-ax}).
We define a pseudo-outcome as any variable that has the same conditional
expectation as the outcome of interest given covariates $X$. For
example, when estimating the CATE from randomized trial data with
binary treatments, \citet{Tian2014-fz} point out that $2Y(2A-1)$
functions as a pseudo-outcome (that is, $\mathbb{E}\left(2Y(2A-1)\,\,|\,\,X\right)=\mathbb{E}\left(Y^{(1)}-Y^{(0)}|X\right)$).
\textcolor{black}{{} Alternatively, when estimating the conditional
covariance $Cov\left(A,Y|X\right)$, the quantity $A(Y-\mathbb{E}\left(Y|X\right))$
functions as a pseudo-outcome (that is, $\mathbb{E}\left[A(Y-\mathbb{E}\left(Y|X\right))|X\right]=Cov\left(A,Y|X\right)$;
see, e.g., \citealp{Newey2018-da}).}

Recent work on the CATE takes inspiration from the rich literature
on double robustness and cross-fitting\emph{, }concepts that are popularly
applied to study the unconditional average treatment effect (ATE).
Both terms refer to algorithms that combine initial estimates of nuisance
functions (typically the propensity score, $\pi_{0}(X)=\mathbb{E}\left(A|X\right)$,
and outcome regression, $\mu_{0}(A,X)=\mathbb{E}\left(Y|A,X\right)$),
into a final estimate. The term\emph{ cross-fitting} (CF) refers to
procedures that split the data into two partitions; use the first
to estimate nuisance functions; and use the second to perform a bias
correction step (\citealp{Kennedy2020-ie,Kennedy2022-jv,Chernozhukov2022-jx,Chernozhukov2022-iq};
see also related work from, e.g., \citealt{Bickel1982-ht,Schick1986-qk,Bickel1988-mt},
as well as \citealt{Athey2016-vj}). 

Early use of the term \emph{doubly robust} (DR) referred to procedures
that depend on two nuisance function estimates, and that remain consistent
if at least one of the two nuisance models is correctly specified
(see, e.g., \citealt{Scharfstein1999-qw,Robins2000-xh,Bang2005-kd}).
However, this interpretation has become less common as researchers
increasingly move away from parametric assumptions in favor of flexible
machine learning tools. More recent usage of the term ``doubly robust''
emphasizes the fact that the same techniques can be used to combine
\emph{flexible} models of each nuisance function in such a way that
the final estimator has an error bound that depends on the two nuisance
errors only via their product. In particular, cross-fit DR estimates
for the ATE generally have bias on the order of \textcolor{black}{
\begin{equation}
\max_{a\in\{0,1\}}\mathbb{E}\left[\pi_{0}(X)\left\{ \frac{1}{\hat{\pi}(X)}-\frac{1}{\pi_{0}(X)}\right\} \left\{ \hat{\mu}(a,X)-\mu_{0}(a,X)\right\} |\hat{\pi},\hat{\mu}\right]\label{eq:cauchy-schwartz}
\end{equation}
}\citep{Chernozhukov2022-iq,Kennedy2022-jv,Kennedy2022-da}\textcolor{black}{.
This bias term is often described as }\textcolor{black}{\emph{second
order,}}\textcolor{black}{{} as it depends only on second order products
(see, e.g., \citealp{Kennedy2022-jv}). }

\textcolor{black}{Similar bias properties have recently been transported
to the CATE literature. For example, \citet{Kennedy2022-jv}'s DR-Learner
for the CATE attains a bias analogous to Eq (\ref{eq:cauchy-schwartz}).
\citet{Nie2020-ih} derive a CATE estimator that is efficient under
conditions strong enough to ensure that Eq (\ref{eq:cauchy-schwartz})
is $O_{\mathbb{P}}(n^{-1/2})$. \citet{Diaz2018-ax} derive a conditional
restricted mean survival time effect estimate with error terms analogous
to Eq (\ref{eq:cauchy-schwartz}).}

That said, \citet{Newey2018-da} point out that the ATE bias in Eq
\ref{eq:cauchy-schwartz} can be non-negligible when $\hat{\pi}$
and $\hat{\mu}$ are estimated non-parametrically (i.e., with increasing
flexibility) from the same sample partition. The authors show that,
if series estimators with an increasing number of basis functions
$k_{n}$ are used to estimate $\hat{\pi}$ and $\hat{\mu}$, then
the bias can be on the order of $k_{n}/n$. This is true not only
for estimates of the ATE, but also for estimates of any linear functional
(see Section \ref{sec:Estimands}). The authors suggest using two
separate subsamples to estimate $\pi_{0}$ and $\mu_{0}$, and a third
subsample to produce a final estimate. This eliminates the dependence
between $\hat{\pi}$ and $\hat{\mu}$ and greatly reduces the bias
term in Eq (\ref{eq:cauchy-schwartz}). We refer to this approach
as ``three-way'' cross-fitting, and attempt to develop intuition
for why it can be helpful in Section \ref{subsec:Intuition-for-3-way}.

Where \citeauthor{Newey2018-da} focus on estimands in the form of
expectations (e.g., the ATE), this paper studies linear functionals
in the form of \emph{conditional} expectations (e.g., the CATE). We
propose a combination 3-way cross-fitting and pseudo-outcome regression
to estimate conditional effects and other linear parameters. Here
too, 3-way CF can improve convergence rates relative to 2-way CF,
although the differences are less stark than in the unconditional
case. Our proposed approach and results are related to those of
the lp-R-learner \citep{Kennedy2022-da,Kennedy2022-cn}, although
there are several notable differences. Most importantly, our bound
is higher than the minimax rate achieved by \citet{Kennedy2022-cn}
in the $C=X$ setting, although it may provide a tighter bound in
the $C\neq X$ setting (see Section \ref{sec:Results}).

\subsection{Intuition for 3-way cross-fitting\label{subsec:Intuition-for-3-way}}

To illustrate the motivation for 3-way CF, we appeal to the intuition
in a (non-causal) prediction task where sample-splitting is well established:
estimating the conditional variance of an outcome $Y$ given covariates
$X$, i.e., $\mathbb{E}\left[\left\{ Y-\eta_{0}(X)\right\} ^{2}\right]$,
where $\eta_{0}(X)=\mathbb{E}\left(Y|X\right)$. A standard estimation
strategy is to fit a regression model $\hat{\eta}$ as an approximation
for $\eta_{0}$, and to estimate $\mathbb{E}\left[\left(Y-\eta_{0}(X)\right)^{2}\right]$
via 
\begin{equation}
\frac{1}{n}\sum_{i=1}^{n}\left[\left(Y_{i}-\hat{\eta}(X_{i})\right)^{2}\right].\label{eq:sum-var-est}
\end{equation}
Here, it is well known that bias will be incurred if the data for
the summation in Eq (\ref{eq:sum-var-est}) is the same as the data
used for training $\hat{\eta}$. This is because the true residuals
$Y-\eta_{0}(X)$ become \emph{positively correlated} with the model
errors $\hat{\eta}(X)-\eta_{0}(X)$, causing the observed residuals
$Y-\hat{\eta}(X)$ have an artificially low variance. More formally,
the bias of Eq (\ref{eq:sum-var-est}) becomes
\begin{align*}
\mathbb{E}\left[\left(Y-\hat{\eta}\right)^{2}\right]-\mathbb{E}\left[\left(Y-\eta_{0}\right)^{2}\right] & =\mathbb{E}\left[\left(Y-\eta_{0}+\eta_{0}-\hat{\eta}\right)^{2}\right]-\mathbb{E}\left[\left(Y-\eta_{0}\right)^{2}\right]\\
 & =-2\mathbb{E}\left[\left(Y-\eta_{0}\right)\left(\hat{\eta}-\eta_{0}\right)\right]+\mathbb{E}\left[\left(\eta_{0}-\hat{\eta}\right)^{2}\right],
\end{align*}
where the first term represents bias due to fitting $\hat{\eta}$
with the same data used to in Eq (\ref{eq:sum-var-est}). This term
becomes zero if $\hat{\eta}$ is estimated from a separate sample,
as in the 2-way CF workflow.

The above example shows how 2-way cross-fitting can reduce bias when
our workflow requires estimating a single nuisance function ($\eta_{0}$).
To extend this intuition to 3-way CF, consider what happens when we
must estimate \emph{two }nuisance functions. Specifically, consider
the task of estimating the expected conditional covariance $\mathbb{E}\left[\left(A-\pi_{0}(X)\right)\left(Y-\eta_{0}(X)\right)\right]$,
which is also used as a didactic example by \citet{Newey2018-da}.
Again, a standard approach is to fit models $\hat{\eta}$ and $\hat{\pi}$
for $\eta_{0}$ and $\pi_{0}$, and return the estimate 
\begin{equation}
\frac{1}{n}\sum_{i=1}^{n}\left(A-\hat{\pi}(X)\right)\left(Y-\hat{\eta}(X)\right).\label{eq:sum-cov-est}
\end{equation}
Here, even if $\hat{\pi}$ and $\hat{\eta}$ are learned from a dataset
separate from the one used in Eq (\ref{eq:sum-cov-est}), bias can
still occur as a result of correlation between $\hat{\pi}$ and $\hat{\eta}$.
For example, positive correlations between $\hat{\pi}$ and $\hat{\eta}$
can cause the observed residuals $A-\hat{\pi}$ and $Y-\hat{\eta}$
to have an artificially high correlation, creating a positive bias
for the estimate in Eq (\ref{eq:sum-cov-est}). 

Suppose we have an iid sample of size $n$, let $\mathbf{X}$ be a
matrix with $X_{i}$ in the $i^{th}$row, and $\hat{\bm{\pi}}$, $\hat{\bm{\eta}}$,
$\bm{\pi}_{0},$ $\bm{\eta}_{0}$, $\mathbf{a}$, $\mathbf{y}$ be
$n$-length vectors with $i^{th}$ elements equal to $\hat{\pi}(X_{i})$,
$\hat{\eta}(X_{i})$, $\pi_{0}(X_{i})$, $\eta_{0}(X_{i})$, $A_{i}$,
and $Y_{i}$ respectively. The bias of Eq (\ref{eq:sum-cov-est})
under 2-way CF is
\begin{align}
 & \frac{1}{n}\mathbb{E}\mathbb{E}\left[\left(\mathbf{a}-\hat{\bm{\pi}}\right)^{\top}\left(\mathbf{y}-\hat{\bm{\eta}}\right)|\mathbf{X}\right]-\mathbb{E}\left[Cov\left(A,Y|X\right)\right]\nonumber \\
 & \hspace{1em}\hspace{1em}\hspace{1em}\hspace{1em}\hspace{1em}=\frac{1}{n}\mathbb{E}\left[\left(\bm{\pi}_{0}-\mathbb{E}\left(\hat{\bm{\pi}}|\mathbf{X}\right)\right)^{\top}\left(\bm{\eta}_{0}-\mathbb{E}\left(\hat{\bm{\eta}}|\mathbf{X}\right)\right)\right]\label{eq:misspec-part}\\
 & \hspace{1em}\hspace{1em}\hspace{1em}\hspace{1em}\hspace{1em}\hspace{1em}+\frac{1}{n}\text{tr}\mathbb{E}\left[Cov\left(\mathbf{a}-\hat{\bm{\pi}},\mathbf{y}-\hat{\bm{\eta}}|\mathbf{X}\right)\right]-\mathbb{E}\left[Cov\left(A,Y|X\right)\right],\label{eq:cov-part}
\end{align}
where Line (\ref{eq:misspec-part}) represents bias from model misspecification.
Applying the fact that $\hat{\bm{\eta}},\hat{\bm{\pi}}\perp\mathbf{a},\mathbf{y|}\mathbf{X}$
under 2-way CF, Line (\ref{eq:cov-part}) equals
\begin{equation}
\frac{1}{n}\text{tr}\mathbb{E}\left[Cov\left(\hat{\bm{\pi}},\hat{\bm{\eta}}|\mathbf{X}\right)\right].\label{eq:cov-reduced}
\end{equation}
That is, the bias of Eq (\ref{eq:sum-cov-est}) depends on a model
misspecification term (Line (\ref{eq:misspec-part})) and the dependence
in the regression estimates $\hat{\bm{\pi}}$ and $\hat{\bm{\eta}}$
(Line (\ref{eq:cov-reduced})). The bias produced by Lines (\ref{eq:misspec-part})
and (\ref{eq:cov-reduced}) is \emph{second order} in the same sense
as Eq (\ref{eq:cauchy-schwartz}).

While we can see that Line (\ref{eq:cov-reduced}) becomes zero if
we adopt 3-way CF, it remains to show how severe this added bias can
be under 2-way CF, or even under no cross-fitting at all. Here, we
study the asymptotic properties in the special case of series estimators
using a $k_{n}$-dimensional set of basis functions, where $k_{n}$
grows with the sample size. If there exist constants $l$ and $u$
such that $0<l\leq Cov(A,Y|X)\leq u$, then Line (\ref{eq:cov-part})
has magnitude on the order $k_{n}/n$ when no cross-fitting is done
(Appendix \ref{sec:Conditional-covariance-results}). Under 2-way
CF, Line (\ref{eq:cov-part}) \emph{still }has magnitude on the order
of \emph{at least} $k_{n}/n$ (see \citealt{Newey2018-da}, as well
as Appendix \ref{sec:Conditional-covariance-results}). That is, for
the expected conditional covariance, 2-way CF does not diminish the
order of magnitude of Line (\ref{eq:cov-part}), while 3-way CF does.

The remainder of this paper expands the above argument in several
ways. Section \ref{sec:notation} introduces general notation. As
in \citet{Newey2018-da}, Section \ref{sec:Estimands} introduces
a broader set of estimands known as \emph{linear functionals}, which
includes the expected conditional covariance and ATE examples. We
also consider \emph{conditional} versions of these estimands that
can be used to better inform individual treatment decisions (Section
\ref{subsec:Defining-linear-functionals}). Section \ref{sec:Proposed-estimation-procedure}
proposes a spline-based estimator, and Section \ref{sec:Main-Assumptions}
presents our main assumptions. Section \ref{sec:Results} presents
our results, which pertain to the overall error of our estimator rather
than just its bias.

\section{General notation\label{sec:notation}}

Let $Z=(X,A,Y)$ denote a random vector of covariates $X$, exposures
$A$ and outcomes $Y\in\mathcal{Y}$. Let $C$ be a sub-vector of
$X$, upon which we would like to personalize our estimands. Let $d_{X}$
and $d_{C}$ respectively denote the dimensions of $X$ and $C$.

We consider the setting where an analyst has access to (up to) three
training datasets. We denote these datasets by $\tilde{\mathbf{Z}}=(\tilde{\mathbf{X}},\mathbf{\tilde{a}},\mathbf{\tilde{y}})$,
$\hat{\mathbf{Z}}=(\hat{\mathbf{X}},\hat{\mathbf{a}},\hat{\mathbf{y}})$,
and $\bar{\mathbf{Z}}=(\bar{\mathbf{X}},\bar{\mathbf{a}},\bar{\mathbf{y}})$,
where the first two are used to estimate nuisance functions and the
last is used to estimate a target functional. 

Going forward, when defining estimators, we generally include the
accents (e.g., ``hats'' or ``bars'') of all of the datasets that
contribute to the estimator. For example, the function $\hat{\tilde{f}}$
(defined in Section \ref{subsec:nuisance-estimates}, below) depends
on data from $\hat{\mathbf{Z}}$ and $\tilde{\mathbf{Z}}$. Figure
\ref{fig:Flowchart-showing-how} shows a flowchart of how each dataset
contributes to subsequent parameter estimates.

For simplicity, we assume that all three datasets are of size $n$.
In the case of 2-way CF, $\mathbf{\tilde{Z}}=\hat{\mathbf{Z}}\perp\bar{\mathbf{Z}}$.
In the case of 3-way CF, all three datasets are $iid$. We use non-bold
notation to refer to elements of each sample. For example, $\bar{Z}_{i}=(\bar{X}_{i},\bar{A}_{i},\bar{Y}_{i})$
denotes a row from $\bar{\mathbf{Z}}$, where we often omit $i$ for
brevity. Similarly, we sometimes omit the function arguments, e.g.,
abbreviating $\mathbb{E}\left[f(Z)\right]$ as $\mathbb{E}\left[f\right]$.

Next we introduce notation to describe convergence rates. From random
variables $A_{n},B_{n}$, let $A_{n}\lesssim B_{n}$ denote that there
exists a constant $c$ such that $A_{n}\leq cB_{n}$ for all $n$.
Let $A_{n}\asymp B_{n}$ denote that $A_{n}\lesssim B_{n}$ and $B_{n}\lesssim A_{n}$.\textcolor{black}{{}
Let $A_{n}\lesssim_{\mathbb{P}}c_{n}$ denote that $A_{n}=O_{\mathbb{P}}(c_{n})$
for constants $c_{n}$.}\textcolor{teal}{}

We say that a function $f$ is $s$-smooth if there exists a constant
$k$ such that $|f(x)-f_{s,x'}(x)|\leq k||x-x'||^{s}$ for all $x,x'$,
where $f_{s,x'}$ is the $\lfloor s\rfloor^{th}$ order Taylor approximation
of $f$ at $x'$. This form of smoothness is a key property of functions
in a H\"{o}lder class (see, e.g., \citealp{Tsybakov2009-yb,Kennedy2022-da}).

Finally, we introduce notation for expectations and probabilities.
For a (possibly random) function $f$, let $\mathbb{P}(f(Z))=\int f(z)d\mathbb{P}z$
be shorthand for $\mathbb{E}[f(Z)|f]$, and let $||f||_{2}^{2}=\mathbb{P}(f(Z)^{2})$.
That is, $\mathbb{P}$ marginalizes only over $Z$, while $\mathbb{E}[f(Z)]=\mathbb{E}\left[\mathbb{P}\left(f(Z)\right)\right]$
takes expectations over both $f$ and $Z$. We use $\text{Pr}(A)$
to denote the probability of event $A$, in order to disambiguate
from $\mathbb{P}$.

\section{\label{sec:Estimands}Estimands}

Throughout this paper we consider two estimation targets: the expected
conditional covariance and the conditional average treatment effect.
Section \ref{subsec:Examples} introduces these examples in more detail.
For simplicity of presentation, we begin with the unconditional versions
of these estimands. 

Section \ref{subsec:Defining-linear-functionals} introduces a class
of estimands known as \emph{linear functionals}, which includes our
examples. Section \ref{subsec:Defining-linear-functionals} also
describes several convenient properties of linear functionals that
facilitate their study. This section also introduces \emph{conditional}
versions of each estimand, tailored to a randomly selected covariate
profile $\ddot{C}$.

\subsection{Examples\label{subsec:Examples}}
\begin{example*}
\textbf{(Conditional Covariance)} Here, our estimation target in this
example is
\[
\theta_{0,\text{cov}}:=\mathbb{E}\left[\left\{ A-\mathbb{E}\left(A|X\right)\right\} \left\{ Y-\mathbb{E}\left(Y|X\right)\right\} \right].
\]
This quantity is relevant, for example, when studying the variance-weighted
treatment effect \citep{Newey2018-da,Kennedy2022-jv}\emph{\uline{}}
\end{example*}
\begin{example*}
\textbf{(Treatment Effects)} Let $A\in\{0,1\}$, and let $Y^{(a)}$
for $a\in\{0,1\}$ be the potential outcome under treatment $a$,
so that $Y^{(1)}-Y^{(0)}$ is the latent treatment effect for an individual.
Let $Y=AY^{(1)}+(1-A)Y^{(0)}$ be the observed outcome. Under standard
causal identification assumptions of conditional exchangeability ($Y^{(0)},Y^{(1)}\perp A|X$)
and positivity ($0<c\leq\text{Pr}(A=1|X)\leq1-c$ for some constant
$c$), the average outcome if all subjects received treatment is
\begin{align*}
\theta_{0,\text{trt}}:=\mathbb{E}\left(Y^{(1)}\right) & =\mathbb{E}\left[\mathbb{E}\left(Y^{(1)}|X\right)\right]\\
 & =\mathbb{E}\left[\mathbb{E}\left(Y^{(1)}|A=1,X\right)\right]\\
 & =\mathbb{E}\left[\mathbb{E}\left(Y|A=1,X\right)\right],
\end{align*}
and, similarly, the average outcome if all subjects received control
is
\[
\theta_{0,\text{ctrl}}:=\mathbb{E}\left(Y^{(0)}\right)=\mathbb{E}\left[\mathbb{E}\left(Y|A=0,X\right)\right].
\]
Combining these quantities, the average treatment effect is
\begin{align*}
\theta_{0,\text{TE}}:=\theta_{0,\text{trt}}-\theta_{0,\text{ctrl}}.
\end{align*}
In most of the sections below we estimate $\theta_{0,\text{trt}}$
and $\theta_{0,\text{ctrl}}$ separately. We return to the task of
estimating their difference in Corollary \ref{cor:(Conditional-Average-Treatment}.
\end{example*}

\subsection{Linear functionals\label{subsec:Defining-linear-functionals}}

Much previous work focuses on estimating \emph{linear functionals
}\citep{Newey2018-da,Chernozhukov2022-jx,Chernozhukov2022-iq}\emph{,
}which include all of the examples in Section \ref{subsec:Examples}.
We express a version of this property below, with respect to a generic
estimand of interest $\theta_{0}$.
\begin{assumption}
\emph{\label{assu:lin-func}} (Linear Functionals, adapted from \citealp{Newey2018-da})
There exists a random variable $J\in\{0,1\}$ and function $m$ such
that the estimand $\theta_{0}$ can be expressed in the form
\begin{equation}
\theta_{0}=\mathbb{E}\left[m(Z,\gamma_{0})\right],\label{eq:m-form}
\end{equation}
where $\gamma_{0}(x)=\mathbb{E}\left[Y|X=x,J=1\right]$. Let $\mathcal{G}$
be the set of all mappings from $\mathcal{X}\rightarrow\mathbb{R}$,
so that $\gamma_{0}\in\mathcal{G}$. Additionally, assume that (1)
$m$ is known, and linear in $\gamma_{0}$; (2) $J$ is known, and
is a deterministic function of $A$; and (3) there exists an unknown
function $\alpha_{0}:\mathcal{X}\rightarrow\mathbb{R}$ such that
$\mathbb{E}\left[J\alpha_{0}(X)^{2}\right]<\infty$ and 
\begin{equation}
\mathbb{E}\left[m(Z,\gamma)-m(Z,\gamma^{\text{zero}})\right]=\mathbb{E}\left[\alpha_{0}(X)J\gamma(X)\right]\label{eq:linear-form}
\end{equation}
for any $\gamma\in\mathcal{G}$, where $\gamma^{\text{zero}}$ is
the ``zero function'' satisfying $\gamma^{\text{zero}}(x)=0$ for
all $x\in\mathcal{X}$.
\end{assumption}
Eq (\ref{eq:linear-form}) has two important use cases. First, it
provides a convenient characterization of the bias of ``plug-in''
estimates $\frac{1}{n}\sum_{i=1}^{n}m(Z_{i},\hat{\gamma})$ of $\theta_{0}$,
where $\hat{\gamma}$ is a predetermined estimate of $\gamma_{0}$.
Specifically, for any predetermined function $\gamma\in\mathcal{G}$,
we have
\begin{equation}
\mathbb{E}\left[m(Z,\gamma)-m(Z,\gamma_{0})\right]=\mathbb{E}\left[\alpha_{0}(X)J\left\{ \gamma(X)-Y\right\} \right],\label{eq:debias}
\end{equation}
where the right-hand side of Eq (\ref{eq:debias}) is analogous to
the debiasing term in augmented inverse probability of treatment weighted
estimators. Second, Eq (\ref{eq:linear-form}) provides a means of
estimating $\alpha_{0}$ itself. By predefining a variety of functions
$\gamma\in\mathcal{G}$, analysts can use Eq (\ref{eq:linear-form})
to create moment conditions in which the only unknown quantity is
$\alpha_{0}(X)$ (\citealp{Newey2018-da,Chernozhukov2022-jx,Chernozhukov2022-iq};
see also \citealp{Robins2007-gi}). We review estimation of $\alpha_{0}$
in more detail in Section \ref{sec:Proposed-estimation-procedure}.

Table \ref{tab:Estimands-as-linear} includes the definitions of $J$
and $\alpha_{0}$ that satisfy Assumption \ref{assu:lin-func} for
$\theta_{0,\text{cov}}$, $\theta_{0,\text{trt}}$ and $\theta_{0,\text{ctrl}}$.
We illustrate how these definitions relate to Assumption \ref{assu:lin-func}
below.
\begin{example*}
\textbf{(Conditional Covariance)} For $\theta_{0,\text{cov}}$, Eq
(\ref{eq:linear-form}) becomes
\begin{align*}
 & \mathbb{E}\left[m(Z,\gamma)-m(Z,\gamma^{\text{zero}})\right]\\
 & \hspace{1em}=\mathbb{E}\left[A(Y-\gamma)-A(Y-0)\right]=\mathbb{E}\left[-A\gamma\right]=\mathbb{E}\left[-\mathbb{E}\left[A|X\right]\gamma\right].
\end{align*}
\[
\mathbb{E}\left[m(Z,\gamma)-m(Z,\gamma^{\text{zero}})|X\right]=\mathbb{E}\left[-A\gamma(X)|X\right]=\mathbb{E}\left[-A|X\right]\gamma(X)
\]
\end{example*}
\begin{example*}
\textbf{(Treatment Effects)} For $\theta_{0,\text{trt}}$, Eq (\ref{eq:linear-form})
becomes 
\begin{align}
 & \mathbb{E}\left[m(Z,\gamma)-m(Z,\gamma^{\text{zero}})\right]\nonumber \\
 & \hspace{1em}=\mathbb{E}\left[\gamma(X)\right]=\mathbb{E}\left[\gamma(X)\frac{\mathbb{E}\left[A|X\right]}{\mathbb{E}\left[A|X\right]}\right]=\mathbb{E}\left[\gamma(X)A\alpha_{0}\left(X\right)\right],\label{eq:moment-miss}
\end{align}
where $\alpha_{0}(X)=1/\mathbb{E}\left[A|X\right]$ equals the well-known
inverse probability of treatment weight. A similar derivation holds
for $\theta_{0,\text{ctrl}}$.
\end{example*}
\begin{table}
\begin{centering}
\begin{tabular}{|c|c|c|c|c|c|}
\hline 
Estimand  & $J$ & $\gamma_{0}(X)$ & $m(Z,\gamma)$ & $\alpha_{0}(X)$ & $v_{q}(Z)$\tabularnewline
label &  & (i.e., $\mathbb{E}\left[Y|X,J=1\right]$) &  &  & \tabularnewline
\hline 
\hline 
$\theta_{0,\text{cov}}$ & $1$ & $\mathbb{E}\left[Y|X\right]$ & $A(Y-\gamma(X))$ & $-\mathbb{E}\left[A|X\right]$ & $-Aq(X)$\tabularnewline
\hline 
$\theta_{0,\text{trt}}$ & $A$ & $\mathbb{E}\left[Y|X,A=1\right]$ & $\gamma(X)$ & $1/\mathbb{E}\left[A|X\right]$ & $q(X)$\tabularnewline
\hline 
$\theta_{0,\text{ctrl}}$ & $1-A$ & $\mathbb{E}\left[Y|X,A=0\right]$ & $\gamma(X)$ & $1/\mathbb{E}\left[1-A|X\right]$ & $q(X)$\tabularnewline
\hline 
\end{tabular}
\par\end{centering}
\caption{Estimands as linear functionals\label{tab:Estimands-as-linear}}
\end{table}

\textcolor{black}{As noted above, we are interested not only in estimands
that can be expressed as expectations, as in Eq (\ref{eq:m-form}),
but also in }\textcolor{black}{\emph{conditional}}\textcolor{black}{{}
estimands in the form 
\begin{equation}
\ddot{\theta}_{0}:=\mathbb{E}\left[m(Z,\gamma_{0})|C=\ddot{C}\right].\label{eq:conditional-m-form}
\end{equation}
}Let $\ddot{\theta}_{0,\text{cov}}$, $\ddot{\theta}_{0,\text{trt}}$
and $\ddot{\theta}_{0,\text{ctrl}}$ represent the conditional versions
of the estimands on Table \ref{tab:Estimands-as-linear}.

\textcolor{black}{The next section considers estimation strategies
for conditional estimands in the form of Eq (\ref{eq:conditional-m-form})}

\section{Proposed estimation procedure\label{sec:Proposed-estimation-procedure}}

\textcolor{black}{Under Assumption \ref{assu:lin-func}, the conditional
estimand in Eq (\ref{eq:conditional-m-form}) can equivalently be
expressed as
\begin{equation}
\ddot{\theta}_{0}=\mathbb{E}\left[f_{0}(Z)|C=\ddot{C}\right],\hspace{1em}\text{where}\hspace{1em}f_{0}(z):=m(z,\gamma_{0})+\alpha_{0}(X)J\left\{ \gamma_{0}(X)-Y\right\} .\label{eq:f0def}
\end{equation}
The equivalence comes from iterating expectations over $X$ to see
that the added product term in Eq (\ref{eq:f0def}) has mean zero.}\textcolor{teal}{}\textcolor{black}{\emph{
}}\textcolor{black}{We choose $f_{0}$ as our }\textcolor{black}{\emph{pseudo-outcome
}}\textcolor{black}{function due to its established relationship with
doubly robust estimation (see, e.g., \citealp{Newey2018-da,Chernozhukov2022-iq,Kennedy2022-da}).
}The remainder of this section outlines methods first for estimating
$f_{0}$ (i.e., estimating the nuisance functions $\gamma_{0}$ and
$\alpha_{0}$), and then for fitting a pseudo-outcome regression model.

We present our approach using a single iteration of cross-fitting,
where we use one partition to learn $\gamma_{0}$, one partition to
learn $\alpha_{0},$ and partition to fit a pseudo-outcome regression.
This choice is made for simplicity of presentation, and because we
focus on convergence \emph{rates} rather than proportionality constants.
In practice, we expect analysts will perform several iterations, changing
the role of each partition in each iteration. Running multiple cross-fitting
iterations is a well-established means of mitigating sample size loss
(see, e.g., \citealt{Chernozhukov2022-jx,Chernozhukov2022-iq}). Doing
so does not change the final convergence rate (for example, $3n\asymp n$),
although it can improve finite sample performance.

\subsection{\label{subsec:nuisance-estimates}Estimating nuisance functions}

We estimate nuisance functions using the same approach as \citeauthor{Newey2018-da}
(\citeyear{Newey2018-da}; see also \citealp{Chernozhukov2022-jx,Chernozhukov2022-iq}),
which we describe here for completeness. Since $\gamma_{0}$ takes
the form of a conditional expectation (Assumption \ref{assu:lin-func}),
it is relatively straightforward to estimate with standard regression
techniques. Estimation of $\alpha_{0}$ appears to be more complex,
as it does not always take the form of a conditional expectation (e.g.,
if $\alpha_{0}(X)=1/\mathbb{E}(A|X)$, as in $\theta_{0,\text{trt}}$
in Table \ref{tab:Estimands-as-linear}). 

The remarkable thing about linear functionals is that $J\alpha_{0}(X)$
can be estimated via a series regression or a method of moments approach
even when $\alpha_{0}$ is not a conditional expectation. More specifically,
we can approximate $J\alpha_{0}(X)$ with a linear combination of
basis functions $p(J,X)=Jq(X)$, where $q(X)=(q_{1}(X),\dots,q_{k_{n}}(X))$.
The optimal approximation, in sense of minimizing squared error in
the population, would be
\begin{equation}
J\alpha^{\star}(X):=Jq(x)^{\top}\mathbb{E}(p(J,X)p(J,X)^{\top})^{-1}\mathbb{E}(p(J,X)\alpha_{0}(X)),\label{eq:px-est}
\end{equation}
The first expectation can be estimated from the sample variance of
$p(J,X)$. Assumption \ref{assu:lin-func} implies that the second
expectation in Eq (\ref{eq:px-est}) is a $k_{n}$-length vector with
$j^{th}$ element equal to

\begin{equation}
\mathbb{E}(q_{j}(X)J\alpha_{0}(X))=\mathbb{E}\left[m(Z,q_{j})-m(Z,\gamma^{\text{zero}})\right].\label{eq:iqa-est}
\end{equation}
Thus, both expectations in Eq (\ref{eq:px-est}) can be esti\textcolor{black}{mated
via sample moments of observed quantities.}

Based on Eqs (\ref{eq:px-est}) \& (\ref{eq:iqa-est}), we use the
$\tilde{\mathbf{Z}}$ dataset to estimate $\alpha_{0}(x)$ as
\begin{equation}
\tilde{\alpha}(x):=q(x)^{\top}\left(\sum_{i=1}^{n}p(\tilde{J}_{i},\tilde{X}_{i})p(\tilde{J}_{i},\tilde{X}_{i})^{\top}\right)^{-1}\left(\sum_{i=1}^{n}v_{q}(\tilde{Z}_{i})\right),\label{eq:zeta-form}
\end{equation}
where $v_{q}(z)$ is the $k_{n}$-length vector with $j^{th}$ element
equal to $m(z,q_{j})-m(z,\gamma^{\text{zero}})$. In the special case
where the estimand is the conditional covariance ($m(Z,\gamma)=A(Y-\gamma(X))$),
we have $J=1$, $p(J,X)=q(X)$, and $v_{q}(Z)=-q(X)A$, and so $\tilde{\alpha}$
reduces to a standard series estimator.

Conveniently, the same basis function $p$ can also be used to estimate
$\gamma_{0}$. To see why, note that one natural method of approximating
$\gamma_{0}(x)=\mathbb{E}\left[Y|X=x,J=1\right]$ is to use the basis
$q$. As in Eq (\ref{eq:px-est}) the approximation of $\gamma_{0}$
that minimizes the squared error in the subpopulation for which $J=1$
is
\begin{align}
\gamma^{\star}(x) & :=q(x)^{\top}\left[\mathbb{E}\left\{ qq^{\top}|J=1\right\} \right]^{-1}\left\{ \mathbb{E}\left\{ qY|J=1\right\} \right\} \nonumber \\
 & =q(x)^{\top}\left[\frac{1}{\text{Pr}(J=1)}\mathbb{E}\left\{ qJq^{\top}\right\} \right]^{-1}\left[\frac{1}{\text{Pr}(J=1)}\mathbb{E}\left\{ qJY\right\} \right]\nonumber \\
 & =q(x)^{\top}\mathbb{E}\left\{ pp^{\top}\right\} ^{-1}\mathbb{E}\left\{ pY\right\} .\label{eq:gamma-star}
\end{align}
Thus, the least-squares projection of $\gamma_{0}$ onto $q$, in
the subpopulation for which $J=1$, is equivalent to the least-squares
projection of $\gamma_{0}$ onto $p$ in the overall population.

Based on Eq (\ref{eq:gamma-star}), we use the $\hat{\mathbf{Z}}$
dataset to estimate $\gamma_{0}$ with a least squares regression
using the vector $p(J,X)$ as features. We denote the result by
\begin{align}
\hat{\gamma}(x) & :=q(x)^{\top}\left(\sum_{i=1}^{n}p(\hat{J}_{i},\hat{X}_{i})p(\hat{J}_{i},\hat{X}_{i})^{\top}\right)\left(\sum_{i=1}^{n}p(\hat{J}_{i},\hat{X}_{i})\hat{Y}_{i})\right).\label{eq:gamma-hat}
\end{align}

Above, $\tilde{\alpha}\perp\hat{\gamma}$ for 3-way CF, but not for
2-way CF. This is somewhat overloaded notation in the sense that the
interpretation of $\tilde{\alpha}$ and $\hat{\gamma}$ depends on
context. However, this notation will be useful later on, as many results
will be common to both 2-way and 3-way estimators.

\subsection{\label{subsec:Individualized-estimation-with}Conditional effect
estimation via pseudo-outcome regression with splines}

\textcolor{black}{Given the estimates $\hat{\gamma}$ and $\tilde{\alpha}$
from the previous section, we define the estimated pseudo-outcome
$\hat{\tilde{f}}(z):=m(z,\hat{\gamma})+\tilde{\alpha}(x)J(a)(y-\hat{\gamma}(x))$,
and estimate $\ddot{\theta}_{0}=\mathbb{E}\left[f_{0}(Z)|C=\ddot{C}\right]$
by fitting a series regression of $\hat{\tilde{f}}(Z)$ on $C$ in
the data subset $\bar{\mathbf{Z}}$. Let $b(c)$ be a $r_{n}$-dimensional
basis used in this regression, producing the estimate
\begin{equation}
\hat{\tilde{\bar{\ddot{\theta}}}}:=\frac{1}{n}\sum_{i=1}^{n}\bar{\ddot{w}}(\bar{C}_{i})\hat{\tilde{f}}(\bar{Z}_{i}),\text{\ensuremath{\hspace{1em}\text{where\ensuremath{\hspace{1em}}}}}\bar{\ddot{w}}(c)=b(\ddot{C})^{\top}\left(\frac{1}{n}\sum_{i'=1}^{n}b(\bar{C}_{i'})b(\bar{C}_{i'})^{\top}\right)^{-1}b(c).\label{eq:my-w}
\end{equation}
Similar results could be obtained using local polynomial (LP) regression
to estimate $\mathbb{E}\left[f_{0}(Z)|C=\ddot{C}\right]$, but we
omit these results for simplicity of presentation.}

Again, Figure \ref{fig:Flowchart-showing-how} gives a summary of
our proposed workflow and associated notation. In this figure, and
below, we define $\bar{\ddot{\mathbb{P}}}_{n}(f(Z)):=\frac{1}{n}\sum_{i=1}^{n}\bar{\ddot{w}}(\bar{X}_{i})f(\hat{Z}_{i})$
to be the weighted average of a given function $f$ over the dataset
$\bar{\mathbf{Z}}$.\textcolor{black}{{} }

\begin{figure}
\begin{centering}
\includegraphics[width=0.86\columnwidth]{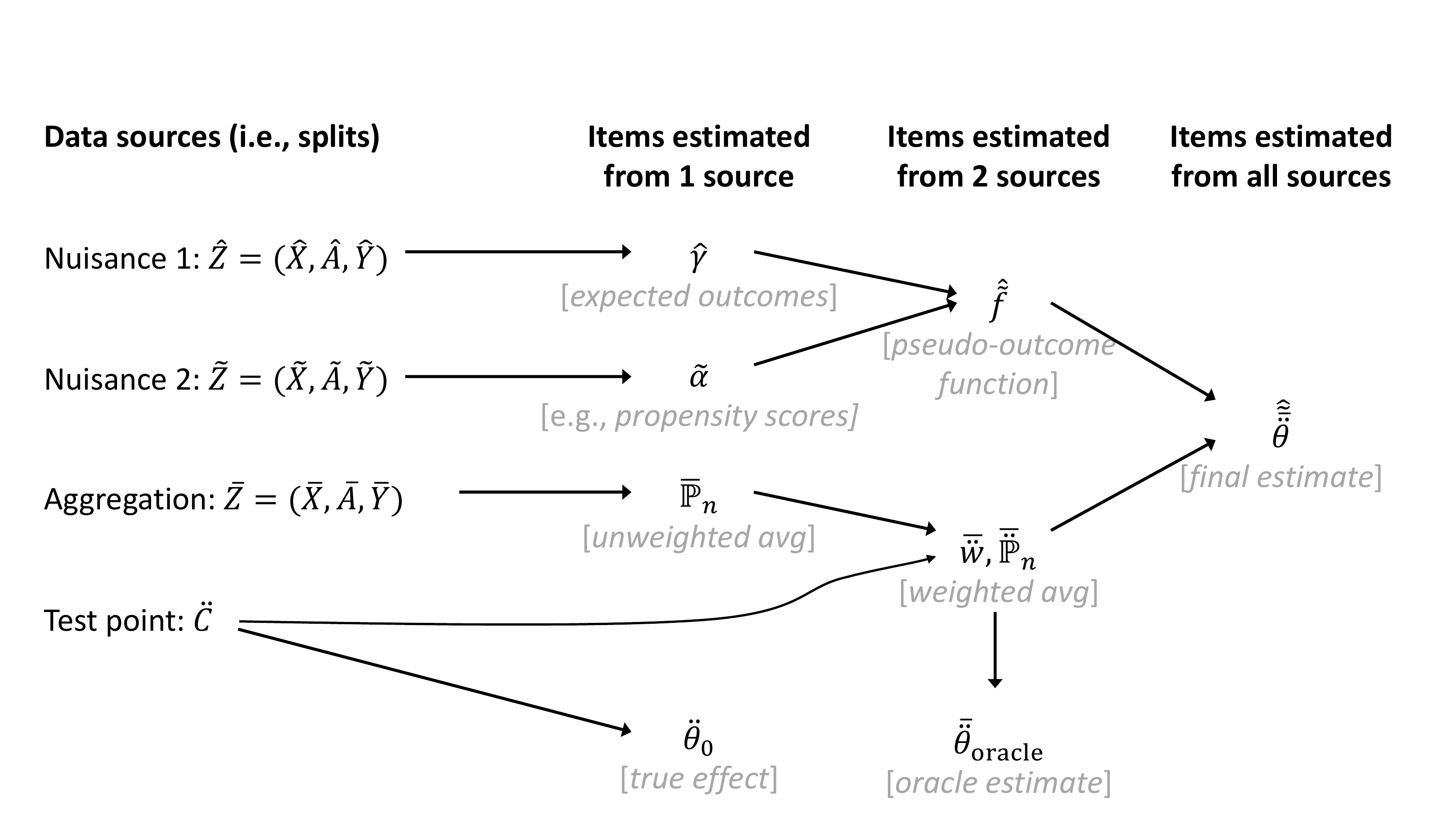}
\par\end{centering}
\caption{\label{fig:Flowchart-showing-how}Flowchart showing how each data
split is used to estimate parameters. An informal description of each
estimator is written in gray. The markings above each estimate (e.g.,
``hats'' or ``bars'') indicate the various data splits that contribute
to that estimate. }

\end{figure}

\section{Main Assumptions\label{sec:Main-Assumptions}}

In order to bound the error of 2-way and 3-way estimators we make
several assumptions. 

The next four assumptions will together imply a bound on the misspecification
error for nuisance estimation. We first assume that $\alpha$ and
$\gamma$ are smooth.
\begin{assumption}
\label{assu:(--Smoothness)--and}($\alpha$-$\gamma$-Smoothness)
$\alpha$ and $\gamma$ are $s_{\alpha}$-smooth and $s_{\gamma}$-smooth
respectively. 
\end{assumption}
Next, we assume that spline estimators are used to estimate $\tilde{\alpha}$
and $\hat{\gamma}$ (as in Section \ref{subsec:nuisance-estimates}).
\begin{assumption}
\label{assu:series-estimates}\mred{(Nuisance estimators) The function}
$\tilde{\alpha}$ is estimated using Eq (\ref{eq:zeta-form}), and
$\hat{\gamma}$ is estimated using a linear regression with the same,
$k_{n}$-dimensional \mred{basis} $p(a,x)=q(x)J(a)$. We also assume
that the basis $q$ is a $\max(\lfloor s_{\alpha}\rfloor,\lfloor s_{\gamma}\rfloor)$
order polynomial spline basis with neighborhoods that are approximately
evenly sized, so that no two points within a neighborhood are $\gtrsim k_{n}^{-1/d_{X}}$
apart.

Similarly, we assume that the basis $b$ used in Section \ref{subsec:Individualized-estimation-with}
is a $r_{n}$-dimensional polynomial spline basis with neighborhoods
sized so that no two points within a neighborhood are $\gtrsim r_{n}^{-1/d_{C}}$
apart.
\end{assumption}
The spacing requirement for $p$ in Assumption \ref{assu:series-estimates}
is satisfied, for example, if $\mathcal{X}$ is a unit hypercube and
we divide each dimension of $\mathcal{X}$ into $j$ evenly sized
segments. We obtain a total of $k_{n}=j^{d_{X}}$ neighborhoods, and
the maximum distance between any two points in a neighborhood is $\sqrt{d_{X}}/j\asymp1/j=k_{n}^{-1/d_{X}}$.
Here, we allow and expect $k_{n}$ to grow with the sample size, although
we limit the rate of growth later on, in Assumption \ref{assu:standardization}. 

Next, we assume that the densities of $X$ and $C$ are not too concentrated
in any one area. Let $K_{b}(c,c')$ be a binary indicator that $c$
and $c'$ are in the same neighborhood, as defined by the spline basis
$b$. Similarly, let $K_{q}(x,x')$ indicate that $x$ and $x'$ are
in the same neighborhood as defined by the basis $q$. 
\begin{assumption}
\label{assu:(Approximately-uniform-),}(Approximately uniform covariates)
There exists a constant $\kappa$ so that\linebreak$\text{Pr}\left(K_{b}(c,C)=1\right)\leq\kappa r_{n}^{-1}$
for any $c$ and $\text{Pr}\left(K_{q}(x,X)=1\right)\leq\kappa k_{n}^{-1}$
for any $x$.
\end{assumption}
\begin{assumption}
\label{assu:standardization}(Standardization) $||q(x)||_{2}^{2}\lesssim k_{n}$
for all $x$; $\mathbb{E}\left[||v_{q}^{(t)}(Z)||_{2}^{2}\right]\lesssim k_{n}$;
and $\mathbb{E}\left(pp^{\top}\right)=I$. Similarly, the basis function
$b$ is standardized so that $||b(c)||_{2}^{2}\lesssim r_{n}$ and
$\mathbb{E}\left(bb^{\top}\right)=I$.
\end{assumption}
For the estimands in Table \ref{tab:Estimands-as-linear}, the first
condition of Assumption \ref{assu:standardization} implies the second.

Together, Assumptions \ref{assu:(--Smoothness)--and}, \ref{assu:series-estimates},
\ref{assu:(Approximately-uniform-),} \& \ref{assu:standardization}
imply a commonly used bound on the nuisance error due to misspecification. 
\begin{lem}
\label{lem:spline-error}(Misspecification error) Let $\gamma^{\star}(x):=q(x)^{\top}\mathbb{E}(pp^{\top})^{-1}\mathbb{E}\left(p\gamma_{0}\right)$
and $\alpha^{\star}(x):=q(x)^{\top}\mathbb{E}(pp^{\top})^{-1}\mathbb{E}\left(p\alpha_{0}\right)$
respectively represent the projections of $\gamma_{0}$ and $\alpha_{0}$
onto $p$. If Assumptions \ref{assu:(--Smoothness)--and}, \ref{assu:series-estimates},
\ref{assu:(Approximately-uniform-),} \& \ref{assu:standardization}
hold, we have 

\begin{equation}
\sup_{x}|\alpha^{\star}(x)-\alpha_{0}(x)|\lesssim k_{n}^{-s_{\alpha}/d_{X}}\hspace{1em}\text{and}\hspace{1em}\sup_{x}|\gamma^{\star}(x)-\gamma_{0}(x)|\lesssim k_{n}^{-s_{\gamma}/d_{X}}.\label{eq:nuisance-error-sup}
\end{equation}
\end{lem}
The intuition for Lemma \ref{lem:spline-error} is that, within any
neighborhood, we can find a polynomial that approximates $\gamma$
and $\alpha$ with a maximum error that is exponentially decreasing
in the neighborhood's size. For example, for $\gamma^{\star}$, this
error is $\lesssim h^{s_{\gamma}}$, where $h$ is the neighborhood
size. Spline estimators essentially partition the covariate space
into neighborhoods that have size proportional to $k_{n}^{-1/d_{X}}$,
and create a local polynomial approximation in each neighborhood.
Thus, as $n$ increases, the misspecification error is $h^{s_{\gamma}}\asymp\left(k_{n}^{-1/d_{X}}\right)^{s_{\gamma}}$.
\begin{assumption}
\label{assu:(Similar-distributions)-}(Positivity) There exists a
constant $k$ such that $0<k\leq\text{Pr}(J=1|X)$.
\end{assumption}
Assumption (\ref{assu:(Similar-distributions)-}) is highly conventional
in the causal inference literature. It essentially states that all
subgroups in the overall population are well represented in the $J=1$
subpopulation.
\begin{assumption}
\label{assu:(Matrix-Bernstein-Inequality)}(Limited basis growth)
$k_{n},r_{n}<n$, and $\frac{k_{n}\log(k_{n})}{n},\frac{r_{n}\log(r_{n})}{n}\rightarrow0$. 
\end{assumption}
Assumption \ref{assu:(Matrix-Bernstein-Inequality)} is needed to
study the asymptotic behavior of the sample covariance matrices for
$p(A,X)$. It is satisfied, for example, if $k_{n}\asymp n^{r}$ for
any $r\in[0,1)$.
\begin{assumption}
\label{assu:Regularity}(Regularity 1) $\mathbb{E}\left(Y^{2}|A,X\right)$
and $J\alpha_{0}(X)$ are bounded.
\end{assumption}
\begin{assumption}
\label{assu:(Regularity-2)}(Regularity 2) The matrix $\bar{\ddot{\mathbb{P}}}_{n}(p(A,X)p(A,X)^{\top})$
is positive semi-definite with probability approaching 1.
\end{assumption}
\begin{assumption}
\label{assu:vvt}(Regularity 3) $\lambda_{\max}\left\{ \mathbb{E}\left(v_{q}(Z)v_{q}(Z)^{\top}\right)\right\} \lesssim1$.
\end{assumption}
For all of the estimands in Table \ref{tab:Estimands-as-linear},
we show later on that Assumption \ref{assu:vvt} follows from Assumptions
\ref{assu:standardization} \& \ref{assu:(Similar-distributions)-}
(Lemma \ref{lem:similar-dist}).
\begin{assumption}
\label{assu:m-convergence}($m$-convergence) For $\bar{Z}\perp\hat{\gamma}$,
we have $\mathbb{E}\left[\left\{ m(\bar{Z},\hat{\gamma})-m(\bar{Z},\gamma_{0})\right\} ^{2}|\hat{\gamma}\right]\lesssim\mathbb{E}\left[\;||\hat{\gamma}(\bar{X})-\gamma_{0}(\bar{X})||_{2}^{2}\;\;|\;\hat{\gamma}\right]$.
\end{assumption}
Assumption \ref{assu:m-convergence} will generally follow from the
fact that $m(z,\gamma)$ is linear in $\gamma$ (Assumption \ref{assu:lin-func}).

Finally, for some of our results, we require $f_{0}$ to be bounded,
and to have a smooth conditional expectation given $C$. Let $\psi(c):=\mathbb{E}(f_{0}(Z)|C=c)$
denote this conditional expectation, so that $\ddot{\theta}=\psi(\ddot{C})$.
\begin{assumption}
\label{assu:bounded-f0}The true pseudo-outcome function $f_{0}$
is bounded on the support of $Z$.
\end{assumption}
\begin{assumption}
\label{assu:smooth-tau} (Smooth $\psi$) The function $\psi$ is
$s_{\psi}$-smooth and the polynomial spline basis $b$ is of order
$\lfloor s_{\psi}\rfloor$. 
\end{assumption}

\section{Main results\label{sec:Results}}

We'll start by comparing $\hat{\tilde{\bar{\ddot{\theta}}}}$ against
the ``oracle'' estimator $\bar{\ddot{\theta}}_{\text{oracle}}=\frac{1}{n}\sum_{i=1}^{n}\bar{\ddot{w}}(\bar{C}_{i})f_{0}(\bar{Z}_{i})$,
where $\bar{\ddot{w}}$ is defined as in Eq (\ref{eq:my-w}).
\begin{thm}
\label{thm:(Error-relative-to-oracle}(Error relative to oracle) Under
Assumptions \ref{assu:lin-func}-\ref{assu:m-convergence}, for both
2-way and 3-way estimators, we have
\begin{align}
\hat{\tilde{\bar{\ddot{\theta}}}}-\bar{\ddot{\theta}}_{\text{oracle}} & \lesssim_{\mathbb{P}}k_{n}^{-\left(s_{\gamma}+s_{\alpha}\right)/d_{X}}+\frac{k_{n}^{-s_{\gamma}/d_{X}}+k_{n}^{-s_{\alpha}/d_{X}}}{\sqrt{n/r_{n}}}+\frac{\sqrt{k_{n}}}{n/\sqrt{r_{n}}}+\frac{k_{n}}{n}.\hphantom{+\frac{k_{n}^{1/2-s_{\gamma}/d_{X}}}{\sqrt{n}}}\label{eq:2-way-upper-bound}
\end{align}

For 3-way estimators in particular,
\begin{equation}
\hat{\tilde{\bar{\ddot{\theta}}}}-\bar{\ddot{\theta}}_{\text{oracle}}\lesssim_{\mathbb{P}}k_{n}^{-\left(s_{\gamma}+s_{\alpha}\right)/d_{X}}+\frac{k_{n}^{-s_{\gamma}/d_{X}}+k_{n}^{-s_{\alpha}/d_{X}}}{\sqrt{n/r_{n}}}+\frac{\sqrt{k_{n}}}{n/\sqrt{r_{n}}}+\frac{k_{n}^{1/2-s_{\gamma}/d_{X}}}{\sqrt{n}}.\hphantom{\frac{k_{n}}{n}ai}\label{eq:3-way-variance}
\end{equation}
Additionally, for 3-way estimates of the conditional covariance (where
$m(Z,\gamma)=A(Y-\gamma(X))$), we have
\begin{equation}
\hat{\tilde{\bar{\ddot{\theta}}}}-\bar{\ddot{\theta}}_{\text{oracle}}\lesssim_{\mathbb{P}}k_{n}^{-\left(s_{\gamma}+s_{\alpha}\right)/d_{X}}+\frac{k_{n}^{-s_{\gamma}/d_{X}}+k_{n}^{-s_{\alpha}/d_{X}}}{\sqrt{n/r_{n}}}+\frac{\sqrt{k_{n}}}{n/\sqrt{r_{n}}}.\hphantom{\frac{k_{n}^{1/2-s_{\gamma}/d_{X}}}{\sqrt{n}}\frac{k_{n}}{n}aaa}\label{eq:3-way-cov-bound}
\end{equation}
\end{thm}
The difference between the bounds for 2-way and 3-way estimators is
that Eq (\ref{eq:3-way-variance}) replaces the $k_{n}/n$ term in
Eq (\ref{eq:2-way-upper-bound}) with $k_{n}^{1/2-s_{\gamma}/d_{X}}/\sqrt{n}$,
which is $\leq1/\sqrt{n}$ whenever $s_{\gamma}>d_{X}/2$. As we will
see in next theorem, this difference is most pronounced when $k_{n}$
grows at a rate faster than $\sqrt{nr_{n}}$, otherwise the $k_{n}/n$
term is $\lesssim\sqrt{r_{n}/n}$. In the special case of the conditional
covariance, the $k_{n}^{1/2-s_{\gamma}/d_{X}}/\sqrt{n}$ term can
be removed completely due to the fact that $\tilde{\alpha}$ reduces
to a series estimator with properties comparable to $\hat{\gamma}$
(see Section \ref{subsec:nuisance-estimates}).

For the simpler scenario of unconditional estimands (where $C$ is
assumed to be constant), \citet{Newey2018-da} are able to replace
the $k_{n}^{1/2-s_{\gamma}/d_{X}}/\sqrt{n}$ in Eq (\ref{eq:3-way-variance})
with a smaller quantity in their Theorem 8. However, not all of the
techniques employed by the authors are easily applicable for conditional
effect estimands. We briefly discuss one key difference in Appendix
\ref{sec:Proofs-for-Section}.

Next, we review a bound for the error of the oracle itself, which
follows from fairly standard spline results (e.g., \citealp{Tsybakov2009-yb};
see also \citealp{Kennedy2022-da}).
\begin{thm}
\label{thm:(Oracle-Error)-Under}(Oracle Error) Under Assumptions
\ref{assu:series-estimates}, \ref{assu:(Approximately-uniform-),},
\ref{assu:standardization}, \ref{assu:bounded-f0} \& \ref{assu:smooth-tau},
we have
\begin{equation}
\bar{\ddot{\theta}}_{\text{oracle}}-\ddot{\theta}_{0}\lesssim_{\mathbb{P}}r_{n}^{-s_{\psi}/d_{C}}+\sqrt{\frac{r_{n}}{n}},\label{eq:oracle-error}
\end{equation}
where the first term captures bias and the second term captures variance. 
\end{thm}
In the special case where we do no individualization and $C=1$ (i.e.,
estimating population average parameters in the form of $\mathbb{E}\left[m(Z,\gamma_{0})\right]$),
the first term in Eq (\ref{eq:oracle-error}) drops away as the oracle
estimator is unbiased, and the second term becomes $\lesssim1/\sqrt{n}$.
Here, the $r_{n}$ quantity in Eqs (\ref{eq:2-way-upper-bound}) \&
(\ref{eq:3-way-variance}) can also be replaced with 1.

In the more general case of conditional estimates, the $\sqrt{r_{n}/n}$
term in Eq (\ref{eq:oracle-error}) dominates both the $\left(k_{n}^{-s_{\gamma}/d_{X}}+k_{n}^{-s_{\alpha}/d_{X}}\right)/\sqrt{n/r_{n}}$
and $\sqrt{k_{n}}/(n/\sqrt{r_{n}})$ terms in Eq (\ref{eq:3-way-variance}),
so that combining results gives

\begin{equation}
\hat{\tilde{\bar{\ddot{\theta}}}}-\ddot{\theta}_{0}\lesssim_{\mathbb{P}}r_{n}^{-s_{\psi}/d_{C}}+\sqrt{\frac{r_{n}}{n}}+k_{n}^{-\left(s_{\gamma}+s_{\alpha}\right)/d_{X}}+\frac{k_{n}^{1/2-s_{\gamma}/d_{X}}}{\sqrt{n}}.\label{eq:combined-bound}
\end{equation}
If $s_{\gamma}\geq d_{X}/2$, this implies that $k_{n}$ should be
as large as possible while satisfying Assumption \ref{assu:(Matrix-Bernstein-Inequality)}. 

Thus far, we have considered the estimands in Table \ref{tab:Estimands-as-linear}.
However, it is important to separately consider estimation of the
conditional average treatment effect 
\[
\tau(c):=\mathbb{E}\left[\mathbb{E}\left(Y|X,A=1\right)-\mathbb{E}\left(Y|X,A=0\right)|C=c\right]
\]
While Theorems \ref{thm:(Error-relative-to-oracle} \& \ref{thm:(Oracle-Error)-Under}
imply that estimation of $\tau$ will be no worse than estimation
of either $\mathbb{E}\left(Y|X,A=0\right)$ or $\mathbb{E}\left(Y|X,A=1\right)$,
we also would like to ensure that estimation of $\tau$ performs \emph{better}
than the estimation of it's components in the case where $\tau$ is
smoother than either $\mathbb{E}\left(Y|X,A=0\right)$ or $\mathbb{E}\left(Y|X,A=1\right)$.
We tackle this specific question in the next corollary.
\begin{cor}
\label{cor:(Conditional-Average-Treatment}(Conditional Average Treatment
Effects) Let $\hat{\tilde{\bar{\ddot{\theta}}}}_{\text{trt}}$ and
$\hat{\tilde{\bar{\ddot{\theta}}}}_{\text{ctrl}}$ be the 3-way CF
estimates of $\ddot{\theta}_{\text{trt}}$ and $\ddot{\theta}_{\text{ctrl}}$,
using the procedure described in Section \ref{sec:Proposed-estimation-procedure}.
If Assumptions \ref{assu:lin-func}-\ref{assu:bounded-f0} hold separately
for $\hat{\tilde{\bar{\ddot{\theta}}}}_{\text{trt}}$ and $\hat{\tilde{\bar{\ddot{\theta}}}}_{\text{ctrl}}$;
the function $\tau$ is $s_{\tau}$-smooth; and the polynomial spline
basis $b$ is of order $\lfloor s_{\tau}\rfloor$, then
\end{cor}
\begin{equation}
\left(\hat{\tilde{\bar{\ddot{\theta}}}}_{\text{trt}}-\hat{\tilde{\bar{\ddot{\theta}}}}_{\text{ctrl}}\right)-\tau(\ddot{C})\lesssim_{\mathbb{P}}r_{n}^{-s_{\tau}/d_{C}}+\sqrt{\frac{r_{n}}{n}}+k_{n}^{-\left(s_{\gamma}+s_{\alpha}\right)/d_{X}}+\frac{k_{n}^{1/2-s_{\gamma}/d_{X}}}{\sqrt{n}}.\label{eq:combined-bound-tau}
\end{equation}
Again, if $s_{\gamma}\geq d_{X}/2$, this implies that $k_{n}$ should
be as large as possible while satisfying Assumption \ref{assu:(Matrix-Bernstein-Inequality)}.
\citet{Kennedy2022-da} comes to a similar conclusion after their
Corollary 2. 

Our approach and results are similar to those of the lp-R-learner
\citep{Kennedy2022-da,Kennedy2022-cn}, although there are important
differences. \citet{Kennedy2022-da} replaces the $k_{n}^{1/2-s_{\gamma}/d_{X}}/\sqrt{n}$
term with a term that depends on the smoothness of the propensity
score. \citet{Kennedy2022-cn} go further, removing this term altogether
and achieving a minimax lower bound. That said, our results more easily
accomodate the scenario where $C\neq X$, taking advantage of a lower
oracle bias $r_{n}^{-s_{\tau}/d_{C}}$. By contrast, because the lp-R-learner
involves weights that depend on the full $X$ vector, it is less easily
extended to the $C\neq X$ scenario. This $C\neq X$ setting may be
of interest in many medical decision making settings where doctors
do not have access to the full covariate vector, or cannot base decisions
off of the full vector $X$ for ethical reasons (e.g., determining
treatment based on income).  \textcolor{purple}{}

\section{Discussion}

In this work we propose a combination of 3-way cross-fitting and pseudo-outcome
regression that produces estimates of conditional effects with quickly
converging second order remainder terms. Our results apply to both
the CATE and the conditional covariance, as well as other linear functionals.

Several exciting areas remain open for future work. First, while we
propose using quantities such as the CATE to inform individual treatment
decisions, it would be interesting to see if similar results hold
when optimizing a binary decision rule \citep{Zhao2012-qx,Zhang2012-ve,Foster2019-nb,Laber2014-hi,Luedtke2016-mi,Luedtke2016-te}.
Second, while we show that 3-way CF can produce faster convergence
\emph{rates} than 2-way CF, it could underperform in finite sample
due to the fact that 3-way CF estimates nuisance functions from smaller
samples splits. Simulation studies could be helpful in studying this
trade-off, and the extent to which is mitigated by multiple iterations
of cross-fitting (see Section \ref{sec:Proposed-estimation-procedure}).

\bibliographystyle{apalike}
\bibliography{_surv}

\appendix

\section{\label{sec:Proofs-for-Section}Proof of Theorem \ref{thm:(Error-relative-to-oracle}}

We start with a general proof outline. First, we review a fact that
\citet{Newey2018-da} frequently apply in their results: if $1_{n}A_{n}\lesssim_{\mathbb{P}}b_{n}$
and $1_{n}$ is an indicator satisfying $\text{Pr}(1_{n}=1)\rightarrow1$
(at any rate), then $A_{n}\lesssim_{\mathbb{P}}b_{n}$ as well. Thus,
when attempting to bound $\hat{\tilde{\bar{\ddot{\theta}}}}-\bar{\ddot{\theta}}_{\text{oracle}}$
in probability it will be sufficient to show that $1_{n}\times(\hat{\tilde{\bar{\ddot{\theta}}}}-\bar{\ddot{\theta}}_{\text{oracle}})$
is bounded in probability for some indicator $1_{n}$ satisfying $\text{Pr}(1_{n}=1)\rightarrow1$.

With this in mind, we choose $1_{n}$ to be the product of several
relevant indicators that depend on the sample moments of $p$ and
$b$. Let $\hat{\mathbb{P}}_{n}$, $\tilde{\mathbb{P}}_{n}$ and $\bar{\mathbb{P}}_{n}$
denote sample averages with respect to the three splits $\hat{\mathbf{Z}}$,
$\tilde{\mathbf{Z}}$ and $\bar{\mathbf{Z}}$. That is, for any function
$f$, 
\[
\hat{\mathbb{P}}_{n}(f(Z))=\frac{1}{n}\sum_{i=1}^{n}f(\hat{Z}_{i});\;\;\tilde{\mathbb{P}}_{n}(f(Z))=\frac{1}{n}\sum_{i=1}^{n}f(\tilde{Z}_{i});\;\;\text{and}\;\;\bar{\mathbb{P}}_{n}(f(Z))=\frac{1}{n}\sum_{i=1}^{n}f(\bar{Z}_{i}).
\]
Recall also that $\bar{\ddot{\mathbb{P}}}_{n}(f(Z))=\frac{1}{n}\sum_{i=1}^{n}\bar{\ddot{w}}(\bar{X}_{i})f(\bar{Z}_{i})$
is the weighted average over $\bar{\mathbf{Z}}$. Let
\[
\hat{\Sigma}:=\hat{\mathbb{P}}_{n}\left[pp^{\top}\right],\tilde{\Sigma}:=\tilde{\mathbb{P}}_{n}\left[pp^{\top}\right],\bar{\Sigma}:=\bar{\mathbb{P}}_{n}\left[pp^{\top}\right],\bar{\ddot{\Sigma}}:=\bar{\ddot{\mathbb{P}}}_{n}\left[pp^{\top}\right],\text{ and }\bar{\mathbf{B}}:=\bar{\mathbb{P}}_{n}\left[bb^{\top}\right].
\]
\label{pg-indicators}Let $\bar{\ddot{1}}$ be the indicator that
$\bar{\ddot{\Sigma}}$ is positive definite (p.d., see Assumption
\ref{assu:(Regularity-2)}); that $\lambda_{\max}\left(\bar{\Sigma}\right)\leq3/2$;
and that $\lambda_{\min}\left(\bar{\mathbf{B}}\right)\geq1/2$. Let
$\hat{1}$ and $\tilde{1}$ respectively be the indicators that $\lambda_{\text{min}}\left(\hat{\Sigma}\right)\geq1/2$,
and $\lambda_{\text{min}}\left(\tilde{\Sigma}\right)\geq1/2$. We
will see in Lemma \ref{lem:(Stable-eigenvalues)-Under} that $\text{Pr}(\hat{1}=1)\rightarrow1$,
$\text{Pr}(\tilde{1}=1)\rightarrow1$, and $\text{Pr}(\bar{\ddot{1}}=1)\rightarrow1$.
Going forward, we aim to bound $\hat{1}\tilde{1}\bar{\ddot{1}}\left(\hat{\tilde{\bar{\ddot{\theta}}}}-\bar{\ddot{\theta}}_{\text{oracle}}\right)=\hat{1}\tilde{1}\bar{\ddot{1}}\bar{\ddot{\mathbb{P}}}_{n}\left(\hat{\tilde{f}}(Z)-f_{0}(Z)\right)$
in probability.

We start by expanding difference $\hat{\tilde{f}}-f_{0}$ to obtain
\begin{align}
\hat{\tilde{f}}(Z)-f_{0}(Z) & =m(Z,\hat{\gamma})-m(Z,\gamma_{0})+\tilde{\alpha}J\left(Y-\hat{\gamma}\right)-\alpha_{0}J\left(Y-\gamma_{0}\right)\nonumber \\
 & =m(Z,\hat{\gamma})-m(Z,\gamma_{0})\mred{-\alpha_{0}J\hat{\gamma}+\alpha_{0}J\gamma_{0}}\nonumber \\
 & \hspace{1em}\hspace{1em}+\tilde{\alpha}JY-\alpha_{0}JY\mred{-\tilde{\alpha}J\gamma_{0}}+\alpha_{0}J\gamma_{0}\nonumber \\
 & \hspace{1em}\hspace{1em}-\tilde{\alpha}J\hat{\gamma}\mred{+\tilde{\alpha}J\gamma_{0}+\alpha_{0}J\hat{\gamma}-\alpha_{0}J\gamma_{0}}\nonumber \\
 & =m(Z,\hat{\gamma})-m(Z,\gamma_{0})-\alpha_{0}J(\hat{\gamma}-\gamma_{0})\label{eq:mean01}\\
 & \hspace{1em}\hspace{1em}+\left(\tilde{\alpha}-\alpha_{0}\right)J(Y-\gamma_{0})\label{eq:mean02}\\
 & \hspace{1em}\hspace{1em}-(\tilde{\alpha}-\alpha_{0})J(\hat{\gamma}-\gamma_{0}).\label{eq:not-mean0}
\end{align}

\textcolor{black}{}We replace the term in Line (\ref{eq:not-mean0}),
$(\tilde{\alpha}-\alpha_{0})J(\hat{\gamma}-\gamma_{0})$, with $(\tilde{\alpha}-\alpha^{\star}+\alpha^{\star}-\alpha_{0})J(\hat{\gamma}-\gamma^{\star}+\gamma^{\star}-\gamma_{0})$,
to obtain
\begin{align}
\hat{1}\tilde{1}\bar{\ddot{1}}\left(\hat{\tilde{\bar{\ddot{\theta}}}}-\bar{\ddot{\theta}}_{\text{oracle}}\right) & =\hat{1}\tilde{1}\bar{\ddot{1}}\bar{\ddot{\mathbb{P}}}_{n}\left(\hat{\tilde{f}}(Z)-f_{0}(Z)\right)\nonumber \\
 & =\hat{1}\tilde{1}\bar{\ddot{1}}\bar{\ddot{\mathbb{P}}}_{n}\left\{ m(Z,\hat{\gamma})-m(Z,\gamma_{0})-\alpha_{0}J(\hat{\gamma}-\gamma_{0})\right.\label{eq:apply-linear}\\
 & \hspace{1em}\hspace{1em}+\left(\tilde{\alpha}-\alpha_{0}\right)J(Y-\gamma_{0})\label{eq:apply-IE}\\
 & \hspace{1em}\hspace{1em}\mred{-(\alpha^{\star}-\alpha_{0})J(\hat{\gamma}-\gamma^{\star})}\label{eq:a_star_0-g_hat_star}\\
 & \hspace{1em}\hspace{1em}\mred{-(\tilde{\alpha}-\alpha^{\star})J(\gamma^{\star}-\gamma_{0})}\label{eq:a_tilde_star-g_star_0}\\
 & \hspace{1em}\hspace{1em}\mred{-(\alpha^{\star}-\alpha_{0})J(\gamma^{\star}-\gamma_{0})}\label{eq:2stars}\\
 & \hspace{1em}\hspace{1em}\left.\mred{-(\tilde{\alpha}-\alpha^{\star})J(\hat{\gamma}-\gamma^{\star})}\right\} .\label{eq:diff-term}
\end{align}

The remainder of our proofs are organized as follows. Section \ref{subsec:Helpful-Lemmas}
introduces notation and provides several helpful Lemmas that are used
throughout our results. Section \ref{subsec:Proof-apply-linear-ie-common-terms}
shows that Lines (\ref{eq:apply-linear}) \& (\ref{eq:apply-IE})
are

\[
\lesssim_{\mathbb{P}}\frac{k_{n}^{-s_{\gamma}/d_{X}}+k_{n}^{-s_{\alpha}/d_{X}}}{\sqrt{n/r_{n}}}+\frac{\sqrt{k_{n}}}{n/\sqrt{r_{n}}},
\]
Section \ref{subsec:Reducing-Line_a_star} shows that Line (\ref{eq:a_star_0-g_hat_star})
is 
\[
\lesssim_{\mathbb{P}}k^{-(s_{\gamma}+s_{\alpha})/d_{X}}+\frac{k_{n}^{-s_{\alpha}/d_{X}}}{\sqrt{n/r_{n}}},
\]
Section \ref{subsec:Reducing-Line_g_star} shows that Line (\ref{eq:a_tilde_star-g_star_0})
is 

\[
\lesssim_{\mathbb{P}}k^{-(s_{\gamma}+s_{\alpha})/d_{X}}+\frac{k_{n}^{-s_{\gamma}/d_{X}}}{\sqrt{n/r_{n}}}+\frac{k_{n}^{1/2-s_{\gamma}/d_{X}}}{\sqrt{n}},
\]
where the last term can be removed when the estimand is the conditional
covariance ($m(Z,\gamma)=A(Y-\gamma(X))$). Section \ref{subsec:Reducing-2stars}
shows that Line (\ref{eq:2stars}) is 
\[
\lesssim_{\mathbb{P}}k_{n}^{-(s_{\alpha}+s_{\gamma})/d_{X}},
\]
All of these results from Sections \ref{subsec:Proof-apply-linear-ie-common-terms}-\ref{subsec:Reducing-2stars}
apply to both 2-way and 3-way estimators. Section \ref{subsec:Reducing-Line_g_star}
is where our results differ most from those of \citet{Newey2018-da},
who are able to use the fact that the residuals after projecting onto
a set of basis vectors are orthogonal to those basis vectors. In our
case however, this technique does not translate, as the residuals
are not orthogonal conditional on the personalization traits $C$.

The primary difference between 2-way and 3-way estimators comes in
Line (\ref{eq:diff-term}), as it involves interactions between the
errors of each nuisance estimator. For 2-way estimators, Section \ref{subsec:Proof-of-last-2way-term}
shows that Line (\ref{eq:diff-term}) is $\lesssim_{\mathbb{P}}k/n$.
For 3-way estimators, Section \ref{subsec:Proof-of-last-term-3way}
leverages the fact that the nuisance errors in Line (\ref{eq:diff-term})
are independent (conditional on $\bar{\mathbf{Z}},\ddot{C}$) to show
that Line (\ref{eq:diff-term}) is

\begin{align}
 & \lesssim_{\mathbb{P}}k_{n}^{-\left(s_{\alpha}+s_{\gamma}\right)/d_{X}}+\frac{k^{-s_{\gamma}/d_{X}}}{\sqrt{n/r_{n}}}+\frac{\sqrt{k_{n}}}{n/\sqrt{r_{n}}}+\frac{k_{n}}{n}\left(k^{-s_{\gamma}/d_{X}}\frac{k_{n}\log\left(k_{n}\right)}{n}\right)\nonumber \\
 & \leq k_{n}^{-\left(s_{\alpha}+s_{\gamma}\right)/d_{X}}+\frac{k^{-s_{\gamma}/d_{X}}}{\sqrt{n/r_{n}}}+\frac{\sqrt{k_{n}}}{n/\sqrt{r_{n}}}+\frac{k_{n}^{1/2-s_{\gamma}/d_{X}}}{\sqrt{n}},\label{eq:leq-k-half}
\end{align}
where Line (\ref{eq:leq-k-half}) comes from Assumption \ref{assu:(Matrix-Bernstein-Inequality)}
and
\begin{align*}
\frac{k_{n}}{n}\left(k^{-s_{\gamma}/d_{X}}\frac{k_{n}\log\left(k_{n}\right)}{n}\right) & =\frac{k_{n}^{1/2-s_{\gamma}/d_{X}}}{n^{1/2}}\left\{ \frac{k_{n}^{3/2}\log\left(k_{n}\right)}{n^{3/2}}\right\} .
\end{align*}
Section \ref{subsec:Proof-of-last-term-3way} also shows that the
$k_{n}^{1/2-s_{\gamma}/d_{X}}/\sqrt{n}$ term in Eq (\ref{eq:leq-k-half})
can be removed when estimating the conditional covariance.

\subsection{Helpful Lemmas\label{subsec:Helpful-Lemmas}}

First, we provide the proof of Lemma \ref{lem:spline-error}, which
says that the spline regression coefficients that minimize the expected
squared error loss have bounded error. The steps are similar to those
in Appendix \ref{sec:proof-Oracle-error}. 
\begin{proof}
For any $s$-smooth function $f$, let $f_{s,x}$ denote the $\lfloor s\rfloor$
order Taylor approximation of $f$ at $x$. Let $q$ be a defined
as in Assumption \ref{assu:series-estimates} \& \ref{assu:standardization}.
Let $\beta_{s,x}$ be a set of coefficients such that $f_{s,x}(x')=\beta_{s,x}^{\top}q(x')$.
Since $q$ contains several neighborhoods, there are multiple valid
choices for $\beta_{s,x}$.

Let 
\[
f^{\star}(x)=q(x)^{\top}\mathbb{E}\left(pp^{\top}\right)^{-1}\mathbb{E}\left(pf\right)=q(x)^{\top}\mathbb{E}\left(qJf\right)
\]
 be the projection of $f$ onto the basis $p$. We will show that
$\sup_{x}|f^{\star}(x)-f(x)|\lesssim r_{n}^{-s/d_{X}}$.

Note that 
\begin{align}
f(x)=f_{s,x}(x)=q(x)^{\top}\beta_{s,x}= & q(x)^{\top}\mathbb{E}\left[qJq^{\top}\right]\beta_{s,x}\nonumber \\
= & q(x)^{\top}\mathbb{E}\left[q(X)Jf_{s,x}(X)\right].\label{eq:combine-f}
\end{align}
Applying this, we have
\begin{align*}
f^{\star}(x)-f(x) & =q(x)^{\top}\mathbb{E}\left[qJf\right]-f(x)\\
 & =q(x)^{\top}\mathbb{E}\left[q(X)J\left\{ f(X)-f_{s,x}(X)\right\} \right] &  & \text{From Eq \eqref{eq:combine-f}}\\
 & =\mathbb{E}\left[q(x)^{\top}q(X)J\left\{ f(X)-f_{s,x}(X)\right\} K_{q}(X,x)\right] &  & \text{Def of }q\text{ and }K_{q}\\
 & \leq\mathbb{E}\left[||q(x)||_{2}^{2}\times|f(X)-f_{s,x}(X)|\times K_{q}(X,x)\right]\\
 & \lesssim k_{n}\mathbb{E}\left[||X-x||^{s}\times K_{q}(X,x)\right] &  & \text{Assm \ref{assu:standardization}}\\
 & \lesssim k_{n}^{1-s/d_{X}}\mathbb{E}\left[K_{q}(X,x)\right] &  & \text{Assm \ref{assu:series-estimates}}\\
 & \lesssim k_{n}^{-s/d} &  & \text{Assm \ref{assu:(Approximately-uniform-),}}.
\end{align*}
\end{proof}
\begin{lem}
\label{lem:(Stable-eigenvalues)-Under}(Stable eigenvalues) Under
Assumptions \ref{assu:standardization} \& \ref{assu:(Matrix-Bernstein-Inequality)},
$\text{Pr}(\hat{1}=1),\text{Pr}(\tilde{1}=1),\text{Pr}(\bar{\ddot{1}}=1)\rightarrow1$.
Additionally $\lambda_{\text{max}}(\hat{\Sigma}-I)\lesssim_{\mathbb{P}}\sqrt{\left\{ k_{n}\log\left(k_{n}\right)\right\} /n}$.
\end{lem}
\begin{proof}
As in \citeauthor{Belloni2015-tn} (\citeyear{Belloni2015-tn}; see
their page 358), note that $\lambda_{\text{min}}\left(\hat{\Sigma}\right)<1/2$
only if there exists $u$ such that $||u||_{2}^{2}=1$ and 
\[
1/2<|u^{\top}\hat{\Sigma}u-1|=|u^{\top}\left(\hat{\Sigma}-I\right)u|\leq\lambda_{\text{max}}(\hat{\Sigma}-I).
\]
We can use the the Matrix Bernstein Inequality (MBI; Sections 1.6.2-1.6.3
of \citealp{Tropp2015-dx}; see also Lemma 6.2 \citealp{Belloni2015-tn},
and \citealp{Rudelson1999-gn}) to limit the probability that $\frac{1}{2}<\lambda_{\text{max}}(\hat{\Sigma}-I)$
occurs. The MBI states that
\[
\mathbb{E}\left(\lambda_{\text{max}}(\hat{\Sigma}-I)\right)\lesssim\frac{k_{n}\log(2k_{n})}{n}+\sqrt{\frac{\lambda_{\text{max}}(I)k_{n}\log(2k_{n})}{n}}\lesssim\sqrt{\frac{k_{n}\log k_{n}}{n}}.
\]

Markov's Inequality and Assumption (\ref{assu:(Matrix-Bernstein-Inequality)})
then tells us that $\lambda_{\text{max}}(\hat{\Sigma}-I)\lesssim_{\mathbb{P}}\sqrt{\left\{ k_{n}\log\left(k_{n}\right)\right\} /n}$,
and
\[
\text{Pr}\left(\frac{1}{2}<\lambda_{\text{max}}(\hat{\Sigma}-I)\right)\lesssim2\sqrt{\frac{k_{n}\log k_{n}}{n}}\rightarrow0.
\]
Thus, $\text{Pr}\left(\frac{1}{2}<\lambda_{\text{min}}(\hat{\Sigma})\right)\rightarrow1$.
The same steps show that $\text{Pr}\left(\frac{1}{2}<\lambda_{\text{min}}(\tilde{\Sigma})\right)$
and $\text{Pr}\left(\frac{1}{2}<\lambda_{\text{min}}(\bar{\mathbf{B}})\right)\rightarrow1$.

It remains to show that $\text{Pr}\left(\lambda_{\text{max}}(\bar{\Sigma})\leq3/2\right)\rightarrow1$.
Following similar steps, we see that $\lambda_{\text{max}}(\bar{\Sigma})>3/2$
only if there exists $u$ such that $||u||_{2}^{2}=1$ and 
\[
1/2<|u^{\top}\bar{\Sigma}u-1|=|u^{\top}\left(\bar{\Sigma}-I\right)u|\leq\lambda_{\text{max}}(\bar{\Sigma}-I).
\]
The same steps as above show that $\text{Pr}\left(\frac{1}{2}<\lambda_{\text{max}}(\bar{\Sigma}-I)\right)\rightarrow0.$ 
\end{proof}
We next show a related fact about the moments of $q$, which is useful
in showing Assumption \ref{assu:vvt}
\begin{lem}
\label{lem:similar-dist}Under Assumptions \ref{assu:standardization}
\& \ref{assu:(Similar-distributions)-}, $\lambda_{\max}\left\{ \mathbb{E}\left(q(X)q(X)^{\top}\right)\right\} \lesssim1.$
\end{lem}
Thus, for all of the estimands in Table \ref{tab:Estimands-as-linear},
Assumption \ref{assu:vvt} follows from Lemma \ref{lem:similar-dist},
Assumption \ref{assu:(Similar-distributions)-} \& Assumption \ref{assu:standardization}.
\begin{proof}
We have
\begin{align*}
\lambda_{\max}\left\{ \mathbb{E}\left(q(X)q(X)^{\top}\right)\right\}  & =\lambda_{\max}\left\{ \mathbb{E}\left(q(X)\frac{\mathbb{E}\left(J|X\right)}{\mathbb{E}\left(J|X\right)}q(X)^{\top}\right)\right\} \\
 & =\lambda_{\max}\left\{ \mathbb{E}\left(q(X)\frac{J}{\mathbb{E}\left(J|X\right)}q(X)^{\top}\right)\right\} \\
 & \lesssim\lambda_{\max}\left\{ \mathbb{E}\left(q(X)Jq(X)^{\top}\right)\right\}  &  & \text{Assm \ref{assu:(Similar-distributions)-}}\\
 & =1 &  & \text{Assm \ref{assu:standardization}}.
\end{align*}
\end{proof}
We next introduce notation to help study the differences $\hat{\gamma}(X)-\gamma^{\star}(X)$
and $\tilde{\alpha}(X)-\alpha^{\star}(X)$. Our strategy here is almost
identical to that of \citeauthor{Newey2018-da}' Lemma A2, although
we include it here for completeness. Let $\delta_{\gamma}^{\star}:=\mathbb{E}\left(p\gamma_{0}\right)$
and $\delta_{\alpha}^{\star}:=\mathbb{E}\left(p\alpha_{0}\right)$,
so that $\gamma^{\star}(x)=q(x)^{\top}\delta_{\gamma}^{\star}$ and
$\alpha^{\star}(x)=q(x)^{\top}\delta_{\alpha}^{\star}$. Let $\hat{\delta}_{\gamma}=\hat{\Sigma}^{-1}\hat{\mathbb{P}}_{n}\left(pY\right)$
and $\tilde{\delta}_{\alpha}=\tilde{\Sigma}^{-1}\hat{\mathbb{P}}_{n}\left(v\right)$,
so that $\hat{\gamma}(x)=q(x)^{\top}\hat{\delta}_{\gamma}$ and $\tilde{\alpha}(x)=q(x)^{\top}\hat{\delta}_{\alpha}$.
Let $\bar{p}_{i}:=p(\bar{A}_{i},\bar{X}_{i}),$ and let
\begin{align}
\hat{p}_{i} & :=p(\hat{A}_{i},\hat{X}_{i}), & \hat{H}_{\gamma i}:= & \hat{p}_{i}\left\{ \hat{Y}_{i}-\gamma_{0}(\hat{X}_{i})\right\} , & \hat{H}_{\gamma\bcd}:= & \frac{1}{n}\sum_{i=1}^{n}\hat{H}_{\gamma i}, & \hat{\Delta}_{\gamma} & :=\hat{\Sigma}^{-1}\hat{H}_{\gamma\bcd},\nonumber \\
 &  & \hat{H}_{\gamma i}^{\star}:= & \hat{p}_{i}\left\{ \gamma_{0}(\hat{X}_{i})-\gamma^{\star}(\hat{X}_{i})\right\} , & \hat{H}_{\gamma\bcd}^{\star}:= & \frac{1}{n}\sum_{i=1}^{n}\hat{H}_{\gamma i}^{\star}, & \hat{\Delta}_{\gamma}^{\star} & :=\hat{\Sigma}^{-1}\hat{H}_{\gamma\bcd}^{\star},\nonumber \\
\tilde{p}_{i} & :=p(\tilde{A}_{i},\tilde{X}_{i}), & \tilde{H}_{\alpha i}:= & v_{q}(\tilde{Z}_{i})-\tilde{p}_{i}\alpha_{0}(\tilde{X}_{i}), & \tilde{H}_{\alpha\bcd}:= & \frac{1}{n}\sum_{i=1}^{n}\tilde{H}_{\alpha i}, & \tilde{\Delta}_{\alpha} & :=\tilde{\Sigma}^{-1}\tilde{H}_{\alpha\bcd},\nonumber \\
 &  & \tilde{H}_{\alpha i}^{\star}:= & \tilde{p}_{i}\left\{ \alpha_{0}(\tilde{X}_{i})-\alpha^{\star}(\tilde{X}_{i})\right\} , & \tilde{H}_{\alpha\bcd}^{\star}:= & \frac{1}{n}\sum_{i=1}^{n}\tilde{H}_{\alpha i}^{\star},\text{ and} & \tilde{\Delta}_{\alpha}^{\star} & :=\tilde{\Sigma}^{-1}\tilde{H}_{\alpha\bcd}^{\star}.\label{eq:all_my_Hs}
\end{align}
Let $\hat{\mathbf{J}},\tilde{\mathbf{J}}$ and $\bar{\mathbf{J}}$
be the vectors containing the values of $J$ in the datasets \textbf{$\hat{\mathbf{Z}},\tilde{\mathbf{Z}}$}
and $\bar{\mathbf{Z}}$ respectively. 

As in Eq (7.3) of \citet{Newey2018-da}, we leverage the fact that
\begin{align}
\hat{\Delta}_{\gamma}+\hat{\Delta}_{\gamma}^{\star} & =\hat{\Sigma}^{-1}\hat{\mathbb{P}}_{n}\left[p\times\left\{ Y-\gamma^{\star}\right\} \right]\nonumber \\
 & =\hat{\Sigma}^{-1}\hat{\mathbb{P}}_{n}\left[p\times\left\{ Y-p^{\top}\delta_{\gamma}^{\star}\right\} \right]\nonumber \\
 & =\hat{\Sigma}^{-1}\hat{\mathbb{P}}_{n}\left(pY\right)-\hat{\Sigma}^{-1}\hat{\mathbb{P}}_{n}\left(pp^{\top}\right)\delta_{\gamma}^{\star}\nonumber \\
 & =\hat{\delta}_{\gamma}-\delta_{\gamma}^{\star},\label{eq:Delta-sum-gam}
\end{align}
and, similarly, 
\begin{equation}
\tilde{\Delta}_{\alpha}+\tilde{\Delta}_{\alpha}^{\star}=\tilde{\delta}_{\alpha}-\delta_{\alpha}^{\star}.\label{eq:Delta-sum-alpha}
\end{equation}
We will see that studying the left-hand sides of Eqs (\ref{eq:Delta-sum-gam})
\& (\ref{eq:Delta-sum-alpha}) provides a tractable means of studying
s $\hat{\gamma}-\gamma^{\star}$ and $\tilde{\alpha}-\alpha^{\star}$.

First, we note $\hat{H}_{\gamma i},\hat{H}_{\gamma i}^{\star},\tilde{H}_{\alpha i},$
and $\tilde{H}_{\alpha i}^{\star}$ all have mean zero.
\begin{lem}
\label{lem:H-moments}Under Assumptions \ref{assu:lin-func} \& \ref{assu:series-estimates},
$\mathbb{E}\left(\hat{H}_{\gamma i}\right)=\mathbb{E}\left(\hat{H}_{\gamma i}^{\star}\right)=\mathbb{E}\left(\tilde{H}_{\alpha i}\right)=\mathbb{E}\left(\tilde{H}_{\alpha i}^{\star}\right)=0$. 
\end{lem}
\begin{proof}
We can see that $\mathbb{E}\left(\hat{H}_{\gamma i}\right)=0$ from
iterated expectations, and $\mathbb{E}\left(\tilde{H}_{\alpha i}\right)=0$
from Assumption \ref{assu:lin-func}. Next,
\begin{align}
\mathbb{E}\left(\hat{H}_{\gamma i}^{\star}\right) & =\mathbb{E}\left[\mred{p(A,X)}\left\{ \gamma_{0}(X)-\mred{\gamma^{\star}(X)}\right\} \right]\nonumber \\
 & =\mathbb{E}\left[\mred{q(X)J}\left\{ \gamma_{0}(X)-\mred{q(X)^{\top}\delta_{\gamma}^{\star}}\right\} \right]\nonumber \\
 & =\mathbb{E}\left[q(X)J\gamma_{0}(X)\right]-\mathbb{E}\left[q(X)Jq(X)^{\top}\right]\delta_{\gamma}^{\star}\nonumber \\
 & =\delta_{\gamma}^{\star}-I\delta_{\gamma}^{\star}\\
 & =0.\nonumber 
\end{align}
The same steps show $\mathbb{E}\left(\tilde{H}_{\alpha i}^{\star}\right)=0$.
\end{proof}
Using this fact, we can study the terms $\hat{\Delta}_{\gamma},\hat{\Delta}_{\gamma}^{\star},\tilde{\Delta}_{\alpha},$
and $\tilde{\Delta}_{\alpha}^{\star}$. 
\begin{lem}
(Delta-values)\label{lem:(Delta-values)} Under Assumptions \ref{assu:lin-func}-\ref{assu:m-convergence}\textcolor{purple}{}
we have:
\begin{enumerate}
\item \label{enu:D-gam}$\mathbb{E}\left(\hat{1}\norm{\hat{\Delta}_{\gamma}}^{2}\right)\lesssim k_{n}/n.$
\item \label{enu:D-alpha}$\mathbb{E}\left(\tilde{1}\norm{\tilde{\Delta}_{\alpha}}^{2}\right)\lesssim k_{n}/n.$
\item \label{enu:D-gam-star}$\mathbb{E}\left(\hat{1}\norm{\hat{\Delta}_{\gamma}^{\star}}^{2}\right)\lesssim k_{n}^{-2s_{\gamma}/d_{X}+1}/n.$
\item \label{enu:D-alpha-star}$\mathbb{E}\left(\tilde{1}\norm{\tilde{\Delta}_{\alpha}^{\star}}^{2}\right)\lesssim k_{n}^{-2s_{\alpha}/d_{X}+1}/n.$
\item \label{enu:DDhat}$\mathbb{E}\left(\hat{1}\hat{\Delta}_{\gamma}\hat{\Delta}_{\gamma}^{\top}|\hat{\mathbf{J}},\hat{\mathbf{X}}\right)\lesssim_{\text{p.s.d.}}\hat{1}n^{-1}\hat{\Sigma}^{-1}$,
where $A\lesssim_{\text{p.s.d.}}B$ indicates that there exists $c$
such that $cB-A$ is positive semi-definite (p.s.d.).
\item \label{enu:DDtilde}$\lambda_{\max}\left\{ \mathbb{E}\left(\tilde{H}_{\alpha\bcd}\tilde{H}_{\alpha\bcd}^{\top}\right)\right\} \lesssim1/n$.
\item \label{enu:ddiff-gam}$\mathbb{E}\left(\hat{1}||\hat{\delta}_{\gamma}-\delta_{\gamma}^{\star}||_{2}^{2}\right)\lesssim k_{n}/n$.
\item \label{enu:ddiff-alpha}$\mathbb{E}\left(\tilde{1}||\tilde{\delta}_{\alpha}-\delta_{\alpha}^{\star}||_{2}^{2}\right)\lesssim k_{n}/n$.
\item \label{enu:gam-hat-star}If $\hat{1},\hat{\delta}_{\gamma}\perp\bar{X}$,
then $\mathbb{E}\left\{ \hat{1}(\hat{\gamma}(\bar{X})-\gamma^{\star}(\bar{X}))^{2}\right\} \lesssim k_{n}/n$.
\item \label{enu:alpha-tilde-star}If $\tilde{1},\tilde{\delta}_{\alpha}\perp\bar{X},$
then $\mathbb{E}\left\{ \tilde{1}(\tilde{\alpha}(\bar{X})-\alpha^{\star}(\bar{X}))^{2}\right\} \lesssim k_{n}/n$.
\end{enumerate}
\end{lem}
\begin{proof}
A fact we use for Points \ref{enu:D-gam}, \ref{enu:D-alpha}, \ref{enu:D-gam-star},
\& \ref{enu:D-alpha-star} is that, for any $iid$ random vector $V$
with mean zero, we have 
\begin{equation}
\mathbb{E}\left[\hat{\mathbb{P}}_{n}\left(V\right)^{\top}\hat{\mathbb{P}}_{n}\left(V\right)\right]=\frac{1}{n^{2}}\sum_{i,j}^{n}\mathbb{E}\left(\hat{V}_{i}^{\top}\hat{V}_{j}\right)=\frac{1}{n^{2}}\sum_{i=1}^{n}\mathbb{E}\left(\hat{V}_{i}^{\top}\hat{V}_{i}\right)=\frac{1}{n}\mathbb{E}\left(\hat{V}_{i}^{\top}\hat{V}_{i}\right),\label{eq:vtv}
\end{equation}
and, similarly,
\begin{equation}
\mathbb{E}\left[\hat{\mathbb{P}}_{n}\left(V\right)\hat{\mathbb{P}}_{n}\left(V\right)^{\top}\right]=\frac{1}{n}\mathbb{E}\left(\hat{V}_{i}\hat{V}_{i}^{\top}\right).\label{eq:vvtsum}
\end{equation}
Next, applying Line \ref{eq:vtv} \& Lemma \ref{lem:H-moments},
\begin{align*}
\mathbb{E}\left[\hat{1}\left\vert \left\vert \hat{\Delta}_{\gamma}\right\vert \right\vert _{2}^{2}\right]\leq\mathbb{E}\left[\hat{1}\hat{H}_{\gamma\bcd}^{\top}\left(\hat{\Sigma}^{-}\right)^{2}\hat{H}_{\gamma\bcd}\right] & \leq4\mathbb{E}\left[\hat{H}_{\gamma\bcd}^{\top}\hat{H}_{\gamma\bcd}\right]\\
 & =\frac{4}{n}\mathbb{E}\left[\hat{H}_{\gamma i}^{\top}\hat{H}_{\gamma i}\right] &  & \text{Eq \ref{eq:vtv}}\\
 & =\frac{4}{n}\mathbb{E}\left[||p||_{2}^{2}(Y-\gamma_{0})^{2}\right]\\
 & \lesssim\frac{k_{n}}{n}. &  & \text{Assms \ref{assu:Regularity} \& \ref{assu:standardization}.}
\end{align*}

\textcolor{red}{}Similar steps show Point \ref{enu:D-alpha}. Specifically,
Eq \ref{eq:vtv} \& Assumptions \ref{assu:standardization} \& \ref{assu:Regularity}
again show that

\[
\mathbb{E}\left(\hat{1}\left\vert \left\vert \hat{\Delta}_{\gamma}^{\star}\right\vert \right\vert _{2}^{2}\right)\leq\frac{4}{n}\mathbb{E}\left(\tilde{H}_{\alpha i}^{\top}\tilde{H}_{\alpha i}\right)\leq\frac{8}{n}\left\{ \mathbb{E}\left(||v_{q}||_{2}^{2}\right)+2\mathbb{E}\left(||p||_{2}^{2}\alpha_{0}^{2}\right)\right\} \lesssim\frac{k_{n}}{n}.
\]

For Point \ref{enu:D-gam-star}, Eq (\ref{eq:vtv}) \& Assumptions
\ref{assu:standardization} show that\textcolor{red}{}
\begin{align*}
\mathbb{E}\left(\hat{1}\left\vert \left\vert \hat{\Delta}_{\gamma}^{\star}\right\vert \right\vert _{2}^{2}\right) & \leq\frac{4}{n}\mathbb{E}\left(\hat{H}_{\gamma i}^{\star^{\top}}\hat{H}_{\gamma i}^{\star}\right)=\frac{4}{n}\mathbb{E}\left\{ ||p||_{2}^{2}(\gamma_{0}-\gamma^{\star})^{2}\right\} \lesssim\frac{k_{n}^{-2s_{\gamma}/d_{X}+1}}{n}.
\end{align*}
where the last $\lesssim$ comes from Assumptions \ref{assu:series-estimates}
\& \ref{assu:standardization}.

Similarly, for Point \ref{enu:D-gam-star}, Eq (\ref{eq:vtv}) \&
Assumptions \ref{assu:standardization} show that \textcolor{red}{}

\begin{align*}
\mathbb{E}\left(\tilde{1}\left\vert \left\vert \tilde{\Delta}_{\alpha}^{\star}\right\vert \right\vert _{2}^{2}\right) & \leq\frac{4}{n}\mathbb{E}\left(\tilde{H}_{\alpha i}^{\star^{\top}}\tilde{H}_{\alpha i}^{\star}\right)=\frac{4}{n}\mathbb{E}\left\{ ||p||_{2}^{2}(\alpha_{0}-\alpha^{\star})^{2}\right\} \lesssim\frac{k_{n}^{-2s_{\alpha}/d_{X}+1}}{n}.
\end{align*}

For Point \ref{enu:DDhat},
\begin{align*}
\mathbb{E}\left(\hat{1}\hat{\Delta}_{\gamma}\hat{\Delta}_{\gamma}^{\top}|\hat{\mathbf{J}},\hat{\mathbf{X}}\right) & =\hat{1}\hat{\Sigma}^{-1}\mathbb{E}\left(\hat{H}_{\gamma\bcd}\hat{H}_{\gamma\bcd}^{\top}|\hat{\mathbf{J}},\hat{\mathbf{X}}\right)\hat{\Sigma}^{-1}\\
 & =\hat{1}\hat{\Sigma}^{-1}\left\{ \frac{1}{n^{2}}\sum_{i=1}^{n}\mathbb{E}\left(\hat{H}_{\gamma i}\hat{H}_{\gamma i}^{\top}|\hat{\mathbf{J}},\hat{\mathbf{X}}\right)\right\} \hat{\Sigma}^{-1} &  & \text{Eq \ref{eq:vvtsum}}\\
 & =\hat{1}\hat{\Sigma}^{-1}\left[\frac{1}{n^{2}}\sum_{i=1}^{n}\hat{p_{i}}\mathbb{E}\left\{ (Y-\gamma_{0})^{2}|\hat{\mathbf{J}},\hat{\mathbf{X}}\right\} \hat{p_{i}}^{\top}\right]\hat{\Sigma}^{-1}\\
 & \lesssim_{\text{p.s.d.}}\hat{1}\hat{\Sigma}^{-1}\left[\frac{1}{n^{2}}\sum_{i=1}^{n}\hat{p_{i}}\hat{p}_{i}^{\top}\right]\hat{\Sigma}^{-1} &  & \text{Assm \ref{assu:Regularity}}\\
 & =\frac{\hat{1}}{n}\hat{\Sigma}^{-1}\hat{\Sigma}\hat{\Sigma}^{-1}\\
 & =\frac{\hat{1}}{n}\hat{\Sigma}^{-1}.
\end{align*}
Similarly, for Point \ref{enu:DDtilde},
\begin{align*}
\lambda_{\max}\left[\mathbb{E}\left(\tilde{H}_{\alpha\bcd}\tilde{H}_{\alpha\bcd}^{\top}\right)\right] & =\frac{1}{n}\lambda_{\max}\left[\mathbb{E}\left(\tilde{H}_{\alpha i}\tilde{H}_{\alpha i}^{\top}\right)\right] &  & \text{Eq \eqref{eq:vvtsum}}\\
 & =\frac{1}{n}\lambda_{\max}\left[\mathbb{E}\left\{ \left(v_{q}-\tilde{p}_{i}\alpha_{0}(\tilde{X}_{i})\right)\left(v_{q}-\tilde{p}_{i}\alpha_{0}(\tilde{X}_{i})\right)^{\top}\right\} \right]\\
 & \lesssim\frac{1}{n}\lambda_{\max}\left[\mathbb{E}\left\{ v_{q}v_{q}^{\top}\right\} \right]+\frac{1}{n}\lambda_{\max}\left[\mathbb{E}\left\{ \tilde{p}_{i}\alpha_{0}(\tilde{X}_{i})^{2}\tilde{p}_{i}^{\top}\right\} \right] &  & \text{Lemma \ref{lem:sq-of-sum}}\\
 & \lesssim\frac{1}{n}+\frac{1}{n}\lambda_{\max}\left[\mathbb{E}\left\{ \tilde{p}_{i}\tilde{p}_{i}^{\top}\right\} \right] &  & \text{Assms \ref{assu:vvt} \& \ref{assu:Regularity}.}\\
 & \lesssim1/n.
\end{align*}

Point \ref{enu:ddiff-gam} follows from Points \ref{enu:D-gam} \&
\ref{enu:D-gam-star} and Eq (\ref{eq:Delta-sum-gam}):
\[
\mathbb{E}\left(\hat{1}||\hat{\delta}_{\gamma}-\delta_{\gamma}^{\star}||_{2}^{2}\right)=\mathbb{E}\left(\hat{1}||\hat{\Delta}_{\gamma}-\hat{\Delta}_{\gamma}^{\star}||_{2}^{2}\right)\leq\mathbb{E}\left(\hat{1}||\hat{\Delta}_{\gamma}||_{2}^{2}\right)+\mathbb{E}\left(\hat{1}||\hat{\Delta}_{\gamma}^{\star}||_{2}^{2}\right)\lesssim k_{n}/n.
\]
Similarly, Point \ref{enu:ddiff-alpha} follows from Points \ref{enu:D-alpha}
\& \ref{enu:D-alpha-star} and Eq (\ref{eq:Delta-sum-alpha}):
\[
\mathbb{E}\left(\tilde{1}||\tilde{\delta}_{\alpha}-\delta_{\alpha}^{\star}||_{2}^{2}\right)=\mathbb{E}\left(\tilde{1}||\tilde{\Delta}_{\alpha}-\Delta_{\alpha}^{\star}||_{2}^{2}\right)\leq\mathbb{E}\left(\tilde{1}||\tilde{\Delta}_{\alpha}||_{2}^{2}\right)+\mathbb{E}\left(\tilde{1}||\Delta_{\alpha}^{\star}||_{2}^{2}\right)\lesssim k_{n}/n.
\]
For Points \ref{enu:gam-hat-star} \& \ref{enu:alpha-tilde-star},
\begin{align*}
\mathbb{E}\left\{ \hat{1}(\hat{\gamma}(\bar{X})-\gamma^{\star}(\bar{X}))^{2}\right\}  & =\mathbb{E}\left[\hat{1}\left(q(\bar{X})^{\top}(\hat{\delta}_{\gamma}-\delta_{\gamma}^{\star})\right)^{2}\right]\\
 & =\mathbb{E}\left[\hat{1}\left(\hat{\delta}_{\gamma}-\delta_{\gamma}^{\star}\right)^{\top}\mathbb{E}\left(q(\bar{X})q(\bar{X})^{\top}|\hat{\delta}_{\gamma},\hat{1}\right)\left(\hat{\delta}_{\gamma}-\delta_{\gamma}^{\star}\right)\right]\\
 & \lesssim\mathbb{E}\left[\hat{1}||\hat{\delta}_{\gamma}-\delta_{\gamma}^{\star}||_{2}^{2}{}^{\top}\right] &  & \text{Lemma \ref{lem:similar-dist}}\\
 & \lesssim k_{n}/n,
\end{align*}
and, following the same steps,
\begin{align*}
\mathbb{E}\left\{ \tilde{1}(\tilde{\alpha}(\bar{X})-\alpha^{\star}(\bar{X}))^{2}\right\}  & =\mathbb{E}\left\{ \tilde{1}\left(q(\bar{X})^{\top}(\tilde{\delta}_{\alpha}-\delta_{\alpha}^{\star})\right)^{2}\right\} \lesssim k_{n}/n.
\end{align*}
\end{proof}
Next, we introduce a bounds on $\bar{\ddot{w}}(\bar{C}_{i})$ and
its conditional moments. 
\begin{lem}
\label{lem:w-bounds}Under Assumptions \ref{assu:series-estimates},
\ref{assu:(Approximately-uniform-),}, \& \ref{assu:standardization},
for any $m\geq1$, we have\textcolor{teal}{}
\begin{enumerate}
\item $\bar{\ddot{1}}|\bar{\ddot{w}}(\bar{C}_{i})|\lesssim r_{n}K_{b}(\ddot{C},\bar{C}_{i})$;
\label{enu:max-val-w}
\item $\mathbb{E}\left[\bar{\ddot{1}}\left|\bar{\ddot{w}}(\bar{C}_{i})\right|^{m}|\ddot{C}\right]\lesssim r_{n}^{m-1}$;
\label{enu:w|ddot}
\item $\mathbb{E}\left[\bar{\ddot{1}}\left|\bar{\ddot{w}}(\bar{C}_{i})\right|^{m}|\bar{C}_{i}\right]\lesssim r_{n}^{m-1}$;\label{enu:w|bar}
\item $\mathbb{E}\left[\bar{\ddot{1}}\left|\bar{\ddot{w}}(\bar{C}_{i})\bar{\ddot{w}}(\bar{C}_{j})\right||\ddot{C}\right]\lesssim1$
for any $i,j\in\{1,\dots,n\}$ such that $i\neq j$\label{enu:w-ij|ddot};
and
\item $\mathbb{E}\left[\bar{\ddot{1}}\left|\bar{\ddot{w}}(\bar{C}_{i})\bar{\ddot{w}}(\bar{C}_{j})\right||\bar{\mathbf{Z}}\right]\lesssim r_{n}$
for any $i,j\in\{1,\dots,n\}$ such that $i\neq j$.\label{enu:w-ij|bar}
\end{enumerate}
\end{lem}
\begin{proof}
For Point \ref{enu:max-val-w}, the definition of $K_{b}$ implies
that $b(c)^{\top}\bar{\text{\textbf{B}}}^{-1}b(c')=0$ whenever $K_{b}(c,c')=0$.
Thus,

\begin{align*}
\bar{\ddot{1}}|\bar{\ddot{w}}(\bar{C}_{i})| & =\bar{\ddot{1}}\left|b(\ddot{C})^{\top}\bar{\text{\textbf{B}}}^{-1}b(\bar{C}_{i})\right| &  & \text{Def of }\bar{\ddot{w}}\text{ (Assm \ref{assu:series-estimates})}\\
 & =\bar{\ddot{1}}\left|b(\ddot{C})^{\top}\bar{\text{\textbf{B}}}^{-1}b(\bar{C}_{i})K_{b}(\ddot{C},\bar{C}_{i})\right|\\
 & \lesssim\bar{\ddot{1}}||b(\ddot{C})||_{2}\left|\left|\bar{\text{\textbf{B}}}^{-1}b(\bar{C}_{i})K_{b}(\ddot{C},\bar{C}_{i})\right|\right|_{2} &  & \text{Cauchy-Schwartz}\\
 & \lesssim||b(\ddot{C})||_{2}\left|\left|b(\bar{C}_{i})\right|\right|_{2}K_{b}(\ddot{C},\bar{C}_{i}) &  & \text{Def of }\bar{\ddot{1}}\\
 & \lesssim r_{n}K_{b}(\ddot{C},\bar{C}_{i}) &  & \text{Assm \ref{assu:standardization}.}
\end{align*}
For Points \ref{enu:w|ddot} \& \ref{enu:w|bar}, applying Assumption
\ref{assu:(Approximately-uniform-),} gives gives
\begin{align*}
\mathbb{E}\left[\bar{\ddot{1}}\left|\bar{\ddot{w}}(\bar{C}_{i})\right|^{m}|\ddot{C}\right] & \lesssim r_{n}^{m}\mathbb{E}\left[K_{b}(\ddot{C},\bar{C}_{i})|\ddot{C}\right]\lesssim r_{n}^{m-1},
\end{align*}
and
\begin{align*}
\mathbb{E}\left[\bar{\ddot{1}}\left|\bar{\ddot{w}}(\bar{C}_{i})\right|^{m}|\bar{\mathbf{A}},\bar{\mathbf{X}}\right]\lesssim r_{n}^{m}\mathbb{E}\left[K_{b}(\ddot{C},\bar{C}_{i})|\bar{C}_{i}\right]\lesssim r_{n}^{m-1}.
\end{align*}
For Point \ref{enu:w-ij|ddot}, 
\begin{align*}
\mathbb{E}\left[\bar{\ddot{1}}\left|\bar{\ddot{w}}(\bar{C}_{i})\bar{\ddot{w}}(\bar{C}_{j})\right||\ddot{C}\right] & \lesssim r_{n}^{2}\mathbb{E}\left[K_{b}(\ddot{C},\bar{C}_{i})K_{b}(\ddot{C},\bar{C}_{j})\;\;|\ddot{C}\right]\\
 & =r_{n}^{2}\mathbb{E}\left[K_{b}(\ddot{C},\bar{C}_{i})\;|\ddot{C}\right]\mathbb{E}\left[K_{b}(\ddot{C},\bar{C}_{j})\;\;|\ddot{C}\right]\\
 & \lesssim1.
\end{align*}
Finally, for Point \ref{enu:w-ij|bar},
\begin{align*}
\mathbb{E}\left[\bar{\ddot{1}}\left|\bar{\ddot{w}}(\bar{C}_{i})\bar{\ddot{w}}(\bar{C}_{j})\right|\;|\bar{\mathbf{C}}\right] & \lesssim r_{n}^{2}\mathbb{E}\left[K_{b}(\ddot{C},\bar{C}_{i})|\bar{\mathbf{C}}\right]\lesssim r_{n}.
\end{align*}
\end{proof}
\textcolor{black}{The next Lemma provides a general strategy for studying
weighted averages. It is based on \citet{Kennedy2020-ie}'s Lemma
2 for }\textcolor{black}{\emph{unweighted}}\textcolor{black}{{} averages,
but takes advantage of Lemma \ref{lem:w-bounds} to obtain a more
specific result.}
\begin{lem}
\textcolor{black}{\label{lem:PnW-convergence} (Bound for randomly
weighted sums) Let $\hat{\tilde{g}}$ be a random function estimated
from data splits $\hat{\mathbf{Z}}$ and $\tilde{\mathbf{Z}}$. If
$\mathbb{E}\left\{ \hat{\tilde{g}}(\bar{Z}_{i})|\bar{C}_{i},\hat{\tilde{g}}\right\} =0$
and }the assumptions of Lemma \ref{lem:w-bounds} hold,\textcolor{black}{{}
then
\[
\bar{\ddot{1}}\bar{\ddot{\mathbb{P}}}_{n}\left(\hat{\tilde{g}}(Z)\right)\lesssim_{\mathbb{P}}\sqrt{\frac{1}{n/r_{n}}\mathbb{E}\left(\hat{\tilde{g}}(Z)^{2}\right)}.
\]
}
\end{lem}
\begin{proof}
We will bound the second moment of $\bar{\ddot{1}}\bar{\ddot{\mathbb{P}}}_{n}\left(\hat{\tilde{g}}(Z)\right)$
and apply Markov's Inequality. To start,
\begin{align}
\mathbb{E}\left[\bar{\ddot{1}}\bar{\ddot{\mathbb{P}}}_{n}\left(\hat{\tilde{g}}(Z)\right)^{2}\right] & =\mathbb{E}\left[\frac{\bar{\ddot{1}}}{n^{2}}\left\{ \sum_{i=1}^{n}\bar{\ddot{w}}(\bar{C}_{i})\left(\hat{\tilde{g}}(\bar{Z}_{i})\right)\right\} ^{2}\right]\nonumber \\
 & \leq\frac{1}{n^{2}}\mathbb{E}\left[\bar{\ddot{1}}\sum_{i=1}^{n}\bar{\ddot{w}}(\bar{C}_{i})^{2}\hat{\tilde{g}}(\bar{Z}_{i})^{2}\right]+\frac{1}{n^{2}}\sum_{i\neq j}\mathbb{E}\left[\bar{\ddot{1}}\bar{\ddot{w}}(\bar{C}_{i})\bar{\ddot{w}}(\bar{C}_{j})\hat{\tilde{g}}(\bar{Z}_{i})\hat{\tilde{g}}(\bar{Z}_{j})\right]\label{eq:g-2parts}
\end{align}

For the off-diagonal terms, iterating expectations over $\bar{\mathbf{C}}$
and $\hat{\tilde{g}}$ gives
\begin{align*}
\mathbb{E}\left[\bar{\ddot{1}}\bar{\ddot{w}}(\bar{C}_{i})\bar{\ddot{w}}(\bar{C}_{j})\hat{\tilde{g}}(\bar{Z}_{i})\hat{\tilde{g}}(\bar{Z}_{j})\right] & =\mathbb{E}\left[\mathbb{E}\left\{ \bar{\ddot{1}}\bar{\ddot{w}}(\bar{C}_{i})\bar{\ddot{w}}(\bar{C}_{j})|\bar{\mathbf{C}}\right\} \mathbb{E}\left\{ \hat{\tilde{g}}(\bar{Z}_{i})|\bar{C}_{i},\hat{\tilde{g}}\right\} \mathbb{E}\left\{ \hat{\tilde{g}}(\bar{Z}_{j})|\bar{C}_{j},\hat{\tilde{g}}\right\} \right]\\
 & =0.
\end{align*}

Plugging this into Eq (\ref{eq:g-2parts}) gives

\begin{align*}
\mathbb{E}\left[\bar{\ddot{1}}\bar{\ddot{\mathbb{P}}}_{n}\left(\hat{\tilde{g}}(Z)\right)^{2}\right] & =\frac{1}{n}\mathbb{E}\left[\bar{\ddot{1}}\bar{\ddot{w}}(\bar{C})^{2}\hat{\tilde{g}}(\bar{Z})^{2}\right]+0\\
 & =\frac{1}{n}\mathbb{E}\left[\mathbb{E}\left(\bar{\ddot{1}}\bar{\ddot{w}}(\bar{C})^{2}|\bar{C}\right)\mathbb{E}\left(\hat{\tilde{g}}(\bar{Z})^{2}|\bar{C}\right)\right]\\
 & \lesssim\frac{r_{n}}{n}\mathbb{E}\left[\mathbb{E}\left(\hat{\tilde{g}}(\bar{Z})^{2}|\bar{C}\right)\right] &  & \text{Lemma \ref{lem:w-bounds}.\ref{enu:w|bar}}\\
 & =\frac{r_{n}}{n}\mathbb{E}\left(\hat{\tilde{g}}(\bar{Z})^{2}\right).
\end{align*}
Markov's Inequality then shows the result.
\end{proof}
Next, we show a Lemma for the \emph{weighted} moments of $p$.
\begin{lem}
\label{lem:(Weighted-moments-p}(Weighted moments of $p$) Under the
assumptions of Lemma \ref{lem:w-bounds},
\begin{enumerate}
\item $\lambda_{\max}\left\{ \mathbb{E}\left(\bar{\ddot{1}}\bar{\ddot{\Sigma}}\right)\right\} \lesssim1$
,\label{enu:-Esig} and
\item $\lambda_{\max}\left\{ \mathbb{E}\left(\bar{\ddot{1}}\bar{\ddot{\Sigma}}^{2}\right)\right\} \lesssim r_{n}$.
\label{enu:Sig2_pr}
\end{enumerate}
\end{lem}
\begin{proof}
For Point \ref{enu:-Esig},
\begin{align*}
\lambda_{\text{max}}\text{\ensuremath{\left\{ \mathbb{E}\left(\bar{\ddot{1}}\ddot{\Sigma}\right)\right\} } } & =\lambda_{\text{max}}\left\{ \mathbb{E}\left(\frac{\bar{\ddot{1}}}{n}\sum_{i=1}^{n}\bar{p}_{i}\bar{\ddot{w}}(\bar{C}_{i})\bar{p}_{i}^{\top}\right)\right\} \\
 & \leq\lambda_{\text{max}}\left\{ \mathbb{E}\left(\frac{1}{n}\sum_{i=1}^{n}\bar{p}_{i}\mathbb{E}\left(\bar{\ddot{1}}|\bar{\ddot{w}}(\bar{C}_{i})|\;|\;\bar{\mathbf{J}},\bar{\mathbf{X}}\right)\bar{p}_{i}^{\top}\right)\right\} \\
 & \lesssim\lambda_{\text{max}}\left\{ \mathbb{E}\left(\frac{1}{n}\sum_{i=1}^{n}\bar{p}_{i}\bar{p}_{i}^{\top}\right)\right\}  &  & \text{Lemma \ref{lem:w-bounds}.\ref{enu:w|bar}}\\
 & =1.
\end{align*}

For Point \ref{enu:Sig2_pr}, let $\kappa$ be a constant such that,
from Lemma \ref{lem:w-bounds}.\ref{enu:max-val-w}, $\bar{\ddot{1}}|\bar{\ddot{w}}_{i}|\leq\bar{\ddot{1}}\kappa r_{n}$.
First we will show that $\bar{\ddot{1}}\left(\bar{\Sigma}r_{n}\kappa-\bar{\ddot{\Sigma}}\right)$
is p.s.d. For any vector $v$, we have
\begin{align}
\bar{\ddot{1}}v^{\top}\left(\kappa r_{n}\bar{\Sigma}-\bar{\ddot{\Sigma}}\right)v & =v^{\top}\left(\frac{\bar{\ddot{1}}}{n}\sum_{i=1}^{n}\bar{p}_{i}\left(\kappa r_{n}-\bar{\ddot{w}}_{i}\right)\bar{p}_{i}^{\top}\right)v\nonumber \\
 & \geq\frac{\bar{\ddot{1}}}{n}\sum_{i=1}^{n}\left(\bar{p}_{i}^{\top}v\right)^{2}\left(\kappa r_{n}-|\bar{\ddot{w}}_{i}|\right)\nonumber \\
 & \geq0.\label{eq:S-Sig-psd}
\end{align}

Next, let $\bar{\ddot{\Sigma}}^{1/2}$ be the square root of $\bar{\ddot{\Sigma}}$
if $\bar{\ddot{1}}=1$ and the identity matrix otherwise. For any
unit-norm vector $v$, we have 

\begin{align*}
v^{\top}\mathbb{E}\left(\bar{\ddot{1}}\bar{\ddot{\Sigma}}^{2}\right)v & \leq\kappa r_{n}v^{\top}\mathbb{E}\left(\bar{\ddot{1}}\bar{\ddot{\Sigma}}^{1/2}\bar{\Sigma}\bar{\ddot{\Sigma}}^{1/2}\right)v &  & \text{from Eq \eqref{eq:S-Sig-psd}}\\
 & \leq\kappa r_{n}v^{\top}\mathbb{E}\left[\bar{\ddot{1}}\bar{\ddot{\Sigma}}\lambda_{\max}\left(\bar{\Sigma}\right)\right]v\\
 & \lesssim r_{n}v^{\top}\mathbb{E}\left[\bar{\ddot{1}}\bar{\ddot{\Sigma}}\right]v &  & \text{Def of }\bar{\ddot{1}}\\
 & \lesssim r_{n} &  & \text{from Point \ref{enu:-Esig}}.
\end{align*}
\end{proof}
A commonly used corollary of Markov's Inequality is that $V\lesssim_{\mathbb{P}}\mathbb{E}(V^{2})^{1/2}$
for any random variable $V$. As we show in the next lemma, Markov's
Inequality can also be used in combination with \emph{conditional
}expectations.
\begin{lem}
\label{lem:conditional-markovs}Let $A_{n}$ be a scalar random variable;
let $B_{n}$ be a random variable, vector or matrix; and let $c_{n}$
be a sequence of constants. If there exists a function $f(B_{n})\geq0$
such that 
\[
\mathbb{E}\left(A_{n}^{2}|B_{n}\right)\lesssim f(B_{n})\lesssim_{\mathbb{P}}c_{n},
\]
 then $A_{n}^{2}\lesssim_{\mathbb{P}}c_{n}.$ 
\end{lem}
See also Lemma 2 of \citet{Kennedy2022-da} and Lemma 6.1 of \citet{Chernozhukov2018-va}
for analogous results.
\begin{proof}
We know there is a constant $m_{1}$ such that $\mathbb{E}\left(A_{n}^{2}|B_{n}\right)\leq m_{1}f(B_{n})$.
Thus, Markov's inequality implies that, for any $\epsilon$,
\[
\text{Pr}\left(A_{n}^{2}>\frac{2m_{1}f(B_{n})}{\epsilon}|B_{n}\right)\leq\frac{\epsilon\mathbb{E}\left(A_{n}^{2}|B_{n}\right)}{2m_{1}f(B_{n})}\leq\epsilon/2.
\]
Since $f(B_{n})\lesssim_{\mathbb{P}}c_{n}$, we also know that there
exists $m_{2}$ and $n'$ such that

\[
\text{Pr}\left(f(B_{n})>m_{2}c_{n}\right)\leq\epsilon/2.
\]
for all $n\geq n'$. Thus, for all $n\geq n'$
\begin{align*}
\text{Pr}\left(A_{n}^{2}>\frac{2m_{1}m_{2}c_{n}}{\epsilon}\right) & =\text{Pr}\left(A_{n}^{2}>\frac{2m_{1}m_{2}c_{n}}{\epsilon}\hspace{1em}\&\hspace{1em}f(B_{n})\leq m_{2}c_{n}\right)\\
 & \hspace{1em}+\text{Pr}\left(A_{n}^{2}>\frac{2m_{1}m_{2}c_{n}}{\epsilon}\hspace{1em}\&\hspace{1em}f(B_{n})>m_{2}c_{n}\right)\\
 & \leq\text{Pr}\left(A_{n}^{2}>\frac{2m_{1}f(B_{n})}{\epsilon}\right)\\
 & \hspace{1em}+\epsilon/2\\
 & =\mathbb{E}\left[\text{Pr}\left(A_{n}^{2}>\frac{2m_{1}f(B_{n})}{\epsilon}|B_{n}\right)\right]+\epsilon/2\\
 & \leq\epsilon.
\end{align*}
\end{proof}
Finally, we include a result relating sums of squares to squares of
sums.
\begin{lem}
(Squares of sums) \label{lem:sq-of-sum} For any finite sequence of
vectors $a_{1},\dots a_{k}$, the matrix $k\sum_{i=1}^{k}a_{i}a_{i}^{\top}-\left(\sum_{i=1}^{k}a_{i}\right)\left(\sum_{j=1}^{k}a_{j}\right)^{\top}$is
positive semidefinite. In particular, when the vectors are of dimension
1, we have $k\sum_{i=1}^{k}a_{i}^{2}-\left(\sum_{i=1}^{k}a_{i}\right)^{2}\geq0$.
\end{lem}
\begin{proof}
We have
\begin{align}
k\sum_{i=1}^{k}a_{i}a_{i}^{\top}-\left(\sum_{i=1}^{k}a_{i}\right)\left(\sum_{j=1}^{k}a_{j}\right)^{\top} & =\sum_{i=1}^{k}\sum_{j=1}^{k}a_{i}a_{i}^{\top}-\sum_{i=1}^{k}\sum_{j=1}^{k}a_{i}a_{j}^{\top}\nonumber \\
 & =\frac{1}{2}\sum_{i=1}^{k}\sum_{j=1}^{k}\left(a_{i}a_{i}^{\top}-2a_{i}a_{j}^{\top}+a_{j}a_{j}^{\top}\right)\nonumber \\
 & =\frac{1}{2}\sum_{i=1}^{k}\sum_{j=1}^{k}(a_{i}-a_{j})(a_{i}-a_{j})^{\top},\label{eq:psd_aa}
\end{align}
Since each summand $(a_{i}-a_{j})(a_{i}-a_{j})^{\top}$ is p.s.d.,
Line (\ref{eq:psd_aa}) is p.s.d. as well.
\end{proof}

\subsection{\label{subsec:Proof-apply-linear-ie-common-terms}Reducing Lines
(\ref{eq:apply-linear}), \& (\ref{eq:apply-IE})}
\begin{proof}
Lines (\ref{eq:apply-linear}) \& (\ref{eq:apply-IE}) can be written
as
\begin{align}
 & \bar{\ddot{1}}\left[\bar{\ddot{\mathbb{P}}}_{n}\left\{ \hat{g}_{1}(Z)+\tilde{g}_{2}(Z)\right\} \right],\label{eq:Pgs}
\end{align}
where 
\begin{align*}
\hat{g}_{1}(Z) & :=\hat{1}\left\{ m(Z,\hat{\gamma})-m(Z,\gamma_{0})-J\alpha_{0}(\hat{\gamma}-\gamma_{0})\right\}  &  & \text{(Line \eqref{eq:apply-linear}),\ensuremath{} and}\\
\tilde{g}_{2}(Z) & :=\tilde{1}\left\{ \left(\tilde{\alpha}-\alpha_{0}\right)J(Y-\gamma_{0})\right\}  &  & \text{(Line \eqref{eq:apply-IE})}.
\end{align*}

We will bound these terms using Lemma \ref{lem:PnW-convergence}.
To apply Lemma \ref{lem:PnW-convergence}, we first note that $\mathbb{E}\left(\hat{g}_{1}(\bar{Z}_{i})|\bar{C}_{i},\hat{g}_{1}\right)=0$
from the linearity property (Eq (\ref{eq:linear-form})). Similarly,
$\mathbb{E}\left(\tilde{g}_{2}(\bar{Z}_{i})|\bar{C}_{i},\tilde{g}_{2}\right)=0$
from iterating expectations over $\bar{X}_{i}$ and $\bar{J}_{i}$. 

Next, we bound the second moment of $\hat{g}_{1}$ and $\tilde{g}_{2}$.
For $\hat{g}_{1}(\bar{Z})$,
\begin{align*}
 & \mathbb{E}\left[\hat{g}_{1}(\bar{Z})^{2}\right]\\
 & \hspace{1em}=\mathbb{E}\left[\hat{1}\left\{ m(\bar{Z},\hat{\gamma})-m(\bar{Z},\gamma_{0})-\bar{J}\alpha_{0}(\hat{\gamma}-\gamma_{0})\right\} ^{2}\right]\\
 & \hspace{1em}\lesssim\mathbb{E}\left\{ \hat{1}\mathbb{E}\left[\left\{ m(\bar{Z},\hat{\gamma})-m(\bar{Z},\gamma_{0})\right\} ^{2}|\hat{\mathbf{Z}}\right]\right\} +\mathbb{E}\left[\hat{1}\alpha_{0}^{2}\bar{J}\left(\hat{\gamma}-\gamma_{0}\right)^{2}\right] &  & \text{Lemma \ref{lem:sq-of-sum}}\\
 & \hspace{1em}\lesssim\sum_{i=1}^{T}\mathbb{E}\mathbb{E}\left[\hat{1}\left(\hat{\gamma}-\gamma_{0}\right)^{2}\;|\;\hat{\mathbf{Z}}\right] &  & \text{Assms \ref{assu:Regularity} \& \ref{assu:m-convergence}}\\
 & \hspace{1em}\leq\sum_{i=1}^{T}\mathbb{E}\left[\hat{1}\left(\hat{\gamma}-\gamma^{\star}\right)^{2}\right]+\mathbb{E}\left[\hat{1}\left(\gamma^{\star}-\gamma_{0}\right)^{2}\right]\\
 & \hspace{1em}\lesssim k_{n}/n+k_{n}^{-2s_{\gamma}/d_{X}} &  & \text{Lemma \ref{lem:(Delta-values)}.\ref{enu:gam-hat-star} \& Eq \eqref{eq:nuisance-error-sup}.}
\end{align*}
Following similar steps for $\tilde{g}_{2}(\bar{Z})$, 
\begin{align*}
\mathbb{E}\left[\tilde{g}_{2}\left(\bar{Z}\right){}^{2}\right] & =\mathbb{E}\left(\tilde{1}\left(\tilde{\alpha}-\alpha_{0}\right)^{2}\bar{J}(\bar{Y}-\gamma_{0})^{2}\right)\\
 & \lesssim\mathbb{E}\left(\tilde{1}\left(\tilde{\alpha}-\alpha_{0}\right)^{2}\bar{J}\right) &  & \text{Assm \ref{assu:Regularity}}\\
 & \leq\mathbb{E}\left(\tilde{1}\left(\tilde{\alpha}-\alpha^{\star}\right)^{2}\right)+\mathbb{E}\left(\left(\alpha^{\star}-\alpha_{0}\right)^{2}\right)\\
 & \lesssim k_{n}/n+k_{n}^{-2s_{\alpha}/d_{X}} &  & \text{Lemma \ref{lem:(Delta-values)}.\ref{enu:alpha-tilde-star} and Eq \eqref{eq:nuisance-error-sup}}.
\end{align*}

Lemma \ref{lem:PnW-convergence} now tells us that 
\begin{align}
\bar{\ddot{\mathbb{P}}}_{n}\left\{ \hat{g}_{1}(Z)+\tilde{g}_{2}(Z)\right\}  & \lesssim_{\mathbb{P}}\frac{1}{\sqrt{n/r_{n}}}\left(\sqrt{\frac{k_{n}}{n}}+k_{n}^{-s_{\gamma}/d_{X}}+k_{n}^{-s_{\alpha}/d_{X}}\right).\label{eq:pgs2}
\end{align}
\end{proof}

\subsection{\label{subsec:Reducing-Line_a_star}Reducing Line (\ref{eq:a_star_0-g_hat_star})}
\begin{proof}
To study Line (\ref{eq:a_star_0-g_hat_star}), we first introduce
notation for the bias of $\hat{\gamma}$. Let $\hat{w}_{\gamma}(x,x'):=q(x)^{\top}\hat{\Sigma}^{-1}q(x')$
be OLS weights for $\hat{\gamma}$ as a function of $\hat{\mathbf{Z}}$,
so that 
\[
\hat{\gamma}(x)=q(x)\hat{\Sigma}^{-1}\frac{1}{n}\sum_{i=1}^{n}q(\hat{X}_{i})\hat{Y}_{i}=\frac{1}{n}\sum_{i=1}^{n}\hat{w}_{\gamma}(x,\hat{X}_{i})\hat{Y}_{i}.
\]
Following the same steps as in Lemma \ref{lem:w-bounds}, we can see
that 
\begin{equation}
\hat{1}|\hat{w}_{\gamma}(x,x')|\lesssim k_{n}K_{q}(x,x')\text{\ensuremath{\hspace{1em}\text{and}\hspace{1em}}\ensuremath{\mathbb{E}\left\{  \hat{1}\hat{w}_{\gamma}(x,X)^{2}\right\} } \ensuremath{\lesssim k_{n}}.}\label{eq:w-gam-results}
\end{equation}
It follows that
\begin{align}
\mathbb{E}\left[\hat{1}Var\left\{ \hat{\gamma}\left(x\right)|\hat{\mathbf{X}},\hat{\mathbf{J}}\right\} \right] & =\frac{1}{n}\sum_{i=1}^{n}\mathbb{E}\left[\hat{1}w_{\gamma}(x,\hat{X}_{i})^{2}Var\left\{ \hat{Y}_{i}|\hat{X}_{i},\hat{J}_{i}\right\} \right]\lesssim\frac{k_{n}}{n}.\label{eq:var-gam-given-x}
\end{align}

For any fixed $x$, let $\gamma^{\dagger}(x;\hat{\mathbf{X}},\hat{\mathbf{J}}):=\mathbb{E}(\hat{\gamma}(x)|\hat{\mathbf{X}},\hat{\mathbf{J}})$.
Since $\hat{\gamma}$ is estimated with a spline regression with a
known outcome $Y$, its bias ($\gamma^{\dagger}(x)-\gamma_{0}(x)$)
has the same structure as the that of the ``oracle'' spline estimator
in Theorem \ref{thm:(Oracle-Error)-Under}. Thus, following the same
steps as in Eq (\ref{eq:spline-bias}), below, we can see that
\begin{align*}
\hat{1}|\gamma^{\dagger}(x,\hat{\mathbf{X}},\hat{\mathbf{J}})-\gamma_{0}(x)| & \lesssim\hat{1}k_{n}^{-s_{\gamma}/d_{C}}\frac{1}{n}\sum_{i=1}^{n}|\hat{w}_{\gamma}(x,\hat{X}_{i})|\\
 & \lesssim k_{n}^{-s_{\gamma}/d_{C}}\frac{k_{n}}{n}\sum_{i=1}^{n}K_{q}(x,\hat{X}_{i}) &  & \text{from Eq \eqref{eq:w-gam-results}}.
\end{align*}
From the triangle inequality and Eq (\ref{eq:nuisance-error-sup}),
we can then see that 
\[
\hat{1}\left|\gamma^{\dagger}(x)-\gamma^{\star}(x)\right|\leq\hat{1}\left|\gamma^{\dagger}(x)+\gamma_{0}(x)\right|+\left|\gamma_{0}(x)-\gamma^{\star}(x)\right|\lesssim k_{n}^{-s_{\gamma}/d_{X}}\left(1+\frac{k_{n}}{n}\sum_{i=1}^{n}K_{q}(x,\hat{X}_{i})\right).
\]

We are now ready to return to the task of bounding Line (\ref{eq:a_star_0-g_hat_star}).
Adding and subtracting $\gamma^{\dagger}$ produces
\begin{equation}
\hat{1}\bar{\ddot{1}}\bar{\ddot{\mathbb{P}}}_{n}\left\{ (\alpha^{\star}-\alpha_{0})J(\hat{\gamma}-\gamma^{\star})\right\} =\hat{1}\bar{\ddot{1}}\bar{\ddot{\mathbb{P}}}_{n}\left\{ (\alpha^{\star}-\alpha_{0})J(\hat{\gamma}-\gamma^{\dagger})\right\} +\hat{1}\bar{\ddot{1}}\bar{\ddot{\mathbb{P}}}_{n}\left\{ (\alpha^{\star}-\alpha_{0})J(\gamma^{\dagger}-\gamma^{\star})\right\} .\label{eq:2-gam-hats}
\end{equation}

For the second term in Eq (\ref{eq:2-gam-hats}), we have
\begin{align}
 & \hat{1}\bar{\ddot{1}}\bar{\ddot{\mathbb{P}}}_{n}\left((\alpha^{\star}-\alpha_{0})J(\gamma^{\dagger}-\gamma^{\star})\right)\nonumber \\
 & \hspace{1em}\leq\frac{1}{n}\sum_{i=1}^{n}\bar{\ddot{1}}|\bar{\ddot{w}}(\bar{C}_{i})|\times|\alpha^{\star}-\alpha_{0}|\times\hat{1}|\gamma^{\dagger}-\gamma^{\star}|\nonumber \\
 & \hspace{1em}\lesssim k_{n}^{-(s_{\alpha}+s_{\gamma})/d_{X}}\frac{1}{n}\sum_{i=1}^{n}\bar{\ddot{1}}|\bar{\ddot{w}}(\bar{C}_{i})|\times\left\{ 1+\frac{k_{n}}{n}\sum_{i'=1}^{n}K_{q}(\bar{X}_{i},\hat{X}_{i'})\right\}  &  & \text{from Eq \eqref{eq:nuisance-error-sup}}\nonumber \\
 & \hspace{1em}\lesssim_{\mathbb{P}}k_{n}^{-(s_{\alpha}+s_{\gamma})/d_{X}}\frac{1}{n}\sum_{i=1}^{n}\mathbb{E}\left[\bar{\ddot{1}}|\bar{\ddot{w}}(\bar{C}_{i})|\times\left\{ 1+\frac{k_{n}}{n}\sum_{i'=1}^{n}\mathbb{E}\left\{ K_{q}(\bar{X}_{i},\hat{X}_{i'})|\bar{X}_{i}\right\} \right\} \right] &  & \text{Markov's Inequality}\nonumber \\
 & \hspace{1em}\lesssim_{\mathbb{P}}k_{n}^{-(s_{\alpha}+s_{\gamma})/d_{X}}\frac{1}{n}\sum_{i=1}^{n}\mathbb{E}\left[\bar{\ddot{1}}|\bar{\ddot{w}}(\bar{C}_{i})|\right] &  & \text{Assm \ref{assu:(Approximately-uniform-),}}\nonumber \\
 & \hspace{1em}\lesssim k_{n}^{-(s_{\alpha}+s_{\gamma})/d_{X}} &  & \text{Lemma \ref{lem:w-bounds}.\ref{enu:w|ddot}.}\label{eq:product-of-biases}
\end{align}
Conditional on $\bar{\mathbf{Z}},\hat{\mathbf{X}}$, $\hat{\mathbf{J}}$
and $\ddot{C}$, the first term in Eq (\ref{eq:2-gam-hats}) has mean
zero and variance\textcolor{red}{}
\begin{align}
 & Var\left[\frac{\hat{1}\bar{\ddot{1}}}{n}\sum_{i=1}^{n}\bar{\ddot{w}}\left(\bar{C}_{i}\right)\left\{ \alpha^{\star}\left(\bar{X}_{i}\right)-\alpha_{0}\left(\bar{X}_{i}\right)\right\} \bar{J}_{i}\left\{ \hat{\gamma}(\bar{X}_{i})-\gamma^{\dagger}(\bar{X}_{i})\right\} |\bar{\mathbf{Z}},\hat{\mathbf{X}},\hat{\mathbf{J}},\ddot{C}\right]\nonumber \\
 & =\frac{1}{n^{2}}\sum_{i=1}^{n}\sum_{i'=1}^{n}\left[\bar{\ddot{1}}\bar{\ddot{w}}\left(\bar{C}_{i}\right)\bar{\ddot{w}}\left(\bar{C}_{i'}\right)\bar{J}_{i}\bar{J}_{i'}\right.\nonumber \\
 & \hspace{1em}\hspace{1em}\hspace{1em}\;\hspace{1em}\times\left\{ \alpha^{\star}\left(\bar{X}_{i}\right)-\alpha_{0}\left(\bar{X}_{i}\right)\right\} \left\{ \alpha^{\star}\left(\bar{X}_{i'}\right)-\alpha_{0}\left(\bar{X}_{i'}\right)\right\} \nonumber \\
 & \hspace{1em}\hspace{1em}\hspace{1em}\;\hspace{1em}\left.\times\hat{1}Cov\left\{ \hat{\gamma}\left(\bar{X}_{i}\right),\hat{\gamma}\left(\bar{X}_{i'}\right)|\bar{\mathbf{Z}},\hat{\mathbf{X}},\hat{\mathbf{J}},\ddot{C}\right\} \right].\nonumber \\
 & \lesssim\frac{k^{-2s_{\alpha}/d_{X}}}{n^{2}}\sum_{i=1}^{n}\sum_{i'=1}^{n}\left|\bar{\ddot{1}}\bar{\ddot{w}}\left(\bar{C}_{i}\right)\bar{\ddot{w}}\left(\bar{C}_{i'}\right)\right|\times\left|\hat{1}Cov\left\{ \hat{\gamma}\left(\bar{X}_{i}\right),\hat{\gamma}\left(\bar{X}_{i'}\right)|\bar{\mathbf{Z}},\hat{\mathbf{X}},\hat{\mathbf{J}}\right\} \right|.\label{eq:summand_ij-cov}
\end{align}
Above, we applied the fact that $Var\left\{ \gamma^{\dagger}\left(\bar{X}_{j};\hat{\mathbf{X}}\right)|\bar{\mathbf{Z}},\hat{\mathbf{X}},\hat{\mathbf{J}},\ddot{C}\right\} =0$.
For the off-diagonal summation terms in (\ref{eq:summand_ij-cov})
where $i\neq i'$ we have
\begin{align}
 & \left|\hat{1}Cov\left\{ \hat{\gamma}\left(\bar{X}_{i}\right),\hat{\gamma}\left(\bar{X}_{i'}\right)|\bar{\mathbf{X}},\hat{\mathbf{X}},\hat{\mathbf{J}}\right\} \right|\nonumber \\
 & =\left|\hat{1}Cov\left\{ \hat{\gamma}\left(\bar{X}_{i}\right),\hat{\gamma}\left(\bar{X}_{i'}\right)|\bar{\mathbf{X}},\hat{\mathbf{X}},\hat{\mathbf{J}}\right\} K_{q}(\bar{X}_{i},\bar{X}_{i'})\right|\label{eq:diff-train}\\
 & \leq\hat{1}Var\left\{ \hat{\gamma}\left(\bar{X}_{i}\right)|\bar{\mathbf{X}},\hat{\mathbf{X}},\hat{\mathbf{J}}\right\} ^{1/2}Var\left\{ \hat{\gamma}\left(\bar{X}_{i}\right)|\bar{\mathbf{X}},\hat{\mathbf{X}},\hat{\mathbf{J}}\right\} ^{1/2}K_{q}(\bar{X}_{i},\bar{X}_{i'}), &  & \text{Cauchy Schwartz}\nonumber \\
 & \lesssim\frac{k_{n}}{n}\mathbb{E}\left[K_{q}(\bar{X}_{i},\bar{X}_{i'})\right] &  & \text{from Eq \eqref{eq:var-gam-given-x}}\nonumber \\
 & \lesssim1/n &  & \text{Assumption \eqref{assu:(Approximately-uniform-),}.}\nonumber 
\end{align}
Above, the equality in Line (\ref{eq:diff-train}) comes from the
fact that the predictions for points in different neighborhoods use
different training data subsets, and are therefore uncorrelated (as
noted by \citealp{Kennedy2022-da}). Thus, the off-diagonal summation
terms in Eq (\ref{eq:summand_ij-cov}) are collectively 
\begin{align*}
 & \lesssim_{\mathbb{P}}\frac{k^{-2s_{\alpha}/d_{X}}}{n^{2}}\sum_{i=1}^{n}\mathbb{E}\left[\left|\bar{\ddot{1}}\bar{\ddot{w}}\left(\bar{C}_{i}\right)\bar{\ddot{w}}\left(\bar{C}_{i'}\right)\right|\times\left|\hat{1}Cov\left\{ \hat{\gamma}\left(\bar{X}_{i}\right),\hat{\gamma}\left(\bar{X}_{i'}\right)|\bar{\mathbf{X}},\hat{\mathbf{X}}\right\} \right|\right]\\
 & \lesssim\frac{k^{-2s_{\alpha}/d_{X}}}{n^{2}/r_{n}} &  & \text{Lemma \ref{lem:w-bounds}.\ref{enu:w-ij|bar}}\\
 & <\frac{k^{-2s_{\alpha}/d_{X}}}{n/r_{n}}.
\end{align*}
Similarly, the diagonal summation terms in Eq (\ref{eq:summand_ij-cov})
are 
\begin{align*}
 & \frac{k^{-2s_{\alpha}/d_{X}}}{n^{2}}\sum_{i=1}^{n}\bar{\ddot{1}}\bar{\ddot{w}}\left(\bar{C}_{i}\right)^{2}\hat{1}Var\left\{ \hat{\gamma}\left(\bar{X}_{i}\right)|\bar{\mathbf{X}},\hat{\mathbf{X}}\right\} \\
 & \lesssim_{\mathbb{P}}\frac{1}{n}k^{-2s_{\alpha}/d_{X}}\mathbb{E}\left[\mathbb{E}\left[\bar{\ddot{1}}\bar{\ddot{w}}\left(\bar{C}_{i}\right)^{2}|\bar{\mathbf{X}},\hat{\mathbf{X}}\right]\;\;\hat{1}Var\left\{ \hat{\gamma}\left(\bar{X}_{i}\right)|\bar{\mathbf{X}},\hat{\mathbf{X}}\right\} \right]\\
 & \lesssim\frac{1}{n}k^{-2s_{\alpha}/d_{X}}r_{n}\mathbb{E}\left[\hat{1}Var\left\{ \hat{\gamma}\left(\bar{X}_{i}\right)|\bar{\mathbf{X}},\hat{\mathbf{X}}\right\} \right]\\
 & \lesssim\frac{k^{-2s_{\alpha}/d_{X}}}{n^{2}/r_{n}}\\
 & <\frac{k^{-2s_{\alpha}/d_{X}}}{n/r_{n}}.
\end{align*}
Thus, from Markov's Inequality (Lemma \ref{lem:conditional-markovs}),
\[
\hat{1}\bar{\ddot{1}}\bar{\ddot{\mathbb{P}}}_{n}\left\{ (\alpha^{\star}-\alpha_{0})J(\hat{\gamma}-\gamma^{\dagger})\right\} \lesssim_{\mathbb{P}}\frac{k^{-s_{\alpha}/d_{X}}}{\sqrt{n/r_{n}}}.
\]
This, combined with Eqs (\ref{eq:2-gam-hats}) \& (\ref{eq:product-of-biases}),
implies
\[
\hat{1}\bar{\ddot{1}}\bar{\ddot{\mathbb{P}}}_{n}\left\{ (\alpha^{\star}-\alpha_{0})J(\hat{\gamma}-\gamma^{\star})\right\} \lesssim_{\mathbb{P}}k_{n}^{-(s_{\alpha}+s_{\gamma})/d_{X}}+\frac{k^{-s_{\alpha}/d_{X}}}{\sqrt{n/r_{n}}}.
\]
\end{proof}

\subsection{\label{subsec:Reducing-Line_g_star}Reducing Line (\ref{eq:a_tilde_star-g_star_0})}

In the special case of the conditional covariance, Line (\ref{eq:a_tilde_star-g_star_0})
can be studied in the same way as Line (\ref{eq:a_star_0-g_hat_star}),
since $\tilde{\alpha}$ is a spline regression with a known outcome
$A$. Following the same steps as in Section \ref{subsec:Reducing-Line_a_star},
we can obtain 
\[
\bar{\ddot{\mathbb{P}}}_{n}\left\{ (\tilde{\alpha}-\alpha^{\star})J(\gamma^{\star}-\gamma_{0})\right\} \lesssim_{\mathbb{P}}k^{-(s_{\gamma}+s_{\alpha})/d_{X}}+\frac{k_{n}^{-s_{\gamma}/d_{X}}}{\sqrt{n/r_{n}}}.
\]

In the more general case of estimating any linear function (Assumption
\ref{assu:lin-func}) where $\tilde{\alpha}$ is defined using a moment
condition, studying $\tilde{\alpha}-\alpha^{\star}$ is more challenging
than studying $\hat{\gamma}-\gamma^{\star}$. To proceed, we first
introduce a hypothetical comparator $\tilde{\alpha}_{\text{oracle}}(x):=q(x)^{\top}\tilde{\Sigma}^{-1}\tilde{\mathbb{P}}_{n}\left(p\alpha_{0}\right)$,
which is the version of $\tilde{\alpha}(x)$ that would have occurred
if we \emph{could} fit a traditional spline regression against $\alpha_{0}$.
By adding and subtracting relevant terms, we can re-express $\tilde{\alpha}-\alpha^{\star}$
in terms of $\tilde{\alpha}_{\text{oracle}}$. We have
\begin{align}
\tilde{\alpha}(\bar{X}_{i})-\alpha^{\star}(\bar{X}_{i}) & =q(\bar{X}_{i})^{\top}\tilde{\Sigma}^{-1}\tilde{\mathbb{P}}_{n}\left(v_{q}\right)-\alpha^{\star}(\bar{X}_{i}) &  & \text{by def. of \ensuremath{\tilde{\alpha}}}\nonumber \\
 & =q(\bar{X}_{i})^{\top}\tilde{\Sigma}^{-1}\tilde{\mathbb{P}}_{n}\left(v_{q}-p\alpha_{0}\right)\nonumber \\
 & \hspace{1em}+q(\bar{X}_{i})^{\top}\tilde{\Sigma}^{-1}\tilde{\mathbb{P}}_{n}\left(p\alpha_{0}\right)-\alpha^{\star}(\bar{X}_{i}).\label{eq:alpha-break1}
\end{align}
To abbreviate Eq (\ref{eq:alpha-break1}), let $\tilde{\phi}:=\tilde{\Sigma}^{-1}\tilde{\mathbb{P}}_{n}\left(v_{q}-p\alpha_{0}\right)$,
so that 

\begin{equation}
\tilde{\alpha}(\bar{X}_{i})-\alpha^{\star}(\bar{X}_{i})=q(\bar{X}_{i})^{\top}\tilde{\phi}+\left\{ \tilde{\alpha}_{\text{oracle}}(\bar{X}_{i})-\alpha^{\star}(\bar{X}_{i})\right\} .\label{eq:a-tilde-star-breakdown}
\end{equation}

From here 
\begin{align}
\tilde{1}\bar{\ddot{1}}\bar{\ddot{\mathbb{P}}}_{n}\left\{ (\tilde{\alpha}-\alpha^{\star})J(\gamma^{\star}-\gamma_{0})\right\}  & =\bar{\ddot{1}}\bar{\ddot{\mathbb{P}}}_{n}\left\{ \tilde{1}(\gamma^{\star}-\gamma_{0})p^{\top}\tilde{\phi}\right\} \label{eq:psi1-term}\\
 & \hspace{1em}+\bar{\ddot{1}}\bar{\ddot{\mathbb{P}}}_{n}\left[\tilde{1}(\gamma^{\star}-\gamma_{0})J\left\{ \tilde{\alpha}_{\text{oracle}}(\bar{X}_{i})-\alpha^{\star}(\bar{X}_{i})\right\} \right]\label{eq:same-as-gam}
\end{align}

Line (\ref{eq:same-as-gam}) is $\lesssim_{\mathbb{P}}k_{n}^{-(s_{\alpha}+s_{\gamma})/d_{X}}+k_{n}^{-s_{\gamma}/d_{X}}/\sqrt{n/r_{n}}$
for the same reasons as in Section (\ref{subsec:Reducing-Line_a_star}).
For Line \ref{eq:psi1-term}, 
\begin{align*}
\mathbb{E}\left[\tilde{1}\left\{ (\gamma^{\star}-\gamma_{0})p^{\top}\tilde{\phi}\right\} ^{2}|\tilde{\mathbf{Z}}\right] & =\tilde{1}\tilde{\phi}^{\top}\mathbb{E}\left\{ p(\gamma^{\star}-\gamma_{0})^{2}p^{\top}\right\} \tilde{\phi}\\
 & \lesssim\tilde{1}k^{-2s_{\gamma}/d_{X}}\tilde{\phi}^{\top}\mathbb{E}\left\{ pp^{\top}\right\} \tilde{\phi} &  & \text{from Eq \eqref{eq:nuisance-error-sup}}\\
 & \lesssim\tilde{1}k^{-2s_{\gamma}/d_{X}}\tilde{\phi}^{\top}\tilde{\phi}\\
 & \lesssim k^{-2s_{\gamma}/d_{X}}\|\tilde{\mathbb{P}}_{n}\left(v_{q}-p\alpha_{0}\right)\|^{2}\\
 & \lesssim k^{1-2s_{\gamma}/d_{X}}/n.
\end{align*}
Above, the last line comes from the fact that $\tilde{\mathbb{P}}_{n}\left(v_{q}-p\alpha_{0}\right)$
is an $k_{n}$-length vector of sample averages, each with mean zero.

Combining results gives

\[
\bar{\ddot{\mathbb{P}}}_{n}\left\{ (\tilde{\alpha}-\alpha^{\star})J(\gamma^{\star}-\gamma_{0})\right\} \lesssim_{\mathbb{P}}k^{-(s_{\gamma}+s_{\alpha})/d_{X}}+\frac{k_{n}^{-s_{\gamma}/d_{X}}}{\sqrt{n/r_{n}}}+\frac{k^{1/2-s_{\gamma}/d_{X}}}{\sqrt{n}}.
\]

\subsection{\label{subsec:Reducing-2stars}Reducing Line \ref{eq:2stars}}
\begin{proof}
Since $(\gamma^{\star}-\gamma_{0})\lesssim k_{n}^{-s_{\gamma}/d_{X}}$
from Eq \ref{eq:nuisance-error-sup}, we see that 
\[
\bar{\ddot{\mathbb{P}}}_{n}\left((\alpha^{\star}-\alpha_{0})J(\gamma^{\star}-\gamma_{0})\right)\lesssim_{\mathbb{P}}k_{n}^{-2(s_{\gamma}+s_{\gamma})/d_{X}}
\]
by following the same steps as in Eq (\ref{eq:product-of-biases}),
with $(\gamma^{\star}-\gamma_{0})$ replacing $(\gamma^{\dagger}-\gamma^{\star})$
throughout, and omitting the $\frac{k_{n}}{n}\sum_{i'=1}^{n}K_{q}(\bar{X}_{i},\hat{X}_{i'})$
term.

\textcolor{purple}{}
\end{proof}

\subsection{\label{subsec:Proof-of-last-2way-term}Reducing Line \ref{eq:diff-term}
for 2-way estimators}
\begin{proof}
Our remaining task is to bound $\hat{1}\tilde{1}\bar{\ddot{1}}\bar{\ddot{\mathbb{P}}}_{n}\left((\tilde{\alpha}-\alpha^{\star})J(\hat{\gamma}-\gamma^{\star})\right)$.
We have
\begin{align}
\bar{\ddot{\mathbb{P}}}_{n}\left((\tilde{\alpha}-\alpha^{\star})J(\hat{\gamma}-\gamma^{\star})\right) & =\frac{1}{n}\sum_{i=1}^{n}\left(\text{\ensuremath{\tilde{\delta}_{\alpha}}}-\text{\ensuremath{\delta_{\alpha}^{\star}}}\right)^{\top}\bar{p}_{i}\bar{\ddot{w}}(\bar{C}_{i})\bar{p}_{i}^{\top}\left(\text{\ensuremath{\hat{\delta}_{\gamma}}}-\ensuremath{\delta_{\gamma}^{\star}}\right)\nonumber \\
 & =\left(\text{\ensuremath{\tilde{\delta}_{\alpha}}}-\text{\ensuremath{\delta_{\alpha}^{\star}}}\right)^{\top}\bar{\ddot{\Sigma}}\left(\text{\ensuremath{\hat{\delta}_{\gamma}}}-\ensuremath{\delta_{\gamma}^{\star}}\right)\nonumber \\
 & =\left(\tilde{\Delta}_{\alpha}+\tilde{\Delta}_{\alpha}^{\star}\right)^{\top}\bar{\ddot{\Sigma}}\left(\hat{\Delta}_{\gamma}+\hat{\Delta}_{\gamma}^{\star}\right) &  & \text{Lines \eqref{eq:Delta-sum-gam} \& \eqref{eq:Delta-sum-alpha}}\nonumber \\
 & =\tilde{\Delta}_{\alpha}^{\top}\bar{\ddot{\Sigma}}\hat{\Delta}_{\gamma}+\tilde{\Delta}_{\alpha}^{\top}\bar{\ddot{\Sigma}}\hat{\Delta}_{\gamma}^{\star}+\tilde{\Delta}_{\alpha}^{\star^{\top}}\bar{\ddot{\Sigma}}\hat{\Delta}_{\gamma}+\tilde{\Delta}_{\alpha}^{\star^{\top}}\bar{\ddot{\Sigma}}\hat{\Delta}_{\gamma}^{\star}.\label{eq:DDterms}
\end{align}
Let $\bar{\ddot{\Sigma}}^{1/2}$ be the square root of $\bar{\ddot{\Sigma}}$
if $\bar{\ddot{1}}=1$, and the identity matrix otherwise. For any
two vectors $V,U\perp\bar{\ddot{\Sigma}}$, 
\begin{align}
\bar{\ddot{1}}|V^{\top}\bar{\ddot{\Sigma}}U| & =\bar{\ddot{1}}|V^{\top}\bar{\ddot{\Sigma}}^{1/2}\bar{\ddot{\Sigma}}^{1/2}U|\leq\bar{\ddot{1}}\left|\left|\bar{\ddot{\Sigma}}^{1/2}V\right|\right|_{2}\left|\left|\bar{\ddot{\Sigma}}^{1/2}U\right|\right|_{2}\label{eq:vsu1}
\end{align}
where Lemma \ref{lem:(Weighted-moments-p} implies that
\begin{equation}
\bar{\ddot{1}}\left|\left|\bar{\ddot{\Sigma}}^{1/2}V\right|\right|_{2}^{2}\lesssim_{\mathbb{P}}\mathbb{E}\left[\bar{\ddot{1}}\left|\left|\bar{\ddot{\Sigma}}^{1/2}V\right|\right|_{2}^{2}\right]=\mathbb{E}\left[V^{\top}\mathbb{E}\left[\bar{\ddot{1}}\bar{\ddot{\Sigma}}\right]V\right]\leq\mathbb{E}\left[||V||_{2}^{2}\right]\lambda_{\max}\left\{ \mathbb{E}\left(\bar{\ddot{1}}\bar{\ddot{\Sigma}}\right)\right\} \lesssim\mathbb{E}\left[||V||_{2}^{2}\right].\label{eq:vsu2}
\end{equation}
Eqs (\ref{eq:vsu1}), (\ref{eq:vsu2}), \& Lemma \ref{lem:(Delta-values)}
then show that

\begin{align}
 & \bar{\ddot{1}}\hat{1}\tilde{1}\left(\tilde{\Delta}_{\alpha}^{\top}\bar{\ddot{\Sigma}}\hat{\Delta}_{\gamma}+\tilde{\Delta}_{\alpha}^{\top}\bar{\ddot{\Sigma}}\hat{\Delta}_{\gamma}^{\star}+\tilde{\Delta}_{\alpha}^{\star^{\top}}\bar{\ddot{\Sigma}}\hat{\Delta}_{\gamma}+\tilde{\Delta}_{\alpha}^{\star^{\top}}\bar{\ddot{\Sigma}}\hat{\Delta}_{\gamma}^{\star}\right)\nonumber \\
 & \hspace{1em}\lesssim_{\mathbb{P}}\sqrt{\mathbb{E}\left(||\tilde{1}\tilde{\Delta}_{\alpha}||_{2}^{2}\right)\mathbb{E}\left(||\hat{1}\hat{\Delta}_{\gamma}||_{2}^{2}\right)}+\sqrt{\mathbb{E}\left(||\tilde{1}\tilde{\Delta}_{\alpha}||_{2}^{2}\right)\mathbb{E}\left(||\hat{1}\hat{\Delta}_{\gamma}^{\star}||_{2}^{2}\right)}\nonumber \\
 & \hspace{1em}\hspace{1em}+\sqrt{\mathbb{E}\left(||\tilde{1}\tilde{\Delta}_{\alpha}^{\star}||_{2}^{2}\right)\mathbb{E}\left(||\hat{1}\hat{\Delta}_{\gamma}||_{2}^{2}\right)}+\sqrt{\mathbb{E}\left(||\tilde{1}\tilde{\Delta}_{\alpha}^{\star}||_{2}^{2}\right)\mathbb{E}\left(||\hat{1}\hat{\Delta}_{\gamma}^{\star}||_{2}^{2}\right)}\\
 & \hspace{1em}\lesssim\frac{k_{n}}{n}+\sqrt{\frac{k_{n}}{n}\times\frac{k_{n}^{-2s_{\gamma}/d_{X}+1}}{n}}\nonumber \\
 & \hspace{1em}+\sqrt{\frac{k_{n}^{-2s_{\alpha}/d_{X}+1}}{n}\times\frac{k_{n}}{n}}+\sqrt{\frac{k_{n}^{-2s_{\alpha}/d_{X}+1}}{n}\times\frac{k_{n}^{-2s_{\gamma}/d_{X}+1}}{n}}\\
 & \hspace{1em}\lesssim k_{n}/n.\label{eq:2-way-terms-proof}
\end{align}

\textcolor{purple}{}
\end{proof}

\subsection{\label{subsec:Proof-of-last-term-3way}Reducing Line \ref{eq:diff-term}
for 3-way estimators}
\begin{proof}
For 3-way estimators we can update Eq (\ref{eq:DDterms}) with an
stronger bound. We have
\begin{align}
\hat{1}\tilde{1}\bar{\ddot{1}}\bar{\ddot{\mathbb{P}}}_{n}\left((\tilde{\alpha}-\alpha^{\star})J(\hat{\gamma}-\gamma^{\star})\right) & =\hat{1}\tilde{1}\bar{\ddot{1}}\left(\text{\ensuremath{\tilde{\delta}_{\alpha}}}-\text{\ensuremath{\delta_{\alpha}^{\star}}}\right)^{\top}\bar{\ddot{\Sigma}}\left(\text{\ensuremath{\hat{\delta}_{\gamma}}}-\ensuremath{\delta_{\gamma}^{\star}}\right)\nonumber \\
 & =\hat{1}\tilde{1}\bar{\ddot{1}}\left(\text{\ensuremath{\tilde{\delta}_{\alpha}}}-\text{\ensuremath{\delta_{\alpha}^{\star}}}\right)^{\top}\bar{\ddot{\Sigma}}\left(\hat{\Delta}_{\gamma}-\hat{\Delta}_{\gamma}^{\star}\right)\nonumber \\
 & =\hat{1}\tilde{1}\bar{\ddot{1}}\left[\left(\text{\ensuremath{\tilde{\delta}_{\alpha}}}-\text{\ensuremath{\delta_{\alpha}^{\star}}}\right)^{\top}\bar{\ddot{\Sigma}}\hat{\Delta}_{\gamma}+\tilde{\Delta}_{\alpha}^{\top}\bar{\ddot{\Sigma}}\hat{\Delta}_{\gamma}^{\star}+\tilde{\Delta}_{\alpha}^{\star^{\top}}\bar{\ddot{\Sigma}}\hat{\Delta}_{\gamma}^{\star}\right].\label{eq:DDterms-1}
\end{align}
We already know that the last term is $\tilde{\Delta}_{\alpha}^{\star^{\top}}\bar{\ddot{\Sigma}}\hat{\Delta}_{\gamma}^{\star}\lesssim_{\mathbb{P}}\frac{1}{n}k_{n}^{-s_{\alpha}/d_{X}-s_{\gamma}/d_{X}+1}\leq k_{n}^{-(s_{\alpha}+s_{\gamma})/d_{X}}$
by Eq (\ref{eq:2-way-terms-proof}). For the first term in Eq (\ref{eq:DDterms-1}),
we have 
\begin{align*}
 & \mathbb{E}\left[\bar{\ddot{1}}\tilde{1}\hat{1}\left(\left(\tilde{\delta}_{\alpha}-\delta_{\alpha}^{\star}\right)^{\top}\bar{\ddot{\Sigma}}\hat{\Delta}_{\gamma}\right)^{2}\right]\\
 & \hspace{1em}=\mathbb{E}\left[\bar{\ddot{1}}\tilde{1}\left(\tilde{\delta}_{\alpha}-\delta_{\alpha}^{\star}\right)^{\top}\bar{\ddot{\Sigma}}\mred{\mathbb{E}\left(\hat{1}\hat{\Delta}_{\gamma}\hat{\Delta}_{\gamma}^{\top}|\bar{\ddot{\Sigma}},\tilde{\textbf{Z}}\right)}\bar{\ddot{\Sigma}}\left(\tilde{\delta}_{\alpha}-\delta_{\alpha}^{\star}\right)\right]\\
 & \hspace{1em}\leq\frac{1}{n}\mathbb{E}\left[\bar{\ddot{1}}\tilde{1}\mred{\hat{1}}\left(\tilde{\delta}_{\alpha}-\delta_{\alpha}^{\star}\right)^{\top}\bar{\ddot{\Sigma}}\mred{\hat{\Sigma}^{-1}}\bar{\ddot{\Sigma}}\left(\tilde{\delta}_{\alpha}-\delta_{\alpha}^{\star}\right)\right] &  & \text{Lemma \ref{lem:(Delta-values)}.\ref{enu:DDhat}, \& }\,\hat{\mathbf{Z}}\ensuremath{\perp}\bar{\ddot{\Sigma}},\tilde{\textbf{Z}}\\
 & \hspace{1em}\leq\frac{1}{n}\mathbb{E}\left[\mred{2}\tilde{1}\left(\tilde{\delta}_{\alpha}-\delta_{\alpha}^{\star}\right)^{\top}\mathbb{E}\left(\bar{\ddot{1}}\bar{\ddot{\Sigma}}^{2}|\tilde{\mathbf{Z}}\right)\left(\tilde{\delta}_{\alpha}-\delta_{\alpha}^{\star}\right)\right] &  & \text{I.E. \& def. of }\hat{1}\\
 & \hspace{1em}\lesssim\frac{1}{n}\times\tilde{1}||\tilde{\delta}_{\alpha}-\delta_{\alpha}^{\star}||_{2}^{2}\times\lambda_{\max}\left\{ \mathbb{E}\left(\bar{\ddot{1}}\bar{\ddot{\Sigma}}^{2}\right)\right\} \\
 & \hspace{1em}\lesssim_{\mathbb{P}}\frac{1}{n}\times\mathbb{E}\left[\tilde{1}||\tilde{\delta}_{\alpha}-\delta_{\alpha}^{\star}||_{2}^{2}\right]\times r_{n} &  & \text{Markov's Ineq + Lemma \ref{lem:(Weighted-moments-p}.\ref{enu:Sig2_pr}}\\
 & \hspace{1em}\lesssim\frac{1}{n}\times\frac{k_{n}}{n}\times r_{n} &  & \text{Lemma \ref{lem:(Delta-values)}}.\ref{enu:ddiff-alpha}.
\end{align*}
So $\left(\tilde{\delta}_{\alpha}+\delta_{\alpha}^{\star}\right)^{\top}\bar{\ddot{\Sigma}}\hat{\Delta}_{\gamma}\lesssim_{\mathbb{P}}\frac{\sqrt{k_{n}}}{n/\sqrt{r_{n}}}$.

For the special case of estimating conditional covariances, the second
term in Eq (\ref{eq:DDterms-1}) satisfies
\begin{align*}
\mathbb{E}\left[\bar{\ddot{1}}\tilde{1}\hat{1}\left(\tilde{\Delta}_{\alpha}^{\top}\bar{\ddot{\Sigma}}\hat{\Delta}_{\gamma}^{\star}\right)^{2}\right] & =\mathbb{E}\left[\bar{\ddot{1}}\hat{1}\hat{\Delta}_{\gamma}^{\star^{\top}}\bar{\ddot{\Sigma}}\mathbb{E}\left(\tilde{1}\tilde{\Delta}_{\alpha}\tilde{\Delta}_{\alpha}^{\top}\right)\bar{\ddot{\Sigma}}\hat{\Delta}_{\gamma}^{\star}\right]\\
 & \lesssim\frac{1}{n}\mathbb{E}\left[\bar{\ddot{1}}\hat{1}\tilde{1}\hat{\Delta}_{\gamma}^{\star^{\top}}\bar{\ddot{\Sigma}}\tilde{\Sigma}\bar{\ddot{\Sigma}}\hat{\Delta}_{\gamma}^{\star}\right] &  & \text{as in Lemma \ref{lem:(Delta-values)}.\ref{enu:DDhat}}\\
 & \lesssim\frac{1}{n}\mathbb{E}\left[\hat{1}\hat{\Delta}_{\gamma}^{\star^{\top}}\mathbb{E}\left(\bar{\ddot{1}}\bar{\ddot{\Sigma}}^{2}\right)\hat{\Delta}_{\gamma}^{\star}\right] &  & \text{def. of }\tilde{1}\\
 & \lesssim\frac{r_{n}}{n}\mathbb{E}\left[\hat{1}\hat{\Delta}_{\gamma}^{\star^{\top}}\hat{\Delta}_{\gamma}^{\star}\right] &  & \text{from Lemma \ref{lem:(Weighted-moments-p}}\\
 & \lesssim\frac{r_{n}}{n^{2}}k_{n}^{-2s_{\gamma}/d_{X}-1} &  & \text{Lemma \ref{lem:(Delta-values)}.\ref{enu:D-gam-star}}\\
 & \lesssim\frac{r_{n}}{n}k_{n}^{-2s_{\gamma}/d_{X}}.
\end{align*}
Markov's inequality then implies that $\bar{\ddot{1}}\tilde{1}\hat{1}\left(\tilde{\Delta}_{\alpha}^{\top}\bar{\ddot{\Sigma}}\hat{\Delta}_{\gamma}^{\star}\right)^{2}\lesssim_{\mathbb{P}}\frac{k_{n}^{-s_{\gamma}/d_{X}}}{\sqrt{n/r_{n}}}$.
Combining results for the conditional covariance case, we have
\[
\bar{\ddot{\mathbb{P}}}_{n}\left((\tilde{\alpha}-\alpha^{\star})J(\hat{\gamma}-\gamma^{\star})\right)\lesssim_{\mathbb{P}}k_{n}^{-(s_{\alpha}+s_{\gamma})/d_{X}}+\frac{\sqrt{k_{n}}}{n/\sqrt{r_{n}}}+\frac{k^{-s_{\gamma}/d_{X}}}{\sqrt{n/r_{n}}}.
\]

Alternatively, for the more general case of estimating a linear functional
(Assumption \ref{assu:lin-func}), we can handle the $\tilde{\Delta}_{\alpha}^{\top}\bar{\ddot{\Sigma}}\hat{\Delta}_{\gamma}^{\star}$
term in Eq (\ref{eq:DDterms-1}) through Lemma (\ref{lem:Under-3-way-splitting}),
below. Combining results then gives
\begin{equation}
\bar{\ddot{\mathbb{P}}}_{n}\left((\tilde{\alpha}-\alpha^{\star})J(\hat{\gamma}-\gamma^{\star})\right)\lesssim_{\mathbb{P}}k_{n}^{-(s_{\alpha}+s_{\gamma})/d_{X}}+\frac{\sqrt{k_{n}}}{n/\sqrt{r_{n}}}+\frac{k^{-s_{\gamma}/d_{X}}}{\sqrt{n/r_{n}}}+\left\{ \frac{k_{n}\log\left(k_{n}\right)}{n}\times\frac{k^{1-s_{\gamma}/d_{X}}}{n}\right\} .\label{eq:3-way-bias}
\end{equation}
\end{proof}
\begin{lem}
\label{lem:Under-3-way-splitting} (Based on Theorem 8 from \citealp{Newey2018-da})
Under 3-way CF and the Assumptions of Lemmas \ref{lem:(Stable-eigenvalues)-Under},
\ref{lem:(Delta-values)} \& \ref{lem:(Weighted-moments-p},
\[
\hat{1}\tilde{1}\bar{\ddot{1}}\tilde{\Delta}_{\alpha}^{\top}\bar{\ddot{\Sigma}}\hat{\Delta}_{\gamma}^{\star}\lesssim_{\mathbb{P}}\frac{k^{-s_{\gamma}/d_{X}}}{\sqrt{n/r_{n}}}+\left(\frac{k^{1-s_{\gamma}/d_{X}}}{n}\times\frac{k_{n}\log\left(k_{n}\right)}{n}\right).
\]
\end{lem}
\begin{proof}
Letting $I$ be the identity matrix,
\begin{align}
\bar{\ddot{1}}\tilde{1}\hat{1}\tilde{\Delta}_{\alpha}^{\top}\bar{\ddot{\Sigma}}\hat{\Delta}_{\gamma}^{\star} & =\bar{\ddot{1}}\tilde{1}\hat{1}\tilde{\Delta}_{\alpha}^{\top}\left[(I-\tilde{\Sigma})\bar{\ddot{\Sigma}}(I-\hat{\Sigma})+\bar{\ddot{\Sigma}}\hat{\Sigma}+\tilde{\Sigma}\bar{\ddot{\Sigma}}-\tilde{\Sigma}\bar{\ddot{\Sigma}}\hat{\Sigma}\right]\hat{\Delta}_{\gamma}^{\star}\nonumber \\
 & =\bar{\ddot{1}}\tilde{1}\hat{1}\tilde{\Delta}_{\alpha}^{\top}\left[(I-\tilde{\Sigma})\bar{\ddot{\Sigma}}(I-\hat{\Sigma})+\bar{\ddot{\Sigma}}\hat{\Sigma}+\mred{\tilde{\Sigma}\bar{\ddot{\Sigma}}(I-\hat{\Sigma})}\right]\hat{\Delta}_{\gamma}^{\star}.\label{eq:Dcrossprod}
\end{align}
We will bound the conditional second moment of each term in Eq (\ref{eq:Dcrossprod})
and then apply Lemma \ref{lem:conditional-markovs} to show that each
term in Eq (\ref{eq:Dcrossprod}) is bounded in probability.

For the first term in Eq (\ref{eq:Dcrossprod}), 
\begin{align}
 & \bar{\ddot{1}}\tilde{1}\hat{1}\tilde{\Delta}_{\alpha}^{\top}(I-\tilde{\Sigma})\bar{\ddot{\Sigma}}(I-\hat{\Sigma})\hat{\Delta}_{\gamma}^{\star}\nonumber \\
 & \hspace{1em}\lesssim\bar{\ddot{1}}\tilde{1}\hat{1}\left|\left|\bar{\ddot{\Sigma}}^{1/2}(I-\tilde{\Sigma})\tilde{\Delta}_{\alpha}\right|\right|_{2}\times\left|\left|\bar{\ddot{\Sigma}}^{1/2}(I-\hat{\Sigma})\hat{\Delta}_{\gamma}^{\star}\right|\right|_{2}, &  & \text{}\label{eq:CS-sig-ddot}
\end{align}
where 
\begin{align}
\mathbb{E}\left[\bar{\ddot{1}}\tilde{1}\left(\bar{\ddot{\Sigma}}^{1/2}(I-\tilde{\Sigma})\tilde{\Delta}_{\alpha}\right)^{2}\;|\tilde{\mathbf{Z}}\right] & =\tilde{1}\left(\tilde{\Delta}_{\alpha}^{\top}(I-\tilde{\Sigma})\mathbb{E}\left[\bar{\ddot{1}}\bar{\ddot{\Sigma}}\right](I-\tilde{\Sigma})\tilde{\Delta}_{\alpha}\right)\nonumber \\
 & \lesssim\tilde{1}\tilde{\Delta}_{\alpha}^{\top}\tilde{\Delta}_{\alpha}\lambda_{\max}\left(1-\tilde{\Sigma}\right)^{2}\lambda_{\max}\left(\mathbb{E}\left(\bar{\ddot{1}}\bar{\ddot{\Sigma}}\right)\right)\nonumber \\
 & \lesssim\left(\tilde{1}\tilde{\Delta}_{\alpha}^{\top}\tilde{\Delta}_{\alpha}\right)\lambda_{\max}\left(I-\tilde{\Sigma}\right)^{2} &  & \text{Lemma \ref{lem:(Weighted-moments-p}}\nonumber \\
 & \lesssim_{\mathbb{P}}\frac{k_{n}}{n}\left(\frac{k_{n}\log k_{n}}{n}\right) &  & \text{Lemmas \ref{lem:(Stable-eigenvalues)-Under} \& \ref{lem:(Delta-values)}.}\label{eq:CS1}
\end{align}
Similarly,
\begin{equation}
\mathbb{E}\left[\bar{\ddot{1}}\hat{1}\left(\bar{\ddot{\Sigma}}^{1/2}(I-\hat{\Sigma})\hat{\Delta}_{\gamma}^{\star}\right)^{2}|\hat{\mathbf{Z}}\right]\lesssim\left(\hat{1}\hat{\Delta}_{\gamma}^{\star^{\top}}\hat{\Delta}_{\gamma}^{\star}\right)\lambda_{\max}\left(I-\hat{\Sigma}\right)^{2}\lesssim_{\mathbb{P}}\frac{k_{n}^{1-2s_{\gamma}/d_{X}}}{n}\left(\frac{k_{n}\log k_{n}}{n}\right).\label{eq:CS2}
\end{equation}
Combining Eqs (\ref{eq:CS1}) \& (\ref{eq:CS2}) with Lemma \ref{lem:conditional-markovs},
we see that the right-hand side of Eq \ref{eq:CS-sig-ddot} is $\lesssim_{\mathbb{P}}\frac{k_{n}^{1-s_{\gamma}/d_{X}}}{n}\left(\frac{k_{n}\log k_{n}}{n}\right).$

For the middle term in Eq (\ref{eq:Dcrossprod}), around $\bar{\ddot{\Sigma}}\hat{\Sigma}$,
let $\hat{R}_{i}=\gamma_{0}(\hat{X}_{i})-\gamma^{\star}(\hat{X}_{i})$.
We have\textcolor{teal}{}
\begin{align*}
\mathbb{E}\left[\tilde{1}\left(\tilde{\Delta}_{\alpha}^{\top}\bar{\ddot{\Sigma}}\hat{\Sigma}\hat{\Delta}_{\gamma}^{\star}\right)^{2}\;|\tilde{\mathbf{Z}}\right] & =\mathbb{E}\left[\bar{\ddot{1}}\tilde{1}\left(\tilde{\Delta}_{\alpha}^{\top}\bar{\ddot{\Sigma}}\hat{\Sigma}\hat{\Sigma}^{-}\hat{H}_{\gamma\bcd}^{\star}\right)^{2}\;|\tilde{\mathbf{Z}},\bar{\ddot{\Sigma}}\right] &  & \text{Def of }\hat{\Delta}_{\gamma}^{\star}\\
 & =\tilde{1}\tilde{\Delta}_{\alpha}^{\top}\mathbb{E}\left[\bar{\ddot{1}}\bar{\ddot{\Sigma}}\hat{H}_{\gamma\bcd}^{\star}\hat{H}_{\gamma\bcd}^{\star^{\top}}\;\bar{\ddot{\Sigma}}\right]\tilde{\Delta}_{\alpha}\\
 & =\frac{1}{n}\tilde{1}\tilde{\Delta}_{\alpha}^{\top}\mathbb{E}\left[\bar{\ddot{1}}\bar{\ddot{\Sigma}}\mathbb{E}\left(\hat{H}_{\gamma i}^{\star}\hat{H}_{\gamma i}^{\star^{\top}}\right)\bar{\ddot{\Sigma}}\right]\tilde{\Delta}_{\alpha} &  & \text{I.E. + Eq \eqref{eq:vvtsum}}\\
 & =\frac{\tilde{1}}{n}\tilde{\Delta}_{\alpha}^{\top}\mathbb{E}\left[\bar{\ddot{1}}\bar{\ddot{\Sigma}}\mathbb{E}\left(\hat{R}_{i}^{2}\hat{p}_{i}\hat{p}_{i}^{\top}\right)\bar{\ddot{\Sigma}}\right]\tilde{\Delta}_{\alpha} &  & \text{Def of \ensuremath{\hat{R}_{i}}}\\
 & \lesssim\frac{\tilde{1}}{n}k_{n}^{-2s_{\gamma}/d_{X}}\tilde{\Delta}_{\alpha}^{\top}\mathbb{E}\left[\bar{\ddot{1}}\bar{\ddot{\Sigma}}\mathbb{E}\left(\hat{p}_{i}\hat{p}_{i}^{\top}\right)\bar{\ddot{\Sigma}}\right]\tilde{\Delta}_{\alpha} &  & \text{Eq \eqref{eq:nuisance-error-sup}}\\
 & =\frac{\tilde{1}}{n}k_{n}^{-2s_{\gamma}/d_{X}}\tilde{\Delta}_{\alpha}^{\top}\mathbb{E}\left(\bar{\ddot{1}}\bar{\ddot{\Sigma}}^{2}\right)\tilde{\Delta}_{\alpha}\\
 & \lesssim_{\mathbb{P}}\frac{1}{n^{2}}k^{1-2s_{\gamma}/d_{X}}r_{n} &  & \text{Lemmas \ref{lem:(Delta-values)}.\ref{enu:D-alpha}, \ref{lem:(Weighted-moments-p}.\ref{enu:Sig2_pr} \& Markov's Ineq.}
\end{align*}
Thus, from Lemma \ref{lem:conditional-markovs} \& Assumption \ref{assu:(Matrix-Bernstein-Inequality)},
we have $\tilde{\Delta}_{\alpha}^{\top}\bar{\ddot{\Sigma}}\hat{\Sigma}\hat{\Delta}_{\gamma}^{\star}\lesssim_{\mathbb{P}}\frac{k^{1/2-s_{\gamma}/d_{X}}}{n/\sqrt{r_{n}}}<\frac{k^{-s_{\gamma}/d_{X}}}{\sqrt{n/r_{n}}}.$

For the last term in Eq (\ref{eq:Dcrossprod}), 
\begin{align*}
 & \hat{1}\mathbb{E}\left[\tilde{1}\bar{\ddot{1}}\left(\tilde{\Delta}_{\alpha}^{\top}\tilde{\Sigma}\bar{\ddot{\Sigma}}(I-\hat{\Sigma})\hat{\Delta}_{\gamma}^{\star}\right)^{2}|\hat{\mathbf{Z}}\right]\\
 & =\hat{1}\hat{\Delta}_{\gamma}^{\star^{\top}}(I-\hat{\Sigma})\mathbb{E}\left[\bar{\ddot{1}}\tilde{1}\bar{\ddot{\Sigma}}\tilde{\Sigma}\mred{\tilde{\Delta}_{\alpha}\tilde{\Delta}_{\alpha}^{\top}}\tilde{\Sigma}\bar{\ddot{\Sigma}}\right](I-\hat{\Sigma})\hat{\Delta}_{\gamma}^{\star}\\
 & =\hat{1}\hat{\Delta}_{\gamma}^{\star^{\top}}(I-\hat{\Sigma})\mathbb{E}\left[\bar{\ddot{1}}\mred{\tilde{1}}\bar{\ddot{\Sigma}}\tilde{\Sigma}\mred{\tilde{\Sigma}^{-1}\tilde{H}_{\alpha\bcd}\tilde{H}_{\alpha\bcd}^{\top}\tilde{\Sigma}^{-1}}\tilde{\Sigma}\bar{\ddot{\Sigma}}\right](I-\hat{\Sigma})\hat{\Delta}_{\gamma}^{\star}\\
 & \lesssim\mred{\frac{\hat{1}}{n}}\hat{\Delta}_{\gamma}^{\star^{\top}}(I-\hat{\Sigma})\mathbb{E}\left(\bar{\ddot{1}}\bar{\ddot{\Sigma}}^{2}\right)(I-\hat{\Sigma})\hat{\Delta}_{\gamma}^{\star} &  & \text{Lemma \ref{lem:(Delta-values)}.\ref{enu:DDtilde}}\\
 & \lesssim\frac{\hat{1}}{n}\hat{\Delta}_{\gamma}^{\star^{\top}}\hat{\Delta}_{\gamma}^{\star}\lambda_{\max}\left(I-\hat{\Sigma}\right)^{2}\lambda_{\max}\left(\mathbb{E}\left(\bar{\ddot{1}}\bar{\ddot{\Sigma}}^{2}\right)\right) &  & \text{Lemma \ref{lem:(Weighted-moments-p}}\\
 & \lesssim_{\mathbb{P}}\frac{1}{n}\times\frac{k^{1-2s_{\gamma}/d_{X}}}{n}\times\frac{k_{n}\log\left(k_{n}\right)}{n}\times r_{n} &  & \text{Lemmas \ref{lem:(Stable-eigenvalues)-Under} \& \ref{lem:(Delta-values)}.\ref{enu:D-gam-star}}\\
 & \leq\frac{k^{-2s_{\gamma}/d_{X}}}{n/r_{n}}. &  & \text{Assm \ref{assu:(Matrix-Bernstein-Inequality)}.}
\end{align*}
Thus, from Lemma \ref{lem:conditional-markovs}, $\hat{1}\tilde{1}\tilde{\Delta}_{\alpha}^{\top}\tilde{\Sigma}\bar{\ddot{\Sigma}}(I-\hat{\Sigma})\hat{\Delta}_{\gamma}^{\star}\lesssim_{\mathbb{P}}\frac{k^{-s_{\gamma}/d_{X}}}{\sqrt{n/r_{n}}}$.
Combining all of the above results completes the proof.
\end{proof}

\section{\label{sec:proof-Oracle-error}Proof of Theorem \ref{thm:(Oracle-Error)-Under}
(oracle error)}

To bound the oracle error, we follow the same steps as Lemma 1 of
\citet{Kennedy2022-da}, and Proposition 1.12 of \citealp{Tsybakov2009-yb}.
We separately consider the bias and variance of the oracle estimator.

As a preliminary, we review a ``reproducing property'' of spline
estimators, analogous to Proposition 1.12 of \citealp{Tsybakov2009-yb}
for local polynomial (LP) estimators: if $g$ is a $\lfloor s_{\psi}\rfloor$-degree
polynomial, then $\frac{1}{n}\sum_{i=1}^{n}\bar{\ddot{w}}(\bar{C}_{i})g(\bar{C}_{i})=g(\ddot{C}$).
To see this, let $\beta_{g}$ be a set of coefficients such that $\beta_{g}^{\top}b(c)=g(c)$.
Since $b$ contains several neighborhoods, there are multiple valid
choices for $\beta_{s,c}$. Any choice will suffice. We have
\begin{align*}
g(\ddot{C})=b(\ddot{C})^{\top}\beta_{g}= & b(\ddot{C})^{\top}\bar{\text{\textbf{B}}}^{-1}\frac{1}{n}\sum_{i=1}^{n}b(\bar{C}_{i})b(\bar{C}_{i})^{\top}\beta_{g}\\
= & \frac{1}{n}\sum_{i=1}^{n}\bar{\ddot{w}}(\bar{C}_{i})b(\bar{C}_{i})^{\top}\beta_{g}\\
= & \frac{1}{n}\sum_{i=1}^{n}\bar{\ddot{w}}(\bar{C}_{i})g(\bar{C}_{i}).
\end{align*}

It follows that $\frac{1}{n}\sum_{i=1}^{n}\bar{\ddot{w}}(\bar{C}_{i})=1$;
and 
\begin{equation}
\frac{1}{n}\sum_{i=1}^{n}\bar{\ddot{w}}(\bar{C}_{i})\psi_{s_{\psi},\ddot{C}}(\bar{C}_{i})=\psi_{s_{\psi},\ddot{C}}(\ddot{C})=\psi(\ddot{C}),\label{eq:expand-around}
\end{equation}
where $\psi_{s_{\psi},\ddot{C}}$ is the $\lfloor s_{\psi}\rfloor$
order Taylor approximation of $\psi$ at $\ddot{C}$, and the second
equality comes from the fact that the approximation is exact at $\ddot{C}$. 

From here, we show that the root mean squared error $\mathbb{E}\left[\left\{ \bar{\ddot{\theta}}_{\text{oracle}}-\psi(\ddot{C})\right\} ^{2}|\ddot{C},\bar{\mathbf{C}}\right]^{1/2}$
converges to zero by showing a bound on the bias and variance of $\bar{\ddot{\theta}}_{\text{oracle}}$
conditional on $\ddot{C}$ and $\bar{\mathbf{C}}$. For the bias,
\begin{align}
 & \mathbb{E}\left(\left\{ \bar{\ddot{\theta}}_{\text{oracle}}-\psi(\ddot{C})\right\} |\ddot{C},\bar{\text{\ensuremath{\mathbf{C}}}}\right)\nonumber \\
 & =\frac{1}{n}\sum_{i=1}^{n}\bar{\ddot{w}}(\bar{C}_{i})\mathbb{E}\left\{ f_{0}(\bar{Z}_{i})|\bar{C}_{i}\right\} -\psi(\ddot{C}) &  & \text{Def of }\bar{\ddot{\theta}}_{\text{oracle}}\nonumber \\
 & =\frac{1}{n}\sum_{i=1}^{n}\bar{\ddot{w}}(\bar{C}_{i})\psi(\bar{C}_{i})-\psi(\ddot{C}) &  & \text{Def of }\psi\nonumber \\
 & =\frac{1}{n}\sum_{i=1}^{n}\bar{\ddot{w}}(\bar{C}_{i})\left\{ \psi(\bar{C}_{i})-\psi_{s_{\psi},\ddot{C}}(\bar{C}_{i})\right\} | &  & \text{Eq \eqref{eq:expand-around}}\nonumber \\
 & \lesssim\frac{1}{n}\sum_{i=1}^{n}|\bar{\ddot{w}}(\bar{C}_{i})|\times||\bar{C}_{i}-\ddot{C}||_{2}^{s_{\psi}} &  & \text{Assm \ref{assu:smooth-tau}}\nonumber \\
 & =\frac{1}{n}\sum_{i=1}^{n}|\bar{\ddot{w}}(\bar{C}_{i})|\times||\bar{C}_{i}-\ddot{C}||_{2}^{s_{\psi}}K_{b}\left(\bar{C}_{i},\ddot{C}\right) &  & \text{Def of \ensuremath{\bar{\ddot{w}}}}\nonumber \\
 & \leq r_{n}^{-s_{\psi}/d_{C}}\frac{1}{n}\sum_{i=1}^{n}|\bar{\ddot{w}}(\bar{C}_{i})| &  & \text{Assm \ref{assu:series-estimates}}\nonumber \\
 & \lesssim_{\mathbb{P}}r_{n}^{-s_{\psi}/d_{C}} &  & \text{Lemma \ref{lem:w-bounds}.\ref{enu:w|ddot} + Markov's Ineq.}\label{eq:spline-bias}
\end{align}
For the conditional variance,
\begin{align*}
Var\left[\bar{\ddot{\theta}}_{\text{oracle}}|\ddot{C},\bar{\text{\ensuremath{\mathbf{C}}}}\right] & =\frac{1}{n^{2}}\sum_{i=1}^{n}\bar{\ddot{w}}(\bar{C}_{i})^{2}Var\left[f_{0}(\bar{Z}_{i})|\bar{C}_{i}\right]\\
 & \lesssim\frac{1}{n^{2}}\sum_{i=1}^{n}\bar{\ddot{w}}(\bar{C}_{i})^{2} &  & \text{Assm \ref{assu:bounded-f0}}\\
 & \lesssim_{\mathbb{P}}\frac{r_{n}}{n} &  & \text{Lemma \ref{lem:w-bounds}.\ref{enu:w|ddot} + Markov's Ineq.}
\end{align*}
Combining these results,
\begin{align*}
\mathbb{E}\left[\left\{ \bar{\ddot{\theta}}_{\text{oracle}}-\psi(\ddot{C})\right\} ^{2}|\ddot{C},\bar{\mathbf{C}}\right] & =Var\left[\bar{\ddot{\theta}}_{\text{oracle}}|\ddot{C},\bar{\mathbf{C}}\right]+\left\{ \mathbb{E}\left(\bar{\ddot{\theta}}_{\text{oracle}}|\ddot{C},\bar{\mathbf{C}}\right)-\psi(\ddot{C})\right\} ^{2}\\
 & \lesssim_{\mathbb{P}}\frac{r_{n}}{n}+r_{n}^{-2s_{\psi}/d_{C}}.
\end{align*}
Thus, $\mathbb{E}\left[\left\{ \bar{\ddot{\theta}}_{\text{oracle}}-\psi(\ddot{C})\right\} ^{2}|\ddot{C},\bar{\mathbf{C}}\right]^{1/2}\lesssim_{\mathbb{P}}\sqrt{\frac{r_{n}}{n}}+r_{n}^{-s_{\psi}/d_{C}}.$
Markov's inequality completes the proof (Lemma \ref{lem:conditional-markovs}). 

\section{\label{sec:Proof-of-Corollary}Proof of Corollary \ref{cor:(Conditional-Average-Treatment}}
\begin{proof}
Let $\bar{\ddot{\theta}}_{\text{oracle},\text{trt}}$ and $\bar{\ddot{\theta}}_{\text{oracle},\text{ctrl}}$
be the oracle estimates of $\ddot{\theta}_{\text{trt}}$ and $\ddot{\theta}_{\text{ctrl}}$,
and let $f_{0,\text{trt}}$ and $f_{0,\text{ctrl}}$ be the corresponding
oracle pseudo-outcome functions, where $\mathbb{E}\left(f_{0,\text{trl}}(Z_{i})-f_{0,\text{ctrl}}(Z_{i})|C=c\right)=\tau(c)$.
From Theorem \ref{thm:(Error-relative-to-oracle} we have
\begin{align*}
\left(\hat{\tilde{\bar{\ddot{\theta}}}}_{\text{trt}}-\hat{\tilde{\bar{\ddot{\theta}}}}_{\text{ctrl}}\right)-\left(\bar{\ddot{\theta}}_{\text{oracle},\text{trt}}-\bar{\ddot{\theta}}_{\text{oracle},\text{ctrl}}\right) & =\left(\hat{\tilde{\bar{\ddot{\theta}}}}_{\text{trt}}-\bar{\ddot{\theta}}_{\text{oracle},\text{trt}}\right)-\left(\hat{\tilde{\bar{\ddot{\theta}}}}_{\text{ctrl}}-\bar{\ddot{\theta}}_{\text{oracle},\text{ctrl}}\right)\\
 & \lesssim_{\mathbb{P}}k_{n}^{-\left(s_{\gamma}+s_{\alpha}\right)/d_{X}}+\sqrt{\frac{r_{n}}{n}}+\frac{k_{n}^{1/2-s_{\gamma}/d_{X}}}{\sqrt{n}}.
\end{align*}
Following the same logic as in Theorem \ref{thm:(Oracle-Error)-Under},
the oracle error satisfies 
\begin{align*}
\left(\bar{\ddot{\theta}}_{\text{oracle},\text{trt}}-\bar{\ddot{\theta}}_{\text{oracle},\text{ctrl}}\right)-\tau(\ddot{C}) & =\frac{1}{n}\sum_{i=1}^{n}\bar{\ddot{w}}(\bar{C}_{i})\left\{ f_{0,\text{trl}}(\bar{Z}_{i})-f_{0,\text{ctrl}}(\bar{Z}_{i})\right\} -\tau(\ddot{C})\\
 & \lesssim_{\mathbb{P}}\sqrt{\frac{r_{n}}{n}}+r_{n}^{-s_{\tau}/d_{C}}.
\end{align*}
Thus, 
\[
\left(\hat{\tilde{\bar{\ddot{\theta}}}}_{\text{trt}}-\hat{\tilde{\bar{\ddot{\theta}}}}_{\text{ctrl}}\right)-\tau(\ddot{C})\lesssim_{\mathbb{P}}k_{n}^{-\left(s_{\gamma}+s_{\alpha}\right)/d_{X}}+\sqrt{\frac{r_{n}}{n}}+\frac{k_{n}^{1/2-s_{\gamma}/d_{X}}}{\sqrt{n}}+r_{n}^{-s_{\tau}/d_{C}}.
\]
\end{proof}

\section{\label{sec:Conditional-covariance-results}Expected conditional covariance
results}

Here we use more detailed notation to keep track of inputs from multiple
samples, as in Figure (\ref{fig:Flowchart-showing-how}). Let $\tilde{\bm{\bar{\pi}}}$,
$\hat{\bar{\bm{\eta}}}$, $\bm{\pi}_{0},$ $\bm{\eta}_{0}$, $\hat{\mathbf{a}}$,
$\mathbf{\hat{y}}$,$\bar{\mathbf{a}}$, and $\bar{\mathbf{y}}$ be
the $n$-length vectors with $i^{th}$ elements equal to $\tilde{\pi}(\bar{X}_{i})$,
$\hat{\eta}(\bar{X}_{i})$, $\pi_{0}(\bar{X}_{i})$, $\eta_{0}(\bar{X}_{i})$,
$\hat{A}_{i}$, $\hat{Y_{i}}$, $\bar{A}_{i},$and $\bar{Y}_{i}$
respectively. Let \textbf{$\mathbf{\hat{Q}},\mathbf{\tilde{Q}}$},\textbf{
}and $\mathbf{\bar{Q}}$ be the $n\times k_{n}$ matrices with $i^{th}$
row equal to $q(\hat{X}_{i})$, $q(\tilde{X}_{i})$, and $q(\bar{X}_{i})$
respectively. Below, we assume there exists constants $l$ and $u$
such that $0<l\leq Cov\left(A,Y|X\right)\leq u$.

Without cross-fitting, $\mathbf{\hat{Q}}=\mathbf{\tilde{Q}}=\mathbf{\bar{Q}}$.
Let $\bar{\mathbf{H}}=\mathbf{\bar{Q}}(\mathbf{\bar{Q}}^{\top}\mathbf{\bar{Q}})^{-1}\mathbf{\bar{Q}}^{\top}$.
We have 
\begin{align*}
\frac{1}{n}\mathbb{E}\left[\text{tr}\left\{ Cov\left(\tilde{\bm{\bar{\pi}}},\hat{\bar{\bm{\eta}}}|\mathbf{X}\right)\right\} \right] & =\frac{1}{n}\mathbb{E}\left[\text{tr}\left\{ Cov\left(\bar{\mathbf{H}}\bar{\mathbf{a}},\bar{\mathbf{H}}\bar{\mathbf{y}}|\mathbf{X}\right)\right\} \right]\\
 & =\frac{1}{n}\mathbb{E}\left[\text{tr}\left\{ \mathbf{\bar{\mathbf{H}}}Cov\left(\bar{\mathbf{a}},\bar{\mathbf{y}}|\mathbf{X}\right)\bar{\mathbf{H}}\right\} \right]\\
 & =\frac{1}{n}\mathbb{E}\left[\text{tr}\left\{ \mathbf{\bar{\mathbf{H}}}Cov\left(\bar{\mathbf{a}},\bar{\mathbf{y}}|\mathbf{X}\right)\right\} \right]\\
 & =\frac{1}{n}\mathbb{E}\left[\sum_{i}\mathbf{\bar{\mathbf{H}}}_{i,i}Cov\left(\bar{A}_{i},\bar{Y}_{i}|\bar{X}_{i}\right)\right]\\
 & \asymp\frac{1}{n}\mathbb{E}\left[\sum_{i}\mathbf{\bar{\mathbf{H}}}_{i,i}\right]\\
 & =\frac{1}{n}\mathbb{E}\left[tr\left(\mathbf{\bar{\mathbf{H}}}\right)\right]\\
 & =k_{n}/n.
\end{align*}
The same reasoning shows that, 
\[
\frac{1}{n}\mathbb{E}\left[\text{tr}\left\{ Cov\left(\tilde{\bm{\bar{\pi}}},\bar{\mathbf{y}}|\mathbf{X}\right)\right\} \right]=\frac{1}{n}\mathbb{E}\left[\text{tr}\left\{ \mathbf{\bar{\mathbf{H}}}Cov\left(\bar{\mathbf{a}},\bar{\mathbf{y}}|\mathbf{X}\right)\right\} \right]\asymp k_{n}/n
\]
and $\frac{1}{n}\mathbb{E}\left[\text{tr}\left\{ Cov\left(\bm{\bar{a}},\hat{\bar{\mathbf{\bm{\eta}}}}|\mathbf{X}\right)\right\} \right]\asymp k_{n}/n$.
Thus, 
\[
\frac{1}{n}\text{tr}\mathbb{E}\left[Cov\left(\bar{\mathbf{a}}-\tilde{\bm{\bar{\pi}}},\bar{\mathbf{y}}-\hat{\bm{\bar{\eta}}}|\mathbf{X}\right)\right]-\mathbb{E}\left[Cov\left(A,Y|X\right)\right]\asymp k_{n}/n.
\]

For 2-way cross-fitting, $\mathbf{\hat{Q}}=\mathbf{\tilde{Q}}\neq\mathbf{\bar{Q}}$.
Let $\hat{\bar{\mathbf{N}}}=\mathbf{\bar{Q}}(\mathbf{\hat{Q}}^{\top}\mathbf{\hat{Q}})^{-1}\mathbf{\hat{Q}}^{\top}$.
In this case,

\begin{align*}
\frac{1}{n}\mathbb{E}\left[\text{tr}\left\{ Cov\left(\tilde{\bm{\bar{\pi}}},\hat{\bar{\bm{\eta}}}|\mathbf{X}\right)\right\} \right] & =\frac{1}{n}\mathbb{E}\left[\text{tr}\left\{ Cov\left(\hat{\bar{\mathbf{N}}}\mathbf{a},\hat{\bar{\mathbf{N}}}\mathbf{y}|\mathbf{X}\right)\right\} \right]\\
 & =\frac{1}{n}\mathbb{E}\left[\text{tr}\left\{ \hat{\bar{\mathbf{N}}}Cov\left(\mathbf{a},\mathbf{y}|\mathbf{X}\right)\hat{\bar{\mathbf{N}}}^{\top}\right\} \right]\\
 & =\frac{1}{n}\mathbb{E}\left[\text{tr}\left\{ \hat{\bar{\mathbf{N}}}^{\top}\hat{\bar{\mathbf{N}}}Cov\left(\mathbf{a},\mathbf{y}|\mathbf{X}\right)\right\} \right]\\
 & =\frac{1}{n}\mathbb{E}\left[\sum_{i}\left(\hat{\bar{\mathbf{N}}}^{\top}\hat{\bar{\mathbf{N}}}\right)_{i,i}Cov\left(\bar{A}_{i},\bar{Y}_{i}|\bar{X}_{i}\right)\right]\\
 & \asymp\frac{1}{n}\mathbb{E}\left[tr\left(\hat{\bar{\mathbf{N}}}^{\top}\hat{\bar{\mathbf{N}}}\right)\right]\\
 & =\frac{1}{n}\mathbb{E}\left[tr\left(\mathbf{\hat{Q}}(\mathbf{\hat{Q}}^{\top}\mathbf{\hat{Q}})^{-1}\mathbf{\bar{Q}}^{\top}\mathbf{\bar{Q}}(\mathbf{\hat{Q}}^{\top}\mathbf{\hat{Q}})^{-1}\mathbf{\hat{Q}}^{\top}\right)\right]\\
 & =\frac{1}{n}tr\left[\mathbb{E}\left((\mathbf{\hat{Q}}^{\top}\mathbf{\hat{Q}})^{-1}\mathbf{\bar{Q}}^{\top}\mathbf{\bar{Q}}\right)\right]\\
 & =\frac{1}{n}tr\left[\mathbb{E}\left((\mathbf{\hat{Q}}^{\top}\mathbf{\hat{Q}})^{-1}\right)\mathbb{E}\left(\mathbf{\bar{Q}}^{\top}\mathbf{\bar{Q}}\right)\right].
\end{align*}
To study the last line, we apply a result from \citet{Groves1969-bt}.
Let $\hat{\mathbf{A}}=\mathbf{\hat{Q}}^{\top}\mathbf{\hat{Q}}$. \citeauthor{Groves1969-bt}
show that $E(\hat{\mathbf{A}}^{-1})-E(\hat{\mathbf{A}})^{-1}$ is
p.s.d. as long as both expectations exist. Thus,
\begin{align*}
tr\left[\mathbb{E}\left((\mathbf{\hat{Q}}^{\top}\mathbf{\hat{Q}})^{-1}\right)\mathbb{E}\left(\mathbf{\bar{Q}}^{\top}\mathbf{\bar{Q}}\right)\right] & =tr\left\{ \mathbb{E}\left[\mathbf{A}^{-1}\right]\mathbb{E}\left[\mathbf{A}\right]\right\} \\
 & =tr\left\{ \left(\mathbb{E}\left[\mathbf{A}^{-1}\right]-\mathbb{E}\left[\mathbf{A}\right]^{-1}+\mathbb{E}\left[\mathbf{A}\right]^{-1}\right)\mathbb{E}\left[\mathbf{A}\right]\right\} \\
 & =tr\left\{ \left(\mathbb{E}\left[\mathbf{A}^{-1}\right]-\mathbb{E}\left[\mathbf{A}\right]^{-1}\right)\mathbb{E}\left[\mathbf{A}\right]\right\} +k_{n}\\
 & =tr\left\{ \mathbb{E}\left[\mathbf{A}\right]^{1/2}\left(\mathbb{E}\left[\mathbf{A}^{-1}\right]-\mathbb{E}\left[\mathbf{A}\right]^{-1}\right)\mathbb{E}\left[\mathbf{A}\right]^{1/2}\right\} +k_{n}\\
 & \geq k_{n},
\end{align*}
where the last line uses the fact that $E(\mathbf{A}^{-1})-E(\mathbf{A})^{-1}$
is p.s.d. Thus, for 2-way CF, 
\[
\frac{1}{n}\mathbb{E}\left[\text{tr}\left\{ Cov\left(\tilde{\bm{\bar{\pi}}},\hat{\bar{\bm{\eta}}}|\mathbf{X}\right)\right\} \right]\gtrsim k_{n}/n.
\]

\textcolor{teal}{}

\end{document}